\pdfoutput=1 
\documentclass[12pt,one column]{IEEEtran}

\usepackage{setspace}
\usepackage{amsmath}
\usepackage{amssymb}
\usepackage{cite}
\usepackage{color}
\usepackage{epsfig}
\usepackage{epsf}
\usepackage{rotating}
\usepackage{mathrsfs}
\usepackage{epsfig}
\usepackage{graphics}
\usepackage{bbold}
\usepackage{theorem}
\pdfoutput=1

\usepackage{amsmath}
\usepackage{amssymb}
\usepackage{cite}
\usepackage{color}
\usepackage{epsfig}
\usepackage{epsf}
\usepackage{rotating}
\usepackage{mathrsfs}
\usepackage{epsfig}
\usepackage{graphics}
\usepackage{theorem}

\newtheorem{thm}{Theorem}
\newtheorem{lem}{Lemma}
\newtheorem{coro}{Corollary}

\newtheorem{proposition}{Proposition}

\begin{document}

\pagenumbering{arabic} \setcounter{page}{0}

\thispagestyle{empty} \vskip 1cm


\title{{On the Design of Signature Codes  in\\ Decentralized Wireless Networks}}
\author{ Kamyar Moshksar and Amir K. Khandani \\
\small Coding \& Signal Transmission Laboratory (www.cst.uwaterloo.ca)\\
Dept. of Elec. and Comp. Eng., University of Waterloo\\ Waterloo, ON, Canada, N2L 3G1 \\
Tel: 519-725-7338, Fax: 519-888-4338\\e-mail: \{kmoshksa,
khandani\}@cst.uwaterloo.ca} \maketitle

\begin{abstract}
This paper addresses a unified approach towards communication in decentralized wireless networks of separate transmitter-receiver pairs. Different transmitters are connected to different receivers through channels with  static and non-frequency selective gains. In general, users are unaware of each other's codebooks and there is no central controller to assign the resources in the network  to the users.  A randomized signaling scheme is introduced in which each user locally spreads its Gaussian signal along a  randomly generated spreading code comprised of a sequence of nonzero elements over a certain alphabet. Along with spreading, each transmitter also masks its output independently from transmission to transmission. Using a conditional version of entropy power inequality and a key lemma on the differential entropy of mixed Gaussian random vectors, achievable rates are developed for the users. Assuming the channel gains are realization of independent continuous random variables, each user finds the optimum parameters  in constructing the randomized spreading and masking sequences by maximizing the average achievable rate per user. It is seen that  as the number of users increases, the achievable Sum Multiplexing Gain of the network approaches that of a centralized orthogonal scheme where multiuser interference is completely avoided.  An interesting observation is that in general the elements of a spreading code are not equiprobable over the underlying alphabet. This is in contrast to the customary use of binary  PN codes in spread spectrum communications in which the code elements may be selected with equal probability over $\{-1,1\}$. This particularly happens if the number of active users is greater than three. Finally, using the recently developed extremal inequality of Liu-Viswanath, we present an optimality result showing that transmission of Gaussian signals via spreading and masking yields higher achievable rates than the maximum achievable rate attained by applying masking only.  
\end{abstract}

\section{Introduction}

An important topic in  modern wireless communications is the subject
of decentralized networks. By definition, a decentralized network of separate transmitter-receiver pairs
has no central controller to allocate the network resources among
the active users. As such, resource allocation must be performed
locally at each node. In general, users are not already aware of the number of active users and the channel gains\footnote{Throughout the paper, we assume the channel from each transmitter to each receiver is modeled by a static and non-frequency selective gain.}. Also, users are not aware of each
other's codebooks implying multiuser detection is not possible, i.e., users treat each other as noise.
Multiuser interference is known to be the main factor limiting the
achievable rates in such networks particularly in the high Signal-to-Noise Ratio (SNR) or interference limited regime. Therefore, all users must follow
a distributed signaling scheme such that the destructive effect of
 interference on each user is minimized, while the resources are
fairly shared among users.

Most of distributed schemes reported in the literature rely on
either \textit{game-theoretic} approaches or \textit{cognitive
radios}. Cognitive radios \cite{2,haykin} have the ability to sense the
unoccupied portion of the available spectrum and use this
information in resource allocation. Although such smart radios avoid the use of a central controller,
they require sophisticated detection techniques for sensing the
spectrum holes and dynamic frequency assignment which add to the
overall system complexity \cite{8,9,10}. 

Distributed strategies based on game theoretic arguments have already attracted a great deal of attention. In \cite{1}, the authors introduce a non-cooperative game theoretic framework to investigate the spectral efficiency issue when several users compete over an unlicensed band with no central controller. Reference \cite{12} offers a brief overview of game theoretic dynamic spectrum sharing. Although these schemes enable us to understand the dynamics of distributed resource allocation, they usually suffer from complexity in software and convergence issues as they rely on iterative algorithms. 

Spread spectrum communications is a natural setup to share the same bandwidth by several users.  This area has attracted tremendous attention by different authors during the past  decades in the context of centralized uplink/downlink multiuser systems. Appealing characteristics of spread spectrum systems have motivated researchers to utilize these schemes in networks without a certain infrastructure, i.e., packet radio or ad-hoc networks\cite{15}. In direct sequence spread spectrum systems, the signal of each user is spread using a pseudo-random noise (PN) code. The challenging point is that in a network without a central controller, if two users use the same spreading code, they will not be capable of recovering the data at the receiver side due to the high amount of interference.  Distributed code assignment techniques are developed in \cite{38,39}. In \cite{38}, using a greedy approximation algorithm and invoking graph theory, a distributed code assignment protocol is suggested. Another category of research is devoted to devise distributed schemes in the reverse link (uplink) of cellular systems. Distributed power assignments algorithms are proposed in \cite{41,42}. Reference \cite{43} proposes a distributed scheduling method called the token-bucket on-off scenario utilized by autonomous mobile stations where its impact on the overall throughput of the reverse link is investigated. Furthermore, decentralized rate assignments in a multi-sector code division multiple access wireless network are discussed in \cite{44}.

Being a standard technique in spread spectrum communications and due
to its interference avoidance nature, Frequency Hopping is the
simplest spectrum sharing method to use in decentralized networks.
As different users typically have no prior information about the
codebooks of the other users, the most efficient method is avoiding
interference by choosing unused channels. As mentioned earlier,
searching the spectrum to find spectrum holes is not an easy task
due to the dynamic spectrum usage. As such, FH is a realization of a
transmission scheme without sensing, while avoiding the collisions
as much as possible.  Frequency hopping is one of the standard
signaling schemes  adopted in ad-hoc networks. In short
range scenarios, bluetooth systems \cite{19,20,21} are the most
popular examples of a wireless personal area network or WPAN. Using
FH over the unlicensed ISM band, a bluetooth system provides robust
communication to unpredictable sources of interference. A
modification of Frequency Hopping called Dynamic Frequency Hopping
(DFH), selects the hopping pattern based on interference
measurements in order to avoid dominant interferers. The performance
of a DFH scheme when applied to a cellular system is assessed in
\cite{22,23,24}.

Distributed rate assignment strategies are recently adopted in the context of  medium access control. It is well-known \cite{53} that the capacity region $\mathfrak{R}$ of  a multiple access channel with $n$ users is a polytope with a $n!$ corner points. Let each corner point of $\mathfrak{R}$ be an $n$-tuple whose elements are among the numbers $R_{1},\cdots,R_{L-1}$ and $R_{L}$. With no cooperation among the users, the authors in \cite{HO} propose that each user selects a codebook of rate $R_{l}$ with probability $p_{l}\in(0,1)$ for $1\leq l\leq L$. Assuming the receiver is aware of the rate selection of all users, the average sum rate of the network is $\bar{R}=\sum_{l_{1},\cdots,l_{n}\in\{1,\cdots,L\}}p_{l_{1}}\cdots p_{l_{n}}(R_{l_{1}}+\cdots+R_{l_{n}})\mathbb{1}_{(R_{l_{1}},\cdots,R_{l_{n}})\in\mathfrak{R}}$ where $\mathbb{1}_{(R_{l_{1}},\cdots,R_{l_{n}})\in\mathfrak{R}}$ is $1$ if $(R_{l_{1}},\cdots,R_{l_{n}})\in\mathfrak{R}$ and $0$ otherwise.  Finally, the numbers $p_{1},\cdots, p_{L-1}$ and $p_{L}$ are derived to maximize $\bar{R}$. Major differences of this scenario with a decentralized wireless network are

\textit{1)} The capacity region of a multiuser interference channel is unknown. 

\textit{2)}  In case transmitters have different choices to select the transmission rate, a certain receiver is not guaranteed to be aware of the transmission rate of interferers. 

\textit{3)} Any user is already unaware of the gains of channels connecting the interferers' transmitters to its receiver.  Also, any user is never capable of finding the amount of interference it imposes on other users. 

It is well-known that in the low SNR regime continuous transmission of $\mathrm{i.i.d.}$ Gaussian signals is optimal. However, as SNR increases, this scheme turns out to be quite inefficient. For instance, the achievable rate of each user eventually saturates, i.e., the achieved Sum Multiplexing Gain\footnote{The Sum Multiplexing Gain represents the scaling of the sum rate in terms of $\log\mathrm{SNR}$ as SNR tends to infinity.} (SMG) is equal to zero. Using the results in \cite{kami-1}, it is easy to see that by using a masking strategy where each user quits transmitting its Gaussian signals independently from transmission to transmission, a nonzero SMG of $\left(1-\frac{1}{n}\right)^{n-1}$ is attained in a decentralized network of $n$ users. This is an interesting result in the sense that if the number of active users tends to infinity, the achieved SMG settles on $\frac{1}{e}>0$. 

In the present paper, we answer the following questions: 

\textit{Question 1-} Is it possible to achieve an SMG larger than $\frac{1}{e}$ as the number of users becomes large?

We propose a distributed signaling scheme where each user spread its Gaussian signal along a spreading code consisting of $\mathrm{i.i.d.}$ elements selected according a globally known Probability Mass Function (PMF) over a finite alphabet $\mathscr{A}$. Thereafter, the resulting sequence is punctured independently from symbol to symbol with a certain probability representing the masking operation. For example, assuming $\mathscr{A}=\{-1,1\}$, let the generated spreading code have length $10$ and be given by
\begin{equation}
\label{ }
\big(1,1,-1,-1,-1,1,-1,1,1,-1\big).
\end{equation}  
Also, an $\mathrm{i.i.d.}$ sequence of $1$'s (representing $\mathsf{TRANSMIT}$) and $0$'s (representing $\mathsf{MASK}$) with length $10$ is generated as
\begin{equation}
\label{ }
\big(0,1,1,0,0,1,1,1,0,0\big).
\end{equation}
Finally, denoting the Gaussian signal to be transmitted by $\boldsymbol{x}$, the sequence
\begin{equation}
\label{ }
\big(0,\boldsymbol{x},-\boldsymbol{x},0,0,\boldsymbol{x},-\boldsymbol{x},\boldsymbol{x},0,0\big).
\end{equation}
is transmitted in $10$ consecutive transmission slots called a transmission frame. This process is repeated independently from transmission frame to transmission frame. We notice that since different users are not aware of each other's signals and the spreading/masking sequences, the noise plus interference vector at the receive side of any user is a mixed Gaussian random vector. We assume the knowledge of interference Probability Density Function (PDF) at the receiver side of each user. We are able to see that using the proposed randomized spreading scheme, the number of active users and the gains of channels conveying the interferers' signals can be easily found by inspecting the interference PDF and solving a set of linear equation.

Assuming all users are \emph{frame-synchronous}, we derive achievable rates for the users in three steps:

\textit{Step 1-} Using Singular Value Decomposition (SVD) of the signal space at the receiver side any user, the interference vector is mapped in the signal space and the complement space\footnote{In any Euclidean space $\mathscr{E}$ and a subspace $\mathscr{U}$ of $\mathscr{E}$,  the complement space $\mathscr{U}^{\perp}$ of $\mathscr{U}$ is the set of elements in $\mathscr{E}$ that are perpendicular to any element in $\mathscr{U}$. } of the signal space.

\textit{Step 2-} A conditional version of entropy power inequality is used to derive a lower bound on the mutual information between the input and output of each user along any transmission frame. The conditioning is made over the contents of the interference vector mapped in the complement space of the signal space.

\textit{Step 3-} The resulting lower bound in the previous step highly depends on the differential entropy of mixed Gaussian random vectors. Since there is no closed formula on the differential entropy of a mixed Gaussian vector, a key Lemma is used to find computable bounds on this differential entropy. This leads us to the final formulation of the achievable rate. 

In a decentralized network of $n$ users, we are able to show that by regulating the length of the transmission frame and the probabilistic structure of the spreading/masking sequences, the resulting lower bound scales like $\mathrm{SMG}(n)\log\mathrm{SNR}$ where $\lim_{n\to\infty}\mathsf{SMG}(n)=1$. This is exactly the SMG of a centralized orthogonal resource allocation scheme where multiuser interference is completely avoided.   

Our focus is not particularly on the high SNR regime. In fact, the length of the transmission frame and the probabilistic  parameters of the spreading/masking codes are sensitive to the choice of SNR. Our proposed achievable rate for any user in general depends on the gains of the channels conveying the interference. As mentioned earlier, each user is capable of finding the channel gains, however, if each user attempts to maximize its achievable rate over the length of the transmission frame and other code parameters, different users come up with different choices which results in inconsistency. To circumvent this difficulty, assuming the channel gains are realizations of $\mathrm{i.i.d.}$ continuous random variables, each user selects the code parameters such that the average of achievable rate per user over different realizations of the channel gains is maximized. This leads to a consistent and distributed method to design the best randomization algorithm in constructing the spreading/masking sequences. 

An interesting observation is that even in the simplest scenario where the underlying alphabet is $\{-1,1\}$ and no masking is applied\footnote{This reminds us of direct sequence spread spectrum communications.}, the elements of the spreading codes are not equiprobable over $\{-1,1\}$. For example, our simulation results show that in a network of $n=4$ users at $\mathrm{SNR}=60\mathrm{dB}$, the elements of the spreading code must be selected to be $1$ with a probability of $0.01$ and $-1$ with a probability of $0.91$ or vice versa.

\textit{Question 2-} What is the highest achievable rate under the masking protocol? Can one do better than masking?

One may raise the question if masking the transmitted signals independently from transmission slot to transmission slot is by itself sufficient, i.e., by selecting the PDF of the transmitted signals properly (probably non-Gaussian), there is no need for spreading. Using an extremal inequality of Liu-Viswanath \cite{LV}, we are able to show that transmission of Gaussian signals along with spreading and masking yields higher achievable rates that the largest achievable rate with masking alone. 

The rest of the paper is organized as follows. Section II offers the system model. In this section, we introduce the randomized spreading coding and discuss how all user can consistently design their spreading/masking sequences. Section III presents the development of achievable rates based on the three steps mentioned earlier. System design is brought in section IV where we offer several design examples. Finally, section V prove the supremacy of blending spreading and masking over masking alone. Conclusion remarks are given in section VI.    



\textbf{Notation-} Throughout the paper, we denote random quantities in bold case such as $\boldsymbol{x}$ and $\vec{\boldsymbol{y}}$. A realization of $\boldsymbol{x}$ is denoted by $x$. A circularly symmetric complex Gaussian random vector $\vec{\boldsymbol{x}}$ of length $m$ with zero mean and  covariance matrix $C$ is denoted by $\mathcal{CN}(0,C)$. A Bernoulli random variable $\boldsymbol{x}\in\{0,1\}$ with $\Pr\{\boldsymbol{x}=1\}=a\in[0,1]$ is denoted by $\mathrm{Ber}(a)$. For a sequence $(a_{l})_{l=1}^{m}\triangleq(a_{1},\cdots,a_{m})$ and a set $\Xi=\{\xi_{1},\cdots,\xi_{m'}\}\subset\{1,\cdots,m\}$ where $\xi_{1}<\cdots<\xi_{m'}$, we define $(a_{l})_{l\in\Xi}\triangleq (a_{\xi_{1}},\cdots,a_{\xi_{m'}})$. We use $\mathrm{E} \{.\}$ for the expectation operator, $\mathrm{Pr}\{\mathcal{E}\}$ for the probability of an event $\mathcal{E}$, $\mathbb{1}_{\mathcal{E}}$  for the indicator function of an event $\mathcal{E}$ and $p_{\boldsymbol{x}}(.)$  for the PDF of a random variable $\boldsymbol{x}$. Also, $\mathrm{I} (\boldsymbol{x};\boldsymbol{y})$ denotes the mutual information between random variables $\boldsymbol{x}$ and $\boldsymbol{y}$, $\mathrm{h} (\boldsymbol{x})$ the differential entropy of a continuous random variable $\boldsymbol{x}$, $\mathrm{H}(\boldsymbol{x})$ the entropy of a discrete random variable $\boldsymbol{x}$, and the binary entropy function is denoted by $\mathscr{H}(x)\triangleq-x\log x-(1-x)\log(1-x)$ for $x\in[0,1]$. For any $x\in[0,1]$, $\bar{x}$ denotes $1-x$. The Dirac delta function is denoted by $\delta(.)$. For integers $m,n\in\mathbb{N}$, a $m\times n$ matrix in which all elements are $0$ or $1$ is shown by $0_{m\times n}$ or $1_{m\times n}$ respectively. For sets $A$ and $B$, the set $A\backslash B$ denotes a set with elements in $A$ and not in $B$. The cardinality of a set $A$ is denoted by $|A|$. For any two vectors of the same size $\vec{x}$ and $\vec{y}$, the vector $\vec{x}\odot\vec{y}$ is the element-wise product of $\vec{x}$ and $\vec{y}$. For two function $f(\gamma)$ and $g(\gamma)$ of a variable $\gamma>0$, we write $f\sim g$ if $\lim_{\gamma\to\infty}\frac{f}{\log\gamma}=\lim_{\gamma\to\infty}\frac{g}{\log\gamma}$ and $f\lesssim g$ if $\lim_{\gamma\to\infty}\frac{f}{\log\gamma}\leq \lim_{\gamma\to\infty}\frac{g}{\log\gamma}$. The notation $f\gtrsim g$ is defined similarly.

\section{System Model} 
 We consider a decentralized communication network of $n$ users\footnote{Users consists of a separate transmitter-receiver pairs.}. The static and non frequency-selective gain of the channel from the $i^{th}$ transmitter to the $j^{th}$ receiver is shown by $h_{i,j}$ which is in general a complex number. In a decentralized network, there is no communication or cooperation among different users. Due to the fact that the network has no fixed infrastructure and there is no central controller to manage the network resources among users, resource allocation and rate assignment must be performed locally at every transmitter. A main feature of such networks is that the $i^{th}$ user is not already informed about the channel gains $(h_{j,i})_{j=1}^{n}$ concerning the links connecting different transmitters to the $i^{th}$ receiver. In fact, every receiver has only access to the interference PDF and  the knowledge about the number of active users and the channel gains $(h_{j,i})_{j=1}^{n}$ can only be inferred through analyzing this PDF. Also, different users are not aware of each other's codebooks. As such, no multiuser detection is possible and users treat the interference as noise.    
 \subsection{Randomized Signature Codes}

 In this part, we introduce a distributed signaling strategy using randomized spreading/masking. For positive integers $T$ and $K$, the codebook of the $i^{th}$ user consists of $2^{TR_{i}}$ codewords where a typical codeword $(\boldsymbol{x}_{i,t})_{t=1}^{T}$ is a sequence of $\mathrm{i.i.d.}$ circularly symmetric complex Gaussian random variables with zero mean and variance $\gamma$.  The $i^{th}$ user transmits $(\boldsymbol{x}_{i,t})_{t=1}^{T}$ in $T$ \emph{transmission frames} where each transmission frame consists of $K$ \emph{transmission slots}. In a typical transmission frame, one of the signals in the codeword $(\boldsymbol{x}_{i,t})_{t=1}^{T}$ is transmitted. To transmit $\boldsymbol{x}_{i,t}$, the $i^{th}$ user randomly constructs two independent sequences called the spreading and the masking codes. The spreading code is  a $K\times 1$ vector $\vec{\boldsymbol{\mathfrak{s}}}_{i,t}$ over an alphabet $\mathscr{A}\subset\mathbb{Z}\backslash\{0\}$ where the elements of $\vec{\boldsymbol{\mathfrak{s}}}_{i,t}$ are $\mathrm{i.i.d.}$ with a globally known PMF $(\mathsf{p}_{a})_{a\in\mathscr{A}}$. The masking code is a $K\times 1$ vector $\vec{\boldsymbol{\mathfrak{m}}}_{i,t}$ whose elements are independent $\mathrm{Ber}(\varepsilon)$ random variables for some $\varepsilon\in(0,1]$. Thereafter, the $i^{th}$ user transmits $\boldsymbol{x}_{i,t}\vec{\boldsymbol{\mathfrak{s}}}_{i,t}\odot\vec{\boldsymbol{\mathfrak{m}}}_{i,t}$ in the $t^{th}$ transmission frame. The vector 
 \begin{equation}\vec{\boldsymbol{s}}_{i,t}\triangleq\vec{\boldsymbol{\mathfrak{s}}}_{i,t}\odot\vec{\boldsymbol{\mathfrak{m}}}_{i,t}
 \end{equation}
  is called the randomized signature code of the $i^{th}$ user in the $t^{th}$ transmission frame. We remark that the spreading and masking codes of the $i^{th}$ user over different transmission frames are constructed independently. The alphabet $\mathscr{A}$ has the property that for any $a\in\mathscr{A}$, we have $-a\in\mathscr{A}$. The received vector at the receiver side of the $i^{th}$ user in a typical transmission frame is given by\footnote{We  omit the frame index for notation simplicity.}
\begin{equation}
\label{e1}
\vec{\boldsymbol{y}}_{i}=\beta h_{i,i}\boldsymbol{x}_{i}\vec{\boldsymbol{s}}_{i}+\sum_{j\neq i}\beta h_{j,i}\boldsymbol{x}_{j}\vec{\boldsymbol{s}}_{j}+\vec{\boldsymbol{z}}_{i}
\end{equation}
 where $\vec{\boldsymbol{z}}_{i}$ is a $\mathcal{CN}(0_{K\times 1},I_{K})$ random vector representing the ambient noise at the $i^{th}$ receiver. Also, $\beta$ is a normalization factor ensuring the average transmission power per symbol of the $i^{th}$ user is $\gamma$, i.e.,\begin{equation}
\label{e11}
\beta^{2}\mathrm{E}\left\{\|\vec{\boldsymbol{s}}_{i}\|_{2}^{2}\right\}=1.
\end{equation}
In (\ref{e1}), we have made the assumption that all active users in the network are frame-synchronous meaning their transmission frames start and end at similar time instants. This is not necessarily a valid assumption in a decentralized network, however, this makes the presentation of the subject much easier. It is clear that the transmitted signals of each user along its transmission frames are correlated while signals transmitted in different transmission frames are independent. Hence, we assume any new active user is capable of detecting the correlated segments along the interference plus noise process, and therefore, synchronizing itself with former active users in the network. However, in case different users are not frame-synchronous and users are not aware of the asynchrony pattern, the communication channel of any user is not ergodic anymore and one must perform outage analysis. 

Using joint typicality at the receiver side of the $i^{th}$ user, any data rate $R_{i}\leq \mathsf{C}_{i}$ is achievable where
\begin{equation}
\mathsf{C}_{i}\triangleq \frac{\mathrm{I}(\boldsymbol{x}_{i},\vec{\boldsymbol{s}}_{i};\vec{\boldsymbol{y}}_{i})}{K}.
\end{equation} 
 The term $\mathrm{I}(\boldsymbol{x}_{i},\boldsymbol{s}_{i};\vec{\boldsymbol{y}}_{i})$ indicates that the $i^{th}$ user is also embedding information in the sequence of $\mathrm{i.i.d.}$ signature codes. In fact, one can assume the codeword of any user consists of two sequences, namely, the sequence of Gaussian signals and the sequence of randomized signature codes. Due to the fact that the signature code of any user is not known to other users and on the other hand, the signature codes are independently changing over different transmission frames, the noise plus interference at the receiver side of any user has a mixed Gaussian PDF. This makes $\mathrm{I}(\boldsymbol{x}_{i},\vec{\boldsymbol{s}}_{i};\vec{\boldsymbol{y}}_{i})$ have no closed expression. Therefore, we need to obtain a tight lower bound on this quantity whose computation only needs data that can be inferred from the noise plus interference PDF at the receiver side of the $i^{th}$ user and be fed back to its associated transmitter in order to regulate the transmission rate. Throughout the paper, the interference term at the receiver side of the $i^{th}$ user is denoted by $\vec{\boldsymbol{w}}_{i}$, i.e., $\vec{\boldsymbol{w}}_{i}=\sum_{j\neq i}\beta h_{j,i}\boldsymbol{x}_{j}\vec{\boldsymbol{s}}_{j}$. One can state $\vec{\boldsymbol{w}}_{i}$ as
\begin{equation}
\label{ }
\vec{\boldsymbol{w}}_{i}=\boldsymbol{S}_{i}\Xi_{i}\vec{\boldsymbol{X}}_{i}
\end{equation} 
where
\begin{equation}
\label{ }
\boldsymbol{S}_{i}\triangleq\begin{pmatrix}
    \vec{\boldsymbol{s}}_{1}  & \cdots&\vec{\boldsymbol{s}}_{i-1}&\vec{\boldsymbol{s}}_{i+1}&\cdots&\vec{\boldsymbol{s}}_{n}   
\end{pmatrix},
\end{equation}
\begin{equation}
\label{ }
\Xi_{i}\triangleq \mathrm{diag}(h_{1,i},\cdots,h_{i-1,i},h_{i+1,i},\cdots,h_{n,i})
\end{equation}
and
\begin{equation}
\label{ }
\vec{\boldsymbol{X}}_{i}=\begin{pmatrix}
    \boldsymbol{x}_{1}  & \cdots&\boldsymbol{x}_{i-1}&\boldsymbol{x}_{i+1}&\cdots&\boldsymbol{x}_{n}   
\end{pmatrix}^{T}.
\end{equation}

\subsection{Considerations on the Channel Gains and the Number of Active Users  }
 In general, we assume that the $i^{th}$ receiver is aware of $h_{i,i}$ which can be done through a training sequence sent by the $i^{th}$ transmitter. Assuming the channel gains are realizations of $\mathrm{i.i.d.}$ random variables with a continuous PDF, then the number of Gaussian components in the mixed Gaussian PDF of the interference in any transmission slot at the receiver side of the $i^{th}$ user is $\left(\frac{|\mathscr{A}|}{2}\right)^{n-1}$ if masking is not performed and $\left(\frac{|\mathscr{A}|}{2}+1\right)^{n-1}$ if masking and spreading are both applied. These levels consist of $\sum_{j\neq i}a_{j}^{2}|h_{j,i}|^{2}\gamma$ where $a_{j}\in\mathscr{A}$. As such, as far as $|\mathscr{A}|\geq 3$, the number of active users can be obtained by finding the number of interference power levels. However, if $\mathscr{A}=\{-a,a\}$ for some $a\in\mathbb{N}$ and masking is not performed, the interference PDF in any transmission slot is Gaussian (the interference vector on any transmission frame is still mixed Gaussian) with power $a^{2}\gamma\sum_{j\neq i}|h_{j,i}|^{2}$. Therefore, the number of active users can not be derived by investigating the interference PDF in one transmission slot. In this case, it can be verified that the joint PDF of any two transmission slots in a transmission frame is a mixed Gaussian PDF with $2^{n-1}$ Gaussian components. This yields a method to find $n$ in case $\mathscr{A}$ has only two elements. 
 
  By symmetry, characterization of $\mathsf{C}_{i}$ demands the knowledge of an arbitrary reordering of the sequence $(h_{j,i})_{j\neq i}$. In this paper, we derive a lower bound $\mathsf{C}_{i}^{(\mathrm{lb})}$ on  $\mathsf{C}_{i}$ which is only a function of the magnitude of the channel gains. Therefore, we need to obtain an arbitrary  reordering of $(|h_{j,i}|)_{j\neq i}$. Let $(\mathsf{h}^{(i)}_{1},\mathsf{h}^{(i)}_{2},\cdots,\mathsf{h}^{(i)}_{n-1})$ be a reordering of $(h_{j,i})_{j\neq i}$ based on magnitude, i.e., $|\mathsf{h}^{(i)}_{1}|<|\mathsf{h}^{(i)}_{2}|<\cdots<|\mathsf{h}^{(i)}_{n-1}|$. We consider the following cases: 
  
  \textit{Case 1-} If $|\mathscr{A}|\geq 4$, let $a$ and $b$ be the two largest elements in $\mathscr{A}$ such that $a>b$. Denoting the $n-1$ largest interference plus noise power levels on each transmission slot by $\pi_{1}<\cdots<\pi_{n-1}$, we have $\beta^{2}\gamma a^{2}\sum_{j=1}^{n-1}|\mathsf{h}_{j}^{(i)}|^{2}+1=\pi_{n-1}$ and $\beta^{2}\gamma a^{2}\sum_{\substack{j=1\\j\neq l}}^{n-1}|\mathsf{h}_{j}^{(i)}|^{2}+\beta^{2}\gamma b^{2}|\mathsf{h}_{l}^{(i)}|^{2}+1=\pi_{n-1-l}$ for $1\leq l\leq n-2$. These $n-1$ linear equations yield $(|\mathsf{h}_{j}^{(i)}|)_{j=1}^{n-1}$.
  
  \textit{Case 2-} Let masking be the only ingredient in constructing the signatures, i.e., spreading is not applied. Denoting the $n-1$ largest interference plus noise power levels on each transmission slot by $\pi_{1}<\cdots<\pi_{n-1}$, we have $\beta^{2}\gamma \sum_{j=1}^{n-1}|\mathsf{h}_{j}^{(i)}|^{2}+1=\pi_{n-1}$ and $\beta^{2}\gamma \sum_{\substack{j=1\\j\neq l}}^{n-1}|\mathsf{h}_{j}^{(i)}|^{2}+1=\pi_{n-1-l}$ for $1\leq l\leq n-2$. These $n-1$ linear equations yield $(|\mathsf{h}_{j}^{(i)}|)_{j=1}^{n-1}$.
  
  \textit{Case 3-} Let $\mathscr{A}=\{-a,a\}$ for some $a\in\mathbb{R}^{+}$ and masking is performed on top of spreading. Then, we can apply the same procedure in case 2. 
  
  \textit{Case 4-} Let $\mathscr{A}=\{-a,a\}$ for some $a\in\mathbb{R}^{+}$ and masking is not applied. The joint PDF of the interference plus noise on any two transmission slots inside a transmission frame is a bivariate mixed Gaussian PDF in which the Gaussian components have covariance matrices of the form 
  \begin{equation}
\label{ }
\begin{pmatrix}
  \beta^{2}\gamma a^{2}\sum_{j=1}^{n-1}|\mathsf{h}_{j}^{(i)}|^{2}+1    & \beta^{2}\gamma a^{2}\sum_{j=1}^{n-1}c_{j}|\mathsf{h}_{j}^{(i)}|^{2}   \\
   \beta^{2}\gamma a^{2}\sum_{j=1}^{n-1}c_{j}|\mathsf{h}_{j}^{(i)}|^{2}   &    \beta^{2}\gamma a^{2}\sum_{j=1}^{n-1}|\mathsf{h}_{j}^{(i)}|^{2}+1\end{pmatrix}
\end{equation}  
  where $c_{j}\in\{-1,1\}$ for $1\leq j\leq n-1$. The $n-2$ largest elements among the off-diagonal elements of these matrices correspond to $\beta^{2}\gamma a^{2}\sum_{\substack{j=1\\j\neq l}}^{n-1}|\mathsf{h}_{j}^{(i)}|^{2}-\beta^{2}\gamma a^{2}|\mathsf{h}_{l}^{(i)}|^{2}$ for $1\leq l\leq n-2$. These elements together with the diagonal element $  \beta^{2}\gamma a^{2}\sum_{j=1}^{n-1}|\mathsf{h}_{j}^{(i)}|^{2}+1$ yield  $(|\mathsf{h}_{j}^{(i)}|)_{j=1}^{n-1}$.
  
  Therefore, we have shown that the $i^{th}$ user can find $n$ and a reordering of the sequence $(h_{j,i})_{j\neq i}$. 
    
\subsection{ A Global Tool To Design The Randomized Signature Codes }
An important issue in a decentralized network is to propose a globally known utility function to be optimized by all user without any cooperation. As mentioned earlier, the receivers can infer the number of active users in the network and the channel gains by inspecting the interference PDF. We consider a scenario where this information is fed back to the transmitters. As mentioned earlier, there is no closed formulation on $\mathsf{C}_{i}$. However, we are able to develop a lower bound $\mathsf{C}_{i}^{(\mathrm{lb})}$ for $\mathsf{C}_{i}$ which is tight enough to guarantee \begin{equation}\lim_{\gamma\to\infty}\frac{\mathsf{C}_{i}^{(\mathrm{lb})}}{\log\gamma}=\lim_{\gamma\to\infty}\frac{\mathsf{C}_{i}}{\log\gamma}.\end{equation} 
In general, $\mathsf{C}_{i}^{(\mathrm{lb})}$ depends on $\vec{h}_{i}\triangleq(h_{j,i})_{j=1}^{n}$. As such, we denote it explicitly by $\mathsf{C}_{i}^{(\mathrm{lb})}(\vec{h}_{i})$. Assuming $(h_{j,i})_{j=1}^{n}$ are realizations of independent $\mathcal{CN}(0,1)$ random variables $(\boldsymbol{h}_{j,i})_{j=1}^{n}$, we propose that the $i^{th}$ user selects $K$, $(\mathsf{p}_{{a}})_{a\in\mathscr{A}}$ and $\varepsilon$ based on 
\begin{equation}
\label{rule}
(\hat{K},(\hat{\mathsf{p}}_{a})_{a\in\mathscr{A}},\hat{\varepsilon})=\arg\sup_{K,(\mathsf{p}_{a})_{a\in\mathscr{A}},\varepsilon}\mathrm{E}\left\{\mathsf{C}_{i}^{(\mathrm{lb})}(\vec{\boldsymbol{h}}_{i})\right\}.
\end{equation}
 After selecting $K$ and $(\mathsf{p}_{a})_{a\in\mathscr{A}}$ using (\ref{rule}), the $i^{th}$ user regulates its actual transmission rate at $R_{i}=\mathsf{C}_{i}^{\mathrm{(lb)}}(\vec{h}_{i})$ using  the realization of $\vec{\boldsymbol{h}}_{i}=\vec{h}_{i}$.

 \section {A Lower Bound $\mathrm{I}(\boldsymbol{x}_{i},\vec{\boldsymbol{s}}_{i};\vec{\boldsymbol{y}}_{i})$ }
 One can write $\mathrm{I}(\boldsymbol{x}_{i},\vec{\boldsymbol{s}}_{i};\vec{\boldsymbol{y}}_{i})$ as
   \begin{eqnarray}
\label{b1}
\mathrm{I}(\boldsymbol{x}_{i},\vec{\boldsymbol{s}}_{i};\vec{\boldsymbol{y}}_{i})=\mathrm{I}(\vec{\boldsymbol{s}}_{i};\vec{\boldsymbol{y}}_{i})+\mathrm{I}(\boldsymbol{x}_{i};\vec{\boldsymbol{y}}_{i}|\vec{\boldsymbol{s}}_{i})\geq\mathrm{I}(\boldsymbol{x}_{i};\vec{\boldsymbol{y}}_{i}|\vec{\boldsymbol{s}}_{i}).\end{eqnarray}  
The term $\mathrm{I}(\boldsymbol{x}_{i};\vec{\boldsymbol{y}}_{i}|\vec{\boldsymbol{s}}_{i})$ is the achievable rate of the $i^{th}$ user as if this user knew the randomized signature code $\vec{\boldsymbol{s}}_{i}$ already, i.e., the achievable rate of the $i^{th}$ user can be in general larger than the case where the signature matrices are already revealed to the receiver side. The extra term $\mathrm{I}(\vec{\boldsymbol{s}}_{i};\vec{\boldsymbol{y}}_{i})$ is bounded from above by $\mathrm{H}(\vec{\boldsymbol{s}}_{i})$ which is not a function of SNR. Therefore, 
\begin{equation}
\label{e14}
\lim_{\gamma\to\infty}\frac{\mathrm{I}(\boldsymbol{x}_{i},\vec{\boldsymbol{s}}_{i};\vec{\boldsymbol{y}}_{i})}{\log\gamma}=\lim_{\gamma\to\infty}\frac{\mathrm{I}(\boldsymbol{x}_{i};\vec{\boldsymbol{y}}_{i}|\vec{\boldsymbol{s}}_{i})}{\log\gamma}.\end{equation} 
 As such,  we ignore the term\footnote{It can be verified that $\mathrm{I}(\vec{\boldsymbol{s}}_{i};\vec{\boldsymbol{y}}_{i})=\sum_{\vec{s}\in \mathrm{supp}(\vec{\boldsymbol{s}}_{i})}\Pr\{\vec{\boldsymbol{s}}_{i}=\vec{s}\} \mathrm{D}\left(p_{\vec{\boldsymbol{y}}_{i}|\vec{\boldsymbol{s}}_{i}}(.|\vec{s})\|p_{\vec{\boldsymbol{y}}_{i}}(.)\right)$. This enables us to compute $\mathrm{I}(\vec{\boldsymbol{s}}_{i};\vec{\boldsymbol{y}}_{i})$ directly.} $\mathrm{I}(\vec{\boldsymbol{s}}_{i};\vec{\boldsymbol{y}}_{i})$ and focus on developing a tight lower bound on $\mathrm{I}(\boldsymbol{x}_{i};\vec{\boldsymbol{y}}_{i}|\vec{\boldsymbol{s}}_{i})$.

 To develop a lower bound on $\mathrm{I}(\boldsymbol{x}_{i};\vec{\boldsymbol{y}}_{i}|\vec{\boldsymbol{s}}_{i})$, our major tools are linear processing of the channel output based on Singular Value Decomposition of the signature code $\vec{\boldsymbol{s}}_{i}$, a conditional version of Entropy Power Inequality and a key upper bound on the differential entropy of a mixed Gaussian random vector. 
We have
\begin{equation}
\label{ }
\mathrm{I}(\boldsymbol{x}_{i};\vec{\boldsymbol{y}}_{i}|\vec{\boldsymbol{s}}_{i})=\sum_{\vec{s}\in\mathrm{supp}(\vec{\boldsymbol{s}}_{i})\backslash\{0_{K\times 1}\}}\Pr\{\vec{\boldsymbol{s}}_{i}=\vec{s}\}\mathrm{I}(\boldsymbol{x}_{i};\vec{\boldsymbol{y}}_{i}|\vec{\boldsymbol{s}}_{i}=\vec{s})\end{equation}

        
In the following, we find a lower bound on $\mathrm{I}(\boldsymbol{x}_{i};\vec{\boldsymbol{y}}_{i}|\vec{\boldsymbol{s}}_{i}=\vec{s})$ for any $\vec{s}\in\mathrm{supp}(\vec{\boldsymbol{s}}_{i})\backslash\{0_{K\times 1}\}$.  

\textbf{Step 1-}  
The matrix $\vec{s}\vec{s}^{\dagger}$ has two eigenvalues, namely zero and $\|\vec{s}\|_{2}^{2}$. The eigenvector corresponding to $\|\vec{s}\|_{2}^{2}$ is $\vec{s}$ and the the eigenvectors corresponding to zero are $K-1$ orthonormal vectors denoted by $\vec{g}_{1},\cdots, \vec{g}_{K-2}$ and $ \vec{g}_{K-1}$  which together with the columns of $\frac{\vec{s}}{\|\vec{s}\|_{2}}$ make an orthonormal basis for $\mathbb{R}^{K}$. Let us define
\begin{equation}
\label{ }
G_{i}(\vec{s})\triangleq\begin{pmatrix}
 \vec{g}_{i,1}&\cdots&\vec{g}_{i,K-1}
\end{pmatrix},
\end{equation}
\begin{equation}
\label{ }
U_{i}(\vec{s})\triangleq\begin{pmatrix}
    \frac{\vec{s}}{\|\vec{s}\|_{2}}  &G_{i}(\vec{s})\end{pmatrix}
\end{equation}
and
\begin{equation}
\label{ }
\vec{d}\triangleq\begin{pmatrix}
    \|\vec{s}\|_{2}   \\
      \vec{0}_{(K-1)\times 1}
\end{pmatrix}.
\end{equation}
Writing the SVD of $\vec{s}$, 
\begin{equation}
\label{ }
\vec{s}=U_{i}(\vec{s})\vec{d}.
\end{equation}
The $i^{th}$ receiver constructs the vector $U_{i}^{\dagger}(\vec{s})\vec{\boldsymbol{y}}_{i}\Big|_{\vec{\boldsymbol{s}}_{i}=\vec{s}}$ upon reception of $\vec{\boldsymbol{y}}_{i}$. We have
\begin{eqnarray}
U_{i}^{\dagger}(\vec{s})\vec{\boldsymbol{y}}_{i}\Big|_{\vec{\boldsymbol{s}}_{i}=\vec{s}}=\beta h_{i,i}\boldsymbol{x}_{i}\vec{d}+U_{i}^{\dagger}(\vec{s})\left(\vec{\boldsymbol{w}}_{i}+\vec{\boldsymbol{z}}_{i}\right).\end{eqnarray}  
We define
\begin{eqnarray}
\label{ }
\boldsymbol{\varphi}_{i}&\triangleq&\left[U_{i}^{\dagger}(\vec{s})\left(\vec{\boldsymbol{w}}_{i}+\vec{\boldsymbol{z}}_{i}\right)\right]_{1}\notag\\
&=&\frac{\beta\vec{s}^{\dagger}\left(\vec{\boldsymbol{w}}_{i}+\vec{\boldsymbol{z}}_{i}\right)}{\|\vec{s}\|_{2}}\end{eqnarray}
\begin{equation}
\label{ }
\boldsymbol{\omega}_{i}\triangleq\left[U_{i}^{\dagger}(\vec{s})\vec{\boldsymbol{y}}_{i}\right]_{1}=\beta h_{i,i}\|\vec{s}\|_{2}\boldsymbol{x}_{i}+\boldsymbol{\varphi}_{i}\end{equation}
and
\begin{eqnarray}
\vec{\boldsymbol{\vartheta}}_{i}&\triangleq&\left[U_{i}^{\dagger}(\vec{s})\vec{\boldsymbol{y}}_{i}\right]_{2}^{K}\notag\\&=&\left[U_{i}^{\dagger}(\vec{s})\left(\vec{\boldsymbol{w}}_{i}+\vec{\boldsymbol{z}}_{i}\right)\right]_{2}^{K}\notag\\
&\stackrel{}{=}&\left[U_{i}^{\dagger}(\vec{s})\left(\vec{\boldsymbol{w}}_{i}+\vec{\boldsymbol{z}}_{i}\right)\right]_{2}^{K}\notag\\&=&\left[\begin{pmatrix}
        \frac{\vec{s}^{\dagger}}{\|\vec{s}\|_{2}}    \\
      G_{i}^{\dagger}(\vec{s})
\end{pmatrix}\left(\vec{\boldsymbol{w}}_{i}+\vec{\boldsymbol{z}}_{i}\right)\right]_{2}^{K}\notag\\
&=&G_{i}^{\dagger}(\vec{s})\left(\vec{\boldsymbol{w}}_{i}+\vec{\boldsymbol{z}}_{i}\right).\notag\\\end{eqnarray}
We have the following  thread of equalities,
\begin{eqnarray}
\label{gh1}
\mathrm{I}(\boldsymbol{x}_{i};\vec{\boldsymbol{y}}_{i}|\vec{\boldsymbol{s}}_{i}=\vec{s})&=&\mathrm{I}(\boldsymbol{x}_{i};U_{i}^{\dagger}(\vec{s})\vec{\boldsymbol{y}}_{i}|\vec{\boldsymbol{s}}_{i}=\vec{s})\notag\\
&=&\mathrm{I}(\boldsymbol{x}_{i};\boldsymbol{\omega}_{i},\vec{\boldsymbol{\vartheta}}_{i})\notag\\
&=&\mathrm{I}(\boldsymbol{x}_{i};\vec{\boldsymbol{\vartheta}}_{i})+\mathrm{I}(\boldsymbol{x}_{i};\boldsymbol{\omega}_{i}|\vec{\boldsymbol{\vartheta}}_{i})\notag\\
&\stackrel{(a)}{=}&\mathrm{I}(\boldsymbol{x}_{i};\boldsymbol{\omega}_{i}|\vec{\boldsymbol{\vartheta}}_{i})\end{eqnarray}
where $(a)$ is by the fact that $\boldsymbol{x}_{i}$ and $\vec{\boldsymbol{\vartheta}}_{i}$ are independent, i.e., $\mathrm{I}(\boldsymbol{x}_{i};\vec{\boldsymbol{\vartheta}}_{i})=0$.

\textbf{Step 2-} In this part, we use the following Lemma without proof. 
\begin{lem}
Let $\vec{\mathbf{\Theta}}_{1}$ and $\vec{\mathbf{\Theta}}_{2}$ be $t\times 1$ complex random vectors and $\mathbf{\Theta}_{3}$ be any random quantity (scalar or vector) with densities. Also, assume that the conditional densities $p_{\vec{\mathbf{\Theta}}_{1}|\mathbf{\Theta}_{3}}(.|.)$ and $p_{\vec{\mathbf{\Theta}}_{2}|\mathbf{\Theta}_{3}}(.|.)$ exist. If $\vec{\mathbf{\Theta}}_{1}$ and $\vec{\mathbf{\Theta}}_{2}$ are conditionally independent given $\mathbf{\Theta}_{3}$, then
\begin{equation}
\label{ }
2^{\frac{1}{t}\mathrm{h}(\vec{\mathbf{\Theta}}_{1}+\vec{\mathbf{\Theta}}_{2}|\mathbf{\Theta}_{3})}\geq 2^{\frac{1}{t}\mathrm{h}(\vec{\mathbf{\Theta}}_{1}|\mathbf{\Theta}_{3})}+2^{\frac{1}{t}\mathrm{h}(\vec{\mathbf{\Theta}}_{2}|\mathbf{\Theta}_{3})}.\end{equation}
\end{lem}
We have
\begin{eqnarray}
\label{gh2}
\mathrm{I}(\boldsymbol{x}_{i};\boldsymbol{\omega}_{i}|\vec{\boldsymbol{\vartheta}}_{i})&=&\mathrm{h}(\boldsymbol{\omega}_{i}|\vec{\boldsymbol{\vartheta}}_{i})-\mathrm{h}(\boldsymbol{\omega}_{i}|\boldsymbol{x}_{i},\vec{\boldsymbol{\vartheta}}_{i})\notag\\
&=&\mathrm{h}(\boldsymbol{\omega}_{i}|\vec{\boldsymbol{\vartheta}}_{i})-\mathrm{h}\left(\beta h_{i,i}\|\vec{s}\|_{2}\boldsymbol{x}_{i}+\boldsymbol{\varphi}_{i}\big|\boldsymbol{x}_{i},\vec{\boldsymbol{\vartheta}}_{i}\right)\notag\\
&\stackrel{}{=}&\mathrm{h}(\boldsymbol{\omega}_{i}|\vec{\boldsymbol{\vartheta}}_{i})-\mathrm{h}(\boldsymbol{\varphi}_{i}|\vec{\boldsymbol{\vartheta}}_{i}).\end{eqnarray}
On the other hand, we know that $\boldsymbol{\omega}_{i}=\beta h_{i,i}\|\vec{s}\|_{2}\boldsymbol{x}_{i}+\boldsymbol{\varphi}_{i}$. Defining $\mathbf{\Theta}_{1}\triangleq\beta h_{i,i}\|\vec{s}\|_{2}\boldsymbol{x}_{i}$ and $\mathbf{\Theta}_{2}\triangleq\boldsymbol{\varphi}_{i}$, it is clear that $\mathbf{\Theta}_{1}$ and $\mathbf{\Theta}_{2}$ are conditionally independent given the collection of random variables $\mathbf{\Theta}_{3}\triangleq\vec{\boldsymbol{\vartheta}}_{i}$. As the conditional densities $p_{\mathbf{\Theta}_{1}|\mathbf{\Theta}_{3}}(.|.)$ and $p_{\mathbf{\Theta}_{2}|\mathbf{\Theta}_{3}}(.|.)$ exist, by Lemma 1, 
\begin{eqnarray}
\label{poe}
2^{\mathrm{h}\left(\boldsymbol{\omega}_{i}|\vec{\boldsymbol{\vartheta}}_{i}\right)}&\geq& 2^{\mathrm{h}\left(\beta h_{i,i}\|\vec{s}\|_{2}\boldsymbol{x}_{i}|\vec{\boldsymbol{\vartheta}}_{i}\right)}+2^{\mathrm{h}\left(\boldsymbol{\varphi}_{i}|\vec{\boldsymbol{\vartheta}}_{i}\right)}\notag\\
&\stackrel{(a)}{=}& 2^{\mathrm{h}\big(\beta h_{i,i}\|\vec{s}\|_{2}\boldsymbol{x}_{i}\big)}+2^{\mathrm{h}\left(\boldsymbol{\varphi}_{i}|\vec{\boldsymbol{\vartheta}}_{i}\right)}\end{eqnarray}
where $(a)$ is by the fact that the collection $\boldsymbol{x}_{i}$ is independent of $\vec{\boldsymbol{\vartheta}}_{i}$.
Dividing both sides of (\ref{poe}) by $2^{\frac{2}{v}\mathrm{h}(\vec{\boldsymbol{\varphi}}_{i}|\vec{\boldsymbol{\vartheta}}_{i})}$,
\begin{eqnarray}
\label{gh3}
&&\mathrm{h}\left(\boldsymbol{\omega}_{i}|\vec{\boldsymbol{\vartheta}}_{i}\right)-\mathrm{h}\left(\boldsymbol{\varphi}_{i}|\vec{\boldsymbol{\vartheta}}_{i}\right)\notag\\
&\geq& \log\left(2^{\left(\mathrm{h}\big(\beta h_{i,i}\|\vec{s}\|_{2}\boldsymbol{x}_{i}\big)-\mathrm{h}\left(\boldsymbol{\varphi}_{i}|\vec{\boldsymbol{\vartheta}}_{i}\right)\right)}+1\right).\notag\\\end{eqnarray}
By (\ref{gh1}), (\ref{gh2}) and (\ref{gh3}),
\begin{eqnarray}
\label{goosht}
\mathrm{I}(\boldsymbol{x}_{i};\vec{\boldsymbol{y}}_{i}|\vec{\boldsymbol{s}}_{i}=\vec{s})&\geq& \log\left(2^{\left(\mathrm{h}\big(\beta h_{i,i}\|\vec{s}\|_{2}\boldsymbol{x}_{i}\big)-\mathrm{h}\left(\boldsymbol{\varphi}_{i}|\vec{\boldsymbol{\vartheta}}_{i}\right)\right)}+1\right).\notag\\ \end{eqnarray}
\textbf{Step 3-}  
  We start by stating the following Lemma.
   \begin{lem}
   Let $\vec{\mathbf{\Theta}}$ be a $t\times 1$ mixed Gaussian random vector with the PDF
\begin{equation}
p_{\vec{\mathbf{\Theta}}}(\vec{\Theta})=\sum_{l=1}^{L}\frac{q_{l}}{\pi^{t}\det \Omega_{l}}\exp-\left(\vec{\Theta}^{T}\Omega_{l}^{-1}\vec{\Theta}\right)
\end{equation}
where $q_{l}\geq 0$ for $1\leq l\leq L$ and $\sum_{l=1}^{L}q_{l}=1$.
Then, 
   \begin{equation}
   \label{polm}
\sum_{l=1}^{L}q_{l}\log\big((\pi e)^{t}\det \Omega_{l}\big)\leq \mathrm{h}(\vec{\mathbf{\Theta}})\leq \sum_{l=1}^{L}q_{l}\log\big((\pi e)^{t}\det \Omega_{l}\big)+\mathrm{H}((q_{l})_{l=1}^{L})
\end{equation}
\end{lem}
     \begin{proof}
   Let us define the random matrix $\mathbf{\Omega}\in\{\Omega_{l}: 1\leq l\leq L\}$ such that $\Pr\{\mathbf{\Omega}=\Omega_{l}\}=q_{l}$ and let $\vec{\mathbf{\Upsilon}}$ be a zero mean Gaussian vector with covariance matrix $I_{t}$. Then, one can easily see that $\vec{\mathbf{\Theta}}=\sqrt{\mathbf{\Omega}}\vec{\mathbf{\Upsilon}}$ in which $\sqrt{\mathbf{\Omega}}$ is the conventional square root of a positive semi-definite matrix. Using the inequalities
   \begin{eqnarray}
\label{ }
\mathrm{h}(\vec{\mathbf{\Theta}}|\mathbf{\Omega})\leq \mathrm{h}(\vec{\mathbf{\Theta}})&\leq&\mathrm{h}(\vec{\mathbf{\Theta}},\mathbf{\Omega})\notag\\
&=&\mathrm{h}(\vec{\mathbf{\Theta}}|\mathbf{\Omega})+\mathrm{H}(\mathbf{\Omega})\notag\\
&=&\mathrm{h}(\vec{\mathbf{\Theta}}|\mathbf{\Omega})+\mathrm{H}((q_{l})_{l=1}^{L})
\end{eqnarray}  
and noting that $\mathrm{h}(\vec{\mathbf{\Theta}}|\mathbf{\Omega})=\sum_{l=1}^{L}q_{l}\log\left((\pi e)^{t}\det\Omega_{l}\right)$, the result is immediate.
       \end{proof}
    
      
The vector $\vec{\boldsymbol{w}}_{i}$ has a mixed Gaussian distribution where the covariance matrices of its separate Gaussian components correspond to different realizations of the matrix $\beta^{2}\gamma\sum_{j\neq i}|h_{j,i}|^{2}\boldsymbol{s}_{j}\boldsymbol{s}_{j}^{\dagger}$. This together with Lemma 2 yields
\begin{equation}
\label{bw3}
\mathrm{h}(\vec{\boldsymbol{w}}_{i}+\vec{\boldsymbol{z}}_{i})\leq\mathrm{h}(\vec{\boldsymbol{w}}_{i}+\vec{\boldsymbol{z}}_{i}|\boldsymbol{S}_{i})+\mathrm{H}\left(\sum_{j\neq i}|h_{j,i}|^{2}\boldsymbol{s}_{j}\boldsymbol{s}_{j}^{\dagger}\right)
\end{equation}
where we have used the fact that $\mathrm{H}\left(\beta^{2}\gamma\sum_{j\neq i}|h_{j,i}|^{2}\boldsymbol{s}_{j}\boldsymbol{s}_{j}^{\dagger}\right)=\mathrm{H}\left(\sum_{j\neq i}|h_{j,i}|^{2}\boldsymbol{s}_{j}\boldsymbol{s}_{j}^{\dagger}\right)$.
One has\begin{eqnarray}
\label{pork}
\mathrm{h}(\boldsymbol{\varphi}_{i}|\vec{\boldsymbol{\vartheta}}_{i})&=&\mathrm{h}(\boldsymbol{\varphi}_{i},\vec{\boldsymbol{\vartheta}}_{i})-\mathrm{h}(\vec{\boldsymbol{\vartheta}}_{i})\notag\\
&=&\mathrm{h}\left(U_{i}^{\dagger}(\vec{s})\left(\vec{\boldsymbol{w}}_{i}+\vec{\boldsymbol{z}}_{i}\right)\right)-\mathrm{h}\left(G_{i}^{\dagger}(\vec{s})\left(\vec{\boldsymbol{w}}_{i}+\vec{\boldsymbol{z}}_{i}\right)\right)\notag\\
&\stackrel{(a)}{=}&\mathrm{h}\left(\vec{\boldsymbol{w}}_{i}+\vec{\boldsymbol{z}}_{i}\right)-\mathrm{h}\left(G_{i}^{\dagger}(\vec{s})\left(\vec{\boldsymbol{w}}_{i}+\vec{\boldsymbol{z}}_{i}\right)\right)\notag\\
&\stackrel{(b)}{\leq}&\mathrm{h}\left(\vec{\boldsymbol{w}}_{i}+\vec{\boldsymbol{z}}_{i}|\boldsymbol{S}_{i}\right)+\mathrm{H}\left(\sum_{j\neq i}|h_{j,i}|^{2}\vec{\boldsymbol{s}}_{j}\vec{\boldsymbol{s}}_{j}^{\dagger}\right)\notag\\
&&-\mathrm{h}\left(G_{i}^{\dagger}(\vec{s})\left(\vec{\boldsymbol{w}}_{i}+\vec{\boldsymbol{z}}_{i}\right)\right)\notag\\
&\stackrel{(c)}{\leq}&\mathrm{h}\left(\vec{\boldsymbol{w}}_{i}+\vec{\boldsymbol{z}}_{i}|\boldsymbol{S}_{i}\right)+\mathrm{H}\left(\sum_{j\neq i}|h_{j,i}|^{2}\vec{\boldsymbol{s}}_{j}\vec{\boldsymbol{s}}_{j}^{\dagger}\right)\notag\\
&&-\mathrm{h}\left(G_{i}^{\dagger}(\vec{s})\left(\vec{\boldsymbol{w}}_{i}+\vec{\boldsymbol{z}}_{i}\right)|\boldsymbol{S}_{i}\right)\end{eqnarray}
where $(a)$ follows by the fact that the matrix $U_{i}(\vec{s})$ is unitary, i.e., $\log|\det(U_{i}(\vec{s}))|=0$, $(b)$ is by  (\ref{bw3}) and $(c)$ is a direct consequence of Lemma 2.  Having  $\boldsymbol{S}_{i}$, the vector $\vec{\boldsymbol{w}}_{i}+\vec{\boldsymbol{z}}_{i}$ is a complex Gaussian vector. Hence, 
   \begin{eqnarray}
 \label{lj1}
\mathrm{h}\left(\vec{\boldsymbol{w}}_{i}+\vec{\boldsymbol{z}}_{i}|\boldsymbol{S}_{i}\right)= K\log(\pi e)+\sum_{S\in\mathrm{supp}(\boldsymbol{S}_{i})}\Pr\{\boldsymbol{S}_{i}=S\}\log\det\left(I_{K}+\beta^{2}\gamma S\Xi_{i}\Xi_{i}^{\dagger}S^{\dagger}\right).\notag\\ \end{eqnarray}
       

  By the same token, 
    \begin{eqnarray}
  \label{lj2}
    &&\mathrm{h}\left(G_{i}^{\dagger}(\vec{s})\left(\vec{\boldsymbol{w}}_{i}+\vec{\boldsymbol{z}}_{i}\right)|\boldsymbol{S}_{i}\right)=(K-1)\log(\pi e)\notag\\&&+\sum_{\substack{S\in\mathrm{supp}(\boldsymbol{S}_{i})}}\Pr\{\boldsymbol{S}_{i}=S\}\log\det\left(I_{K-1}+\beta^{2}\gamma G_{i}^{\dagger}(\vec{s})S\Xi_{i}\Xi_{i}^{\dagger}S^{\dagger}G_{i}^{\dagger}(\vec{s})\right)\notag\\ \end{eqnarray}  
    Using (\ref{lj1}) and (\ref{lj2}) in (\ref{pork}), 
    \begin{eqnarray}
&&\mathrm{h}(\boldsymbol{\varphi}_{i}|\vec{\boldsymbol{\vartheta}}_{i})\leq \log(\pi e)+\mathrm{H}\left(\sum_{j\neq i}|h_{j,i}|^{2}\vec{\boldsymbol{s}}_{j}\vec{\boldsymbol{s}}_{j}^{\dagger}\right)\notag\\
&&+\sum_{S\in\mathrm{supp}(\boldsymbol{S}_{i})}\Pr\{\boldsymbol{S}_{i}=S\}\log\det\left(I_{K}+\beta^{2}\gamma S\Xi_{i}\Xi_{i}^{\dagger}S^{\dagger}\right)\notag\\
&&-\sum_{\substack{S\in\mathrm{supp}(\boldsymbol{S}_{i})}}\Pr\{\boldsymbol{S}_{i}=S\}\log\det\left(I_{K-1}+\beta^{2}\gamma G_{i}^{\dagger}(\vec{s})S\Xi_{i}\Xi_{i}^{\dagger}S^{\dagger}G_{i}^{\dagger}(\vec{s})\right).\end{eqnarray}
Moreover, $\mathrm{h}\big(\beta h_{i,i}\|\vec{s}\|_{2}\boldsymbol{x}_{i}\big)=\log\left(\pi e\beta^{2}|h_{i,i}|^{2}\|\vec{s}\|_{2}^{2}\gamma\right)$. Hence, $\mathrm{h}\big(\beta h_{i,i}\|\vec{s}\|_{2}\boldsymbol{x}_{i}\big)-\mathrm{h}\left(\boldsymbol{\varphi}_{i}|\vec{\boldsymbol{\vartheta}}_{i}\right)$ appearing in (\ref{goosht}) can be bounded from below as
\begin{eqnarray}
\label{bone}
&&\mathrm{h}\big(\beta h_{i,i}\|\vec{s}\|_{2}\boldsymbol{x}_{i}\big)-\mathrm{h}\left(\boldsymbol{\varphi}_{i}|\vec{\boldsymbol{\vartheta}}_{i}\right)\geq\log\left(\beta^{2}|h_{i,i}|^{2}\|\vec{s}\|_{2}^{2}\gamma\right)-\mathrm{H}\left(\sum_{j\neq i}|h_{j,i}|^{2}\vec{\boldsymbol{s}}_{j}\vec{\boldsymbol{s}}_{j}^{\dagger}\right)\notag\\
&&-\sum_{S\in\mathrm{supp}(\boldsymbol{S}_{i})}\Pr\{\boldsymbol{S}_{i}=S\}\log\det\left(I_{K}+\beta^{2}\gamma S\Xi_{i}\Xi_{i}^{\dagger}S^{\dagger}\right)\notag\\
&&+\sum_{\substack{S\in\mathrm{supp}(\boldsymbol{S}_{i})}}\Pr\{\boldsymbol{S}_{i}=S\}\log\det\left(I_{K-1}+\beta^{2}\gamma G_{i}^{\dagger}(\vec{s})S\Xi_{i}\Xi_{i}^{\dagger}S^{\dagger}G_{i}^{\dagger}(\vec{s})\right).\end{eqnarray}
 Substituting (\ref{bone}) in (\ref{goosht}), 
\begin{eqnarray}
\mathrm{I}(\boldsymbol{x}_{i};\vec{\boldsymbol{y}}_{i}|\vec{\boldsymbol{s}}_{i}=\vec{s})\geq\log\left(2^{-\mathrm{H}\left(\sum_{j\neq i}|h_{j,i}|^{2}\vec{\boldsymbol{s}}_{j}\vec{\boldsymbol{s}}_{j}^{\dagger}\right)}\varrho_{i}(\gamma;\vec{s})+1\right)\end{eqnarray}
where
\begin{eqnarray}
\label{booh}
\varrho_{i}(\gamma;\vec{s})\triangleq \frac{|h_{i,i}|^{2}\|\vec{s}\|_{2}^{2}\gamma}{\mathrm{E}\{\|\vec{\boldsymbol{s}}_{i}\|_{2}^{2}\}}\prod_{\substack{S\in\mathrm{supp}(\boldsymbol{S}_{i})}}\left(\frac{\det\left(I_{K-1}+\beta^{2}\gamma G_{i}^{\dagger}(\vec{s})S\Xi_{i}\Xi_{i}^{\dagger}S^{\dagger}G_{i}^{\dagger}(\vec{s})\right)}{\det\left(I_{K}+\beta^{2}\gamma S\Xi_{i}\Xi_{i}^{\dagger}S^{\dagger}\right)}\right)^{\Pr\{\boldsymbol{S}_{i}=S\}}.\end{eqnarray}
Finally, we get the following lower bound on $\frac{\mathrm{I}(\boldsymbol{x}_{i};\vec{\boldsymbol{y}}_{i}|\vec{\boldsymbol{s}}_{i})}{K}$ denoted by $\mathsf{C}_{i}^{(\mathrm{lb})}(\vec{h}_{i})$, i.e., 
\begin{equation}
\label{fuv}
\mathsf{C}_{i}^{(\mathrm{lb})}(\vec{h}_{i})\triangleq\frac{1}{K}\sum_{\vec{s}\in\mathrm{supp}(\vec{\boldsymbol{s}}_{i})\backslash\{0_{K\times 1}\}}\Pr\{\vec{\boldsymbol{s}}_{i}=\vec{s}\}\log\left(2^{-\mathrm{H}\left(\sum_{j\neq i}|h_{j,i}|^{2}\vec{\boldsymbol{s}}_{j}\vec{\boldsymbol{s}}_{j}^{\dagger}\right)}\varrho_{i}(\gamma;\vec{s})+1\right).\end{equation} 
 An important observation is that if the $i^{th}$ user sets its transmission rate at $R_{i}=\mathsf{C}_{i}^{(\mathrm{lb})}(\vec{h}_{i})$, then
 \begin{equation}\lim_{\gamma\to\infty}\frac{R_{i}}{\log\gamma}=\frac{\Pr\{\vec{\boldsymbol{s}}_{i}\notin\mathrm{csp}(\boldsymbol{S}_{i})\}}{K}.\end{equation} To prove this, we need some preliminary results in linear analysis. 

\textit{Definition 2}- Let $\mathscr{E}$ be an Euclidean space over $\mathbb{R}$ and $\mathscr{U}$ be a subspace of $\mathscr{E}$. We define
\begin{equation}
\label{ }
\mathscr{U}^{\perp}\triangleq\{v\in \mathscr{E}: v\perp u, \forall u\in\mathscr{U}\}.
\end{equation}
\begin{lem}
Let $\mathscr{E}$ be an Euclidean space over $\mathbb{R}$. If $\mathscr{U}$ is a subspace of $\mathscr{E}$, then for each $v\in\mathscr{E}$, there are unique elements $v_{1}\in\mathscr{U}$ and $v_{2}\in\mathscr{U}^{\perp}$ such that $v=v_{1}+v_{2}$.  
\end{lem}
\textit{Definition 3}- In the setup of Lemma 4, $v_{1}$ is called the projection of $v$ in $\mathscr{U}$ and is denoted by $\mathrm{proj}(v;\mathscr{U})$. By the same token, $v_{2}=\mathrm{proj}(v;\mathscr{U}^{\perp})$.

\textit{Definition 4}- Let $\mathscr{E}$ be an Euclidean space over $\mathbb{R}$ and $\mathscr{U}_{1}$ and $\mathscr{U}_{2}$ be subspaces of $\mathscr{E}$. We define
\begin{equation}
\label{ }
\mathrm{proj}(\mathscr{U}_{1};\mathscr{U}_{2})\triangleq \mathrm{span}\{\mathrm{proj}(v;\mathscr{U}_{2}) : v\in\mathscr{U}_{1}\}.
\end{equation}  
 \begin{lem}
Let $\mathscr{E}$ be an Euclidean vector space over $\mathbb{R}$ and $\mathscr{U}_{1}$ and $\mathscr{U}_{2}$ be subspaces of $\mathscr{E}$. Then,
\begin{equation}
\label{ }
\mathrm{dim}(\mathscr{U}_{1}\cup\mathscr{U}_{2})=\mathrm{dim}(\mathscr{U}_{1})+\mathrm{dim}(\mathrm{proj}(\mathscr{U}_{2};\mathscr{U}_{1}^{\perp})).
\end{equation}
\end{lem}
\begin{lem}
Let $X$ be a $p\times q$ matrix such that $\mathrm{rank}(X)=q$. Then, for any $q\times r$ matrix $Y$, we have $\mathrm{rank}(XY)=\mathrm{rank}(Y)$.
\end{lem}
\begin{proposition}
Regulating its transmission rate at $\mathsf{C}_{i}^{(\mathrm{lb})}(\vec{h}_{i})$, the $i^{th}$ user achieves an SNR scaling of \begin{equation}
\label{}
\lim_{\gamma\to\infty}\frac{\mathsf{C}_{i}^{(\mathrm{lb})}(\vec{h}_{i})}{\log\gamma}=\frac{\Pr\{\vec{\boldsymbol{s}}_{i}\notin\mathrm{csp}(\boldsymbol{S}_{i})\}}{K}.\end{equation}
\end{proposition}
\begin{proof}
Using the fact that for any matrix $X$, $\mathrm{rank}(XX^{\dagger})=\mathrm{rank}(X)$, it is easy to see that for any $\vec{s}\in\mathrm{supp}(\vec{\boldsymbol{s}}_{i})$ and $S\in\mathrm{supp}(\boldsymbol{S}_{i})$, we have $\log\det\left(I_{K-1}+\beta^{2}\gamma G_{i}^{\dagger}(\vec{s})S\Xi_{i}\Xi_{i}^{\dagger}S^{\dagger}G_{i}^{\dagger}(\vec{s})\right)$ scales like $\mathrm{rank}(G_{i}^{\dagger}S)\log\gamma$ and $\log\det\left(I_{K}+\beta^{2}\gamma S\Xi_{i}\Xi_{i}^{\dagger}S^{\dagger}\right)$ scales like $\mathrm{rank}(S)\log\gamma$. This yields
\begin{eqnarray}
\label{bghm}
\lim_{\gamma\to\infty}\frac{\mathsf{C}_{i}^{(\mathrm{lb})}(\vec{h}_{i})}{\log\gamma}&=&\sum_{\vec{s}\in\mathrm{supp}(\vec{\boldsymbol{s}}_{i})\backslash\{0_{K\times 1}\}}\Pr\{\vec{\boldsymbol{s}}_{i}=\vec{s}\}
\notag\\&&+\sum_{\substack{S\in\mathrm{supp}(\boldsymbol{S}_{i})\\\vec{s}\in\mathrm{supp}(\vec{\boldsymbol{s}}_{i})\backslash\{0_{K\times 1}\}}}\Pr\{\boldsymbol{S}_{i}=S\}\Pr\{\vec{\boldsymbol{s}}_{i}=\vec{s}\}\mathrm{rank}(G_{i}^{\dagger}(\vec{s})S)\notag\\
&&-\sum_{\substack{S\in\mathrm{supp}(\boldsymbol{S}_{i})\\\vec{s}\in\mathrm{supp}(\vec{\boldsymbol{s}}_{i})\backslash\{0_{K\times 1}\}}}\Pr\{\boldsymbol{S}_{i}=S\}\Pr\{\vec{\boldsymbol{s}}_{i}=\vec{s}\}\mathrm{rank}(S)\notag\\
&=&\Pr\{\vec{\boldsymbol{s}}_{i}\neq0_{K\times 1}\}\notag\\&&+\sum_{\substack{S\in\mathrm{supp}(\boldsymbol{S}_{i})\\\vec{s}\in\mathrm{supp}(\vec{\boldsymbol{s}}_{i})}}\Pr\{\boldsymbol{S}_{i}=S\}\Pr\{\vec{\boldsymbol{s}}_{i}=\vec{s}\}\mathrm{rank}(G_{i}^{\dagger}(\vec{s})S)\notag\\
&&-\sum_{S\in\mathrm{supp}(\boldsymbol{S}_{i})}\Pr\{\boldsymbol{S}_{i}=S\}\Pr\{\vec{\boldsymbol{s}}_{i}=0_{K\times 1}\}\mathrm{rank}(G_{i}^{\dagger}(0_{K\times 1})S)\notag\\&&-\mathrm{E}\{\mathrm{rank}(\boldsymbol{S}_{i})\}\Pr\{\vec{\boldsymbol{s}}_{i}\neq 0_{K\times 1}\}\notag\\
&\stackrel{(a)}{=}&\Pr\{\vec{\boldsymbol{s}}_{i}\neq0_{K\times 1}\}\notag\\&&+\mathrm{E}\left\{\mathrm{rank}(\boldsymbol{G}_{i}^{\dagger}(\vec{\boldsymbol{s}}_{i})\boldsymbol{S}_{i})\right\}\notag\\
&&-\mathrm{E}\{\mathrm{rank}(\boldsymbol{S}_{i})\}\Pr\{\vec{\boldsymbol{s}}_{i}= 0_{K\times 1}\}\notag\\&&-\mathrm{E}\{\mathrm{rank}(\boldsymbol{S}_{i})\}\Pr\{\vec{\boldsymbol{s}}_{i}\neq 0_{K\times 1}\}\notag\\&=&\Pr\{\vec{\boldsymbol{s}}_{i}\neq0_{K\times 1}\}+\mathrm{E}\left\{\mathrm{rank}(\boldsymbol{G}_{i}^{\dagger}(\vec{\boldsymbol{s}}_{i})\boldsymbol{S}_{i})-\mathrm{rank}(\boldsymbol{S}_{i})\right\}\end{eqnarray}
where $(a)$ is by the fact that $G_{i}(0_{K\times 1})=I_{K}$.
 We  show that
\begin{equation}  
\label{ }
\mathbb{1}_{\{\vec{\boldsymbol{s}}_{i}\neq0_{K\times 1}\}}+\mathrm{rank}(\boldsymbol{G}_{i}^{\dagger}(\vec{\boldsymbol{s}}_{i})\boldsymbol{S}_{i})=\mathrm{rank}\left([\boldsymbol{S}_i|\vec{\boldsymbol{s}}_i]\right)\end{equation}
holds almost surely.  

Let us write
\begin{equation}
\label{kom}
\mathrm{rank}([\boldsymbol{S}_i|\vec{\boldsymbol{s}}_i])=\mathrm{dim}(\mathrm{span}(\vec{\boldsymbol{s}}_i)\cup\mathrm{csp}(\boldsymbol{S}_{i})).
\end{equation}
Using this in Lemma 4, 
\begin{eqnarray}
\label{bzz6}
\mathrm{rank}([\boldsymbol{S}_i|\vec{\boldsymbol{s}}_i])&=&\mathrm{dim}(\mathrm{span}(\vec{\boldsymbol{s}}_i))+\mathrm{dim}(\mathrm{proj}(\mathrm{csp}(\boldsymbol{S}_{i});(\mathrm{span}(\vec{\boldsymbol{s}}_i))^{\perp}))\notag\\
&=&\mathbb{1}_{\{\vec{\boldsymbol{s}}_{i}\neq0_{K\times 1}\}}+\mathrm{dim}(\mathrm{proj}(\mathrm{csp}(\boldsymbol{S}_{i});(\mathrm{span}(\vec{\boldsymbol{s}}_i))^{\perp})).\end{eqnarray}
On the other hand, by the definition of $\boldsymbol{G}_{i}(\vec{\boldsymbol{s}}_{i})$, 
\begin{equation}
\label{bzz4}(\mathrm{span}(\vec{\boldsymbol{s}}_i))^{\perp}=\mathrm{csp}(\boldsymbol{G}_{i}(\vec{\boldsymbol{s}}_{i})).\end{equation} It is easily seen that for any $1\leq k\leq K-1$, the $k^{th}$ column of the matrix $\boldsymbol{G}_{i}^{\dagger}(\vec{\boldsymbol{s}}_{i})\boldsymbol{s}_{i}$ yields the proper linear combination of the columns of $\boldsymbol{G}_{i}(\vec{\boldsymbol{s}}_{i})$ which constructs the projection of the $k^{th}$ column of $\boldsymbol{s}_{i}$ into the space $\mathrm{csp}(\boldsymbol{G}_{i}(\vec{\boldsymbol{s}}_{i}))$, i.e., 
\begin{equation}
\label{bzz3}
\boldsymbol{G}_{i}(\vec{\boldsymbol{s}}_{i})[\boldsymbol{G}_{i}^{\dagger}(\vec{\boldsymbol{s}}_{i})\boldsymbol{S}_{i}]_{k}=\mathrm{proj}\left([\boldsymbol{S}_{i}]_{k};\mathrm{csp}(\boldsymbol{G}_{i}(\vec{\boldsymbol{s}}_{i}))\right).
\end{equation}
Therefore, 
\begin{eqnarray}
\label{bzz2}
\mathrm{span}\left(\Big\{\boldsymbol{G}_{i}(\vec{\boldsymbol{s}}_{i})[\boldsymbol{G}_{i}^{\dagger}(\vec{\boldsymbol{s}}_{i})\boldsymbol{s}_{i}]_{k}\Big\}_{k=1}^{K-1}\right)&=&\mathrm{proj}\left(\mathrm{span}\left(\Big\{[\boldsymbol{S}_{i}]_{k}\Big\}_{k=1}^{K-1}\right);\mathrm{csp}(\boldsymbol{G}_{i}(\vec{\boldsymbol{s}}_{i}))\right)\notag\\
&=&\mathrm{proj}(\mathrm{csp}(\boldsymbol{S}_{i});\mathrm{csp}(\boldsymbol{G}_{i}(\vec{\boldsymbol{s}}_{i}))).\end{eqnarray}
However, \begin{equation}
\label{bzz1}
\mathrm{span}\left(\Big\{\boldsymbol{G}_{i}(\vec{\boldsymbol{s}}_{i})[\boldsymbol{G}_{i}^{\dagger}(\vec{\boldsymbol{s}}_{i})\boldsymbol{S}_{i}]_{k}\Big\}_{k=1}^{K}\right)=\mathrm{csp}(\boldsymbol{G}_{i}(\vec{\boldsymbol{s}}_{i})\boldsymbol{G}_{i}^{\dagger}(\vec{\boldsymbol{s}}_{i})\boldsymbol{S}_{i}).\end{equation}
By (\ref{bzz1}) and (\ref{bzz2}), 
\begin{equation}
\label{bzz5}
\mathrm{proj}(\mathrm{csp}(\boldsymbol{S}_{i});\mathrm{csp}(\boldsymbol{G}_{i}(\vec{\boldsymbol{s}}_{i})))=\mathrm{csp}(\boldsymbol{G}_{i}(\vec{\boldsymbol{s}}_{i})\boldsymbol{G}_{i}^{\dagger}(\vec{\boldsymbol{s}}_{i})\boldsymbol{S}_{i}).\end{equation}
Using (\ref{bzz5}) and (\ref{bzz4}) in (\ref{bzz6}),
\begin{eqnarray}
\mathrm{rank}([\boldsymbol{S}_{i}|\vec{\boldsymbol{s}}_{i}])&=&\mathbb{1}_{\{\vec{\boldsymbol{s}}_{i}\neq0_{K\times 1}\}}+\mathrm{dim}(\mathrm{csp}(\boldsymbol{G}_{i}(\vec{\boldsymbol{s}}_{i})\boldsymbol{G}_{i}^{\dagger}(\vec{\boldsymbol{s}}_{i})\boldsymbol{S}_{i}))\notag\\
&=&\mathbb{1}_{\{\vec{\boldsymbol{s}}_{i}\neq0_{K\times 1}\}}+\mathrm{rank}(\boldsymbol{G}_{i}(\vec{\boldsymbol{s}}_{i})\boldsymbol{G}_{i}^{\dagger}(\vec{\boldsymbol{s}}_{i})\boldsymbol{S}_{i})\notag\\
&\stackrel{(a)}{=}&\mathbb{1}_{\{\vec{\boldsymbol{s}}_{i}\neq0_{K\times 1}\}}+\mathrm{rank}(\boldsymbol{G}_{i}^{\dagger}(\vec{\boldsymbol{s}}_{i})\boldsymbol{S}_{i})\end{eqnarray}
where $(a)$ follows by Lemma 5 as $\boldsymbol{G}_{i}(\vec{\boldsymbol{s}}_{i})$ has independent columns.
Taking expectation from both sides, 
\begin{equation}
\label{ }
\mathrm{E}\left\{\mathrm{rank}([\boldsymbol{S}_{i}|\vec{\boldsymbol{s}}_{i}])\right\}=\Pr\{\vec{\boldsymbol{s}}_{i}\neq0_{K\times 1}\}+\mathrm{E}\left\{\mathrm{rank}(\boldsymbol{G}_{i}^{\dagger}(\vec{\boldsymbol{s}}_{i})\boldsymbol{S}_{i})\right\}.
\end{equation}
 Using this in (\ref{bghm}), 
\begin{eqnarray}
\lim_{\gamma\to\infty}\frac{\mathsf{C}_{i}^{(\mathrm{lb})}}{\log\gamma}&=&\mathrm{E}\left\{\mathrm{rank}([\boldsymbol{S}_{i}|\vec{\boldsymbol{s}}_{i}])-\mathrm{rank}(\boldsymbol{S}_{i})\right\}\notag\\
&=&\Pr\left\{\vec{\boldsymbol{s}}_{i}\notin\mathrm{csp}(\boldsymbol{S}_{i})\right\}.
\end{eqnarray}
This completes the proof.
 \end{proof}
Finally, the following Proposition proves that $\mathsf{C}_{i}$ and $\mathsf{C}_{i}^{(\mathrm{lb})}$ have the SNR scaling.
\begin{proposition}
$\mathsf{C}_{i}^{(\mathrm{lb})}$ and $\mathsf{C}_{i}$ have the same SNR scaling. 
\end{proposition}
\begin{proof}
See Appendix A. 
\end{proof}

 An important consequence of Proposition 1 is the following observation. Since, all the users utilize the same algorithm to construct their randomized signature codes, the achievable  SMG is
 \begin{equation}
\label{canada}
\mathsf{SMG}(n)=\frac{n\Pr\left\{\vec{\boldsymbol{s}}_{1}\notin\mathrm{csp}\left([\vec{\boldsymbol{s}}_{2}|\vec{\boldsymbol{s}}_{3}|\cdots|\vec{\boldsymbol{s}}_{n-1}|\vec{\boldsymbol{s}}_{n}]\right)\right\}}{K}.
\end{equation}
Computing $\Pr\left\{\vec{\boldsymbol{s}}_{1}\notin\mathrm{csp}\left([\vec{\boldsymbol{s}}_{2}|\vec{\boldsymbol{s}}_{3}|\cdots|\vec{\boldsymbol{s}}_{n-1}|\vec{\boldsymbol{s}}_{n}]\right)\right\}$ can be quite a tedious task specially for $n\geq 3$.  Let the underlying alphabet to construct the spreading codes be $\{-1,1\}$. Here, we examine two particular RSCs by computing the achieved $\mathsf{SMG}(n)$ through simulations for the cases where masking is applied or ignored. In each case, we assume the elements of any randomized spreading code are selected independently and uniformly over $\{-1,1\}$, i.e., $\mathsf{p}_{1}=\mathsf{p}_{-1}=\frac{1}{2}$. In case masking is applied, we set $\varepsilon=\frac{1}{2}$. Taking $K=n$, the results are sketched in fig. \ref{f2}. It is seen that 

\textit{1-} By increasing $n$, the achieved $\mathsf{SMG}(n)$ approaches unity in both cases.   This is the SMG of a frequency division scenario where interference is completely avoided. 

\textit{2-} Masking improves the SMG.

 \begin{figure}[h!b!t]
  \centering
  \includegraphics[scale=.7] {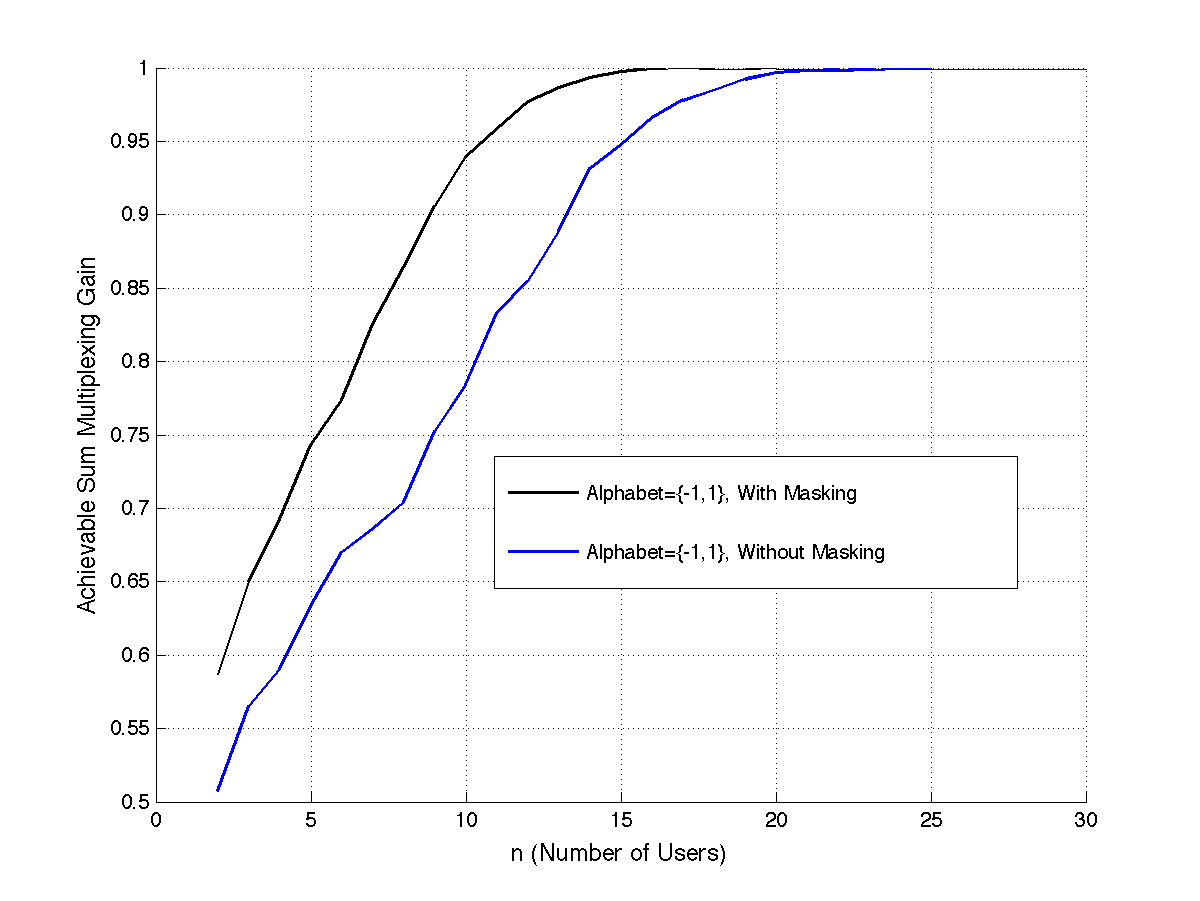}
  \caption{Comparison of the achieved SMG with/without masking where users construct their spreading codes on $\{-1,1\}$ using a uniform PMF $\mathsf{p}_{1}=\mathsf{p}_{-1}=\frac{1}{2}$. It is assumed that $K=n$. In case masking is applied, we have $\varepsilon=\frac{1}{2}$.}
  \label{f2}
 \end{figure} 
 
 \textit{Example 1-} Let us consider an RSC scheme where $K=1$, i.e., no spreading is applied. In this case, for each $1\leq i\leq n$, the vector $\vec{\boldsymbol{s}}_{i}=\vec{\boldsymbol{s}}_{i}=\boldsymbol{s}_{i}\in\{0,1\}$ is simply a $\mathrm{Ber}(\varepsilon)$ random variable for some $\varepsilon\in(0,1]$. Hence, 
 \begin{eqnarray}
\Pr\left\{\vec{\boldsymbol{s}}_{1}\notin\mathrm{csp}\left([\vec{\boldsymbol{s}}_{2}|\vec{\boldsymbol{s}}_{3}|\cdots|\vec{\boldsymbol{s}}_{n-1}|\vec{\boldsymbol{s}}_{n}]\right)\right\}&=&\Pr\left\{\boldsymbol{s}_{1}\notin\mathrm{span}\big(\left\{\boldsymbol{s}_{2},\boldsymbol{s}_{3},\cdots,\boldsymbol{s}_{n-1},\boldsymbol{s}_{n}\right\}\big)\right\}\notag\\
&=&\bar{\varepsilon}\Pr\left\{0\notin\mathrm{span}\big(\left\{\boldsymbol{s}_{2},\boldsymbol{s}_{3},\cdots,\boldsymbol{s}_{n-1},\boldsymbol{s}_{n}\right\}\big)\right\}\notag\\
&&+\varepsilon\Pr\left\{1\notin\mathrm{span}\big(\left\{\boldsymbol{s}_{2},\boldsymbol{s}_{3},\cdots,\boldsymbol{s}_{n-1},\boldsymbol{s}_{n}\right\}\big)\right\}\notag\\
&\stackrel{(a)}{=}&\varepsilon\Pr\left\{1\notin\mathrm{span}\big(\left\{\boldsymbol{s}_{2},\boldsymbol{s}_{3},\cdots,\boldsymbol{s}_{n-1},\boldsymbol{s}_{n}\right\}\big)\right\}\notag\\
&\stackrel{(b)}{=}&\varepsilon\Pr\left\{\boldsymbol{s}_{2}=\boldsymbol{s}_{3}=\cdots=\boldsymbol{s}_{n-1}=\boldsymbol{s}_{n}=0\right\}\notag\\
&=&\varepsilon(1-\varepsilon)^{n-1}
\end{eqnarray} 
 where $(a)$ is by the fact that $\Pr\left\{0\notin\mathrm{span}\big(\left\{\boldsymbol{s}_{2},\boldsymbol{s}_{3},\cdots,\boldsymbol{s}_{n-1},\boldsymbol{s}_{n}\right\}\big)\right\}=0$ and $(b)$ is by the fact that $1\notin\mathrm{span}\big(\left\{\boldsymbol{s}_{2},\boldsymbol{s}_{3},\cdots,\boldsymbol{s}_{n-1},\boldsymbol{s}_{n}\right\}\big)$ whenever $\boldsymbol{s}_{i}=0$ for $2\leq i\leq n$. Maximizing $\varepsilon(1-\varepsilon)^{n-1}$ over $\varepsilon$, a Sum Multiplexing Gain of $\left(1-\frac{1}{n}\right)^{n-1}$ is achieved. Increasing $n$, the achieved $\mathsf{SMG}(n)$ drops to $\frac{1}{e}<1$. Comparing this to the results in fig. \ref{f2}, spreading the signals ($K=n$ compared to $K=1$) can highly improve the Sum Multiplexing Gain in the network. $\square$
 
 \textit{Example 2-} Let $n=2$ and $\mathscr{A}=\{-1,1\}$. For $i\in\{1,2\}$, elements of $\vec{\boldsymbol{s}}_{i}$ are $\mathrm{i.i.d.}$ random variables taking the values $0$, $1$ and $-1$ with probabilities $\overline{\varepsilon}$, $\varepsilon\mathsf{p}_{1}$ and $\varepsilon\mathsf{p}_{-1}$ respectively.  We have
 \begin{eqnarray}
\Pr\{\vec{\boldsymbol{s}}_{1}\notin\mathrm{span}(\left\{\vec{\boldsymbol{s}}_{2}\right\})\}&=&1-\Pr\{\vec{\boldsymbol{s}}_{1}\in\mathrm{span}(\left\{\vec{\boldsymbol{s}}_{2}\right\})\}\notag\\
&=&1-\Pr\{\vec{\boldsymbol{s}}_{1}=0_{K\times 1}\}-\Pr\{\vec{\boldsymbol{s}}_{1}\neq0_{K\times 1},\vec{\boldsymbol{s}}_{1}=\pm\vec{\boldsymbol{s}}_{2}\}\notag\\
&=&1-\overline{\varepsilon}^{K}-\Pr\{\vec{\boldsymbol{s}}_{1}\neq0_{K\times 1},\vec{\boldsymbol{s}}_{1}=\vec{\boldsymbol{s}}_{2}\}-\Pr\{\vec{\boldsymbol{s}}_{1}\neq0_{K\times 1},\vec{\boldsymbol{s}}_{1}=-\vec{\boldsymbol{s}}_{2}\}.\end{eqnarray}
However, 
\begin{eqnarray}
\Pr\{\vec{\boldsymbol{s}}_{1}\neq0_{K\times 1},\vec{\boldsymbol{s}}_{1}=\vec{\boldsymbol{s}}_{2}\}&=&\sum_{\vec{s}\in\mathrm{supp}(\vec{\boldsymbol{s}}_{1})\backslash\{0_{K\times 1}\}}\left(\Pr\{\vec{\boldsymbol{s}}_{1}=\vec{s}\}\right)^{2}\notag\\
&=&\sum_{k=0}^{K-1}\sum_{l=0}^{K-k}{K\choose k}{K-k\choose l}\overline{\varepsilon}^{2k}(\varepsilon\mathsf{p}_{1})^{2l}(\varepsilon\mathsf{p}_{-1})^{2(K-k-l)}\notag\\
&=&\left(\overline{\varepsilon}^{2}+\varepsilon^{2}(\mathsf{p}_{1}^{2}+\mathsf{p}_{-1}^{2})\right)^{K}-\overline{\varepsilon}^{2K}.
\end{eqnarray}
Similarly, 
\begin{equation}
\label{ }
\Pr\{\vec{\boldsymbol{s}}_{1}\neq0_{K\times 1},\vec{\boldsymbol{s}}_{1}=\vec{\boldsymbol{s}}_{2}\}=\left(\overline{\varepsilon}^{2}+2\varepsilon^{2}\mathsf{p}_{1}\mathsf{p}_{-1}\right)^{K}-\overline{\varepsilon}^{2K}.
\end{equation}
Therefore, 
\begin{eqnarray}
\label{ }
\mathsf{SMG}(2)&=&\frac{2\Pr\{\vec{\boldsymbol{s}}_{1}\notin\mathrm{span}(\left\{\vec{\boldsymbol{s}}_{2}\right\})\}}{K}\notag\\&=&\frac{2}{K}\left(1-\overline{\varepsilon}^{K}+2\overline{\varepsilon}^{2K}-\left(\overline{\varepsilon}^{2}+\varepsilon^{2}(\mathsf{p}_{1}^{2}+\mathsf{p}_{-1}^{2})\right)^{K}-\left(\overline{\varepsilon}^{2}+2\varepsilon^{2}\mathsf{p}_{1}\mathsf{p}_{-1}\right)^{K}\right).\end{eqnarray}
This expression is maximized at $\mathsf{p}_{1}=\mathsf{p}_{-1}=\frac{1}{2}$ uniformly for any $\varepsilon\in(0,1]$ and $K\geq 1$. Thus, 
\begin{equation}
\label{boro}
\sup_{\mathsf{p}_{1},\mathsf{p}_{-1}}\mathsf{SMG}(2)=\frac{2\left(1-\overline{\varepsilon}^{K}+2\overline{\varepsilon}^{2K}-2\left(\overline{\varepsilon}^{2}+\frac{\varepsilon^{2}}{2}\right)^{K}\right)}{K}.
\end{equation} 
This function is maximized at $K=2$ and $\varepsilon=0.756$ where an SMG  of $\sup_{\varepsilon, K, \mathsf{p}_{1},\mathsf{p}_{-1}}\mathsf{SMG}(2)=0.7091$ is achieved.  We notice that 

\textit{1-} Although one's intuition expects $\varepsilon=\frac{1}{2}$ is the best choice of the On-Off probability, the \emph{optimum} masking probability is not $\frac{1}{2}$.

\textit{2-} Compared to the Sum Multiplexing Gain of $\frac{1}{2}$ achieved in example 1 without spreading, we see that spreading in fact increases the achieved SMG. $\square$

\textit{Remark 1-} For any $m_{1}\times m_{2}$ matrix $A$ and a $m_{1}\times m_{1}$ diagonal matrix $D$ we have $\mathrm{rank}(AA^{\dagger}D)=\mathrm{rank}(A)$. Using this, for any $\vec{s}\in\mathrm{supp}(\vec{\boldsymbol{s}}_{i})\backslash\{0_{K\times 1}\}$ one can write  
\begin{eqnarray}
\label{booh}
\varrho_{i}(\gamma;\vec{s})\triangleq \frac{|h_{i,i}|^{2}\|\vec{s}\|_{2}^{2}\gamma}{\mathrm{E}\{\|\vec{\boldsymbol{s}}_{i}\|_{2}^{2}\}}\prod_{\substack{S\in\mathrm{supp}(\boldsymbol{S}_{i})}}\frac{\prod_{l=1}^{\mathrm{rank}(G_{i}^{\dagger}(\vec{s})S)}\left(1+\frac{\gamma \lambda_{G_{i}^{\dagger}(\vec{s})S\Xi_{i}}^{(l)}}{\mathrm{E}\left\{\|\vec{\boldsymbol{s}}_{i}\|_{2}^{2}\right\}}\right)^{\Pr\{\boldsymbol{S}_{i}=S\}}}{\prod_{l=1}^{\mathrm{rank(S)}}\left(1+\frac{\gamma \lambda_{S\Xi_{i}}^{(l)}}{\mathrm{E}\left\{\|\vec{\boldsymbol{s}}_{i}\|_{2}^{2}\right\}}\right)^{\Pr\{\boldsymbol{S}_{i}=S\}}}.\end{eqnarray}

where we have replaces $\beta^{2}=\frac{1}{\mathrm{E}\left\{\|\vec{\boldsymbol{s}}_{i}\|_{2}^{2}\right\}}$ and by definition, $(\lambda_{A}^{(l)})_{l=1}^{\mathrm{rank}(A)}$ are nonzero eigenvalues of the matrix $AA^{\dagger}$.
 
For sufficiently large SNR values, one can write $\mathsf{C}_{i}^{(\mathrm{lb})}(\vec{h}_{i})$ given in (\ref{fuv}) as
 \begin{eqnarray}
\label{fuvvv}
&&\mathsf{C}_{i}^{(\mathrm{lb})}(\vec{h}_{i})\approx\frac{\Pr\{\vec{\boldsymbol{s}}_{i}\notin\mathrm{csp}(\boldsymbol{S}_{i})\}}{K}\log\gamma-\frac{\Pr\{\vec{\boldsymbol{s}}_{i}\neq 0_{K\times 1}\}\mathrm{H}\left(\sum_{j\neq i}|h_{j,i}|^{2}\vec{\boldsymbol{s}}_{j}\vec{\boldsymbol{s}}_{j}^{\dagger}\right)}{K}\notag\\
&&+\frac{1}{K}\sum_{\vec{s}\in\mathrm{supp}(\vec{\boldsymbol{s}}_{i})\backslash\{0_{K\times 1}\}}\Pr\{\vec{\boldsymbol{s}}_{i}=\vec{s}\}\log\left(\frac{|h_{i,i}|^{2}\|\vec{s}\|_{2}^{2}}{\mathrm{E}\{\|\vec{\boldsymbol{s}}_{i}\|_{2}^{2}\}} \prod_{\substack{S\in\mathrm{supp}(\boldsymbol{S}_{i})}}\frac{\prod_{l=1}^{\mathrm{rank}(G_{i}^{\dagger}(\vec{s})S)}\left(\frac{ \underline{\pi}_{i}\lambda_{G_{i}^{\dagger}(\vec{s})S\Xi_{i}}^{(l)}}{\mathrm{E}\{\|\vec{\boldsymbol{s}}_{i}\|_{2}^{2}\}}\right)^{\Pr\{\boldsymbol{S}_{i}=S\}}}{ \prod_{l=1}^{\mathrm{rank}(S)}\left(\frac{\overline{\pi}_{i}\lambda_{S\Xi_{i}}^{(l)}}{\mathrm{E}\{\|\vec{\boldsymbol{s}}_{i}\|_{2}^{2}\}}\right)^{\Pr\{\boldsymbol{S}_{i}=S\}}}\right).\notag\\\end{eqnarray}

 There are three major factors playing role in the formulation of $\mathsf{C}_{i}^{(\mathrm{lb})}(\vec{h}_{i})$ in the high SNR regime, namely, the \emph{Multiplexing Gain per user},
 \begin{equation}
\label{MG}
\mathsf{MG}\triangleq\frac{\Pr\{\vec{\boldsymbol{s}}_{i}\notin\mathrm{csp}(\boldsymbol{S}_{i})\}}{K},\end{equation}
 the \emph{Interference Entropy Factor},
\begin{equation}
\label{ }
\mathsf{IEF}\triangleq \frac{\Pr\{\vec{\boldsymbol{s}}_{i}\neq0_{K\times 1}\}\mathrm{H}\left(\sum_{j\neq i}|h_{j,i}|^{2}\vec{\boldsymbol{s}}_{j}\vec{\boldsymbol{s}}_{j}^{\dagger}\right)}{K}\end{equation}
and the \emph{Channel plus Signature Factor}
\begin{equation}
\label{ }
\mathsf{CSF}_{i}\triangleq \frac{1}{K}\sum_{\vec{s}\in\mathrm{supp}(\vec{\boldsymbol{s}}_{i})\backslash\{0_{K\times 1}\}}\Pr\{\vec{\boldsymbol{s}}_{i}=\vec{s}\}\log\left(\frac{|h_{i,i}|^{2}\|\vec{s}\|_{2}^{2}}{\mathrm{E}\{\|\vec{\boldsymbol{s}}_{i}\|_{2}^{2}\}} \prod_{\substack{S\in\mathrm{supp}(\boldsymbol{S}_{i})}}\frac{\prod_{l=1}^{\mathrm{rank}(G_{i}^{\dagger}(\vec{s})S)}\left(\frac{\lambda_{G_{i}^{\dagger}(\vec{s})S\Xi_{i}}^{(l)}}{\mathrm{E}\{\|\vec{\boldsymbol{s}}_{i}\|_{2}^{2}\}}\right)^{\Pr\{\boldsymbol{S}_{i}=S\}}}{ \prod_{l=1}^{\mathrm{rank}(S)}\left(\frac{\lambda_{S\Xi_{i}}^{(l)}}{\mathrm{E}\{\|\vec{\boldsymbol{s}}_{i}\|_{2}^{2}\}}\right)^{\Pr\{\boldsymbol{S}_{i}=S\}}}\right).\end{equation}
 In fact, 
 \begin{equation}
\label{ }
\mathsf{C}_{i}^{(\mathrm{lb})}(|h_{i,i}|^{2},\underline{\pi}_{i},\overline{\pi}_{i})\approx\mathsf{MG}\log\gamma-\mathsf{IEF}+\mathsf{CSF}_{i}.
\end{equation}
In general, $\mathsf{MG}$ does not depend on the user index. Also, assuming the channel gains are realizations of $\mathrm{i.i.d.}$ continuous random variables, the entropy $\mathrm{H}\left(\sum_{j\neq i}|h_{j,i}|^{2}\vec{\boldsymbol{s}}_{j}\vec{\boldsymbol{s}}_{j}^{\dagger}\right)$ is not a function of $i\in\{1,2,\cdots,n\}$, i.e., $\mathsf{IEF}$ does not depend on the user index either. In this case, a simple argument shows that 
\begin{equation}
\label{ }
\mathrm{H}\left(\sum_{j\neq i}|h_{j,i}|^{2}\vec{\boldsymbol{s}}_{j}\vec{\boldsymbol{s}}_{j}^{\dagger}\right)=(n-1)\mathrm{H}\left(\vec{\boldsymbol{s}}_{1}\vec{\boldsymbol{s}}_{1}^{\dagger}\right).\end{equation}

The interplay between $\mathsf{MG}$, $\mathsf{IEF}$ and $\mathsf{CSF}_{i}$ determines the behavior of the achievable rate. This behavior highly depends on the randomized algorithm in constructing the Signature Codes. As we will see in the next section,  a larger $\mathsf{MG}$ is usually achieved at the cost of a larger $\mathsf{IEF}$. It is clear that a larger $\mathsf{IEF}$ reduces the rate specially in moderate ranges of SNR. However, due to the fact that $\mathsf{MG}$ has also increased,  the rate is lifted up is the high SNR regime . These opposing effects identify a tradeoff between rate in moderate SNR and high SNR regime.  $\square$

\section{System Design} 
 In this section, we assume the channel gains $(h_{i,j})_{i,j=1}^{n}$ are realizations of independent $\mathcal{CN}(0,1)$ random variables $(\boldsymbol{h}_{i,j})_{i,j=1}^{n}$ representing Rayleigh fading. In the previous section, we have developed a lower bound 
 \begin{equation}\mathsf{C}_{i}^{(\mathrm{lb})}(\vec{\boldsymbol{h}}_{i})=\frac{1}{K}\sum_{\vec{s}\in\mathrm{supp}(\vec{\boldsymbol{s}}_{i}\backslash\{0_{K\times 1}\})}\Pr\{\vec{\boldsymbol{s}}_{i}=\vec{s}\}\log\left(2^{-\mathrm{H}\left(\sum_{j\neq i}|\boldsymbol{h}_{j,i}|^{2}\vec{\boldsymbol{s}}_{j}\vec{\boldsymbol{s}}_{j}^{\dagger}\right)}\boldsymbol{\varrho}_{i}(\gamma;\vec{s})+1\right)\end{equation} where
 \begin{equation}
\label{}
\boldsymbol{\varrho}_{i}(\gamma;\vec{s})\triangleq \frac{|\boldsymbol{h}_{i,i}|^{2}\|\vec{s}\|_{2}^{2}\gamma}{\mathrm{E}\{\|\vec{\boldsymbol{s}}_{i}\|_{2}^{2}\}}\prod_{\substack{S\in\mathrm{supp}(\boldsymbol{S}_{i})}}\left(\frac{\det\left(I_{K-1}+\beta^{2}\gamma G_{i}^{\dagger}(\vec{s})S\mathbf{\Xi}_{i}\mathbf{\Xi}_{i}^{\dagger}S^{\dagger}G_{i}^{\dagger}(\vec{s})\right)}{\det\left(I_{K}+\beta^{2}\gamma S\mathbf{\Xi}_{i}\mathbf{\Xi}_{i}^{\dagger}S^{\dagger}\right)}\right)^{\Pr\{\boldsymbol{S}_{i}=S\}},\end{equation} 
and 
\begin{equation}
\label{ }
\mathbf{\Xi}_{i}=\mathrm{diag}\left(\boldsymbol{h}_{1,i},\cdots,\boldsymbol{h}_{i-1,i},\boldsymbol{h}_{i+1,i},\cdots,\boldsymbol{h}_{n,i}\right).
\end{equation}
 The global design criteria is to choose $K$, $(\mathsf{p}_{a})_{a\in\mathscr{A}}$ and $\varepsilon$ based on
\begin{equation}
\label{ }
(\hat{K},(\hat{\mathsf{p}}_{a})_{a\in\mathscr{A}}, \hat{\varepsilon})=\arg\sup_{K,(\mathsf{p}_{a})_{a\in\mathscr{A}},\varepsilon}\mathrm{E}\left\{\mathsf{C}_{i}^{(\mathrm{lb})}(\vec{\boldsymbol{h}}_{i})\right\}.\end{equation}

\textit{Example 3-}  Let us consider a network with $n=2$ users. For $i\in\{1,2\}$, we define 
\begin{equation}
\label{ }
i'=\left\{\begin{array}{cc}
    2  & i=1   \\
     1& i=2  
\end{array}\right..
\end{equation}
In this case, we have

\textit{1-} $\mathbf{\Xi}_{i}=\boldsymbol{h}_{i',i}$.

\textit{2-} Since $\boldsymbol{S}_{i}=\vec{\boldsymbol{s}}_{i'}$, for each $\vec{t}\in\mathrm{supp}(\boldsymbol{S}_{i})\backslash\{0_{K\times 1}\}$, we have $\mathrm{rank}(\boldsymbol{h}_{i',i}\vec{t})= 1$ and $\lambda_{\boldsymbol{h}_{i',i}\vec{t}}^{(1)}=|\boldsymbol{h}_{i',i}|^{2}\|\vec{t}\|_{2}^{2}$. 

\textit{3-} For each $\vec{s}\in\mathrm{supp}(\vec{\boldsymbol{s}}_{i})$ and $\vec{t}\in\mathrm{supp}(\boldsymbol{S}_{i})$, we have $\mathrm{rank}(\boldsymbol{h}_{i',i}G_{i}^{\dagger}(\vec{s})\vec{t})\leq 1$. Indeed, $G_{i}^{\dagger}(\vec{s})\vec{t}\in\mathbb{R}$ and if $G_{i}^{\dagger}(\vec{s})\vec{t}\neq 0$, then $\lambda_{\boldsymbol{h}_{i',i}G_{i}^{\dagger}(\vec{s})\vec{t}}^{(1)}=|\boldsymbol{h}_{i',i}|^{2}|G_{i}^{\dagger}(\vec{s})\vec{t}|^{2}$.

Therefore, $\boldsymbol{\varrho}_{i}(\gamma;\vec{s})$ can be written as
\begin{equation}
\label{lol}
\boldsymbol{\varrho}_{i}(\gamma;\vec{s})= \frac{|\boldsymbol{h}_{i,i}|^{2}\|\vec{s}\|_{2}^{2}\gamma}{\mathrm{E}\{\|\vec{\boldsymbol{s}}_{i}\|^{2}_{2}\}}\prod_{\substack{\vec{t}\in\mathrm{supp}(\vec{\boldsymbol{s}}_{i'})}}\left(\frac{1+\frac{ |\boldsymbol{h}_{i',i}|^{2}|G_{i}^{\dagger}(\vec{s})\vec{t}|^{2}\gamma}{\mathrm{E}\{\|\vec{\boldsymbol{s}}_{i}\|^{2}_{2}\}}}{1+\frac{|\boldsymbol{h}_{i',i}|^{2}\|\vec{t}\|_{2}^{2}\gamma}{\mathrm{E}\{\|\vec{\boldsymbol{s}}_{i}\|^{2}_{2}\}}}\right)^{\Pr\{\vec{\boldsymbol{s}}_{i'}=\vec{t}\}}.\end{equation}

 \textbf{Scheme A-} Let $K=2$ and $\mathscr{A}=\{-1,1\}$ with $\mathsf{p}_{1}=\nu$ and $\mathsf{p}_{-1}=\overline{\nu}$ for some $\nu\in(0,1]$. To simplify the expression for $\boldsymbol{\varrho}_{i}(\gamma;\vec{s})$ in (\ref{lol}), we make the following observations:

\textit{1-} If $\vec{s}$ has only one nonzero element, then 
\begin{equation}
\label{ }
|G_{i}^{\dagger}(\vec{s})\vec{t}|^{2}=\left\{\begin{array}{cc}
    0  & \textrm{$\vec{t}=0_{2\times 1}$ or $\vec{t}=\pm\vec{s}$}   \\
    1  &  \mathrm{oth.} 
\end{array}\right..
\end{equation}

\textit{2-} If $\vec{s}$ has no zero elements, then
\begin{equation}
\label{ }
|G_{i}^{\dagger}(\vec{s})\vec{t}|^{2}=\left\{\begin{array}{cc}
    0  &    \textrm{$\vec{t}=0_{2\times 1}$ or $\vec{t}=\pm\vec{s}$}  \\
    2  &  \vec{t}^{T}\vec{s}=0, \vec{t}\ne0_{2\times 1}\\
    \frac{1}{2} &  \vec{t}^{T}\vec{s}\neq0\end{array}\right..
\end{equation}
As such, it is easy to see that

\textit{1-} If $\vec{s}$ has only one nonzero element, then 
\begin{equation}
\label{ }
\boldsymbol{\varrho}_{i}(\gamma;\vec{s})=\frac{|\boldsymbol{h}_{i,i}|^{2}\gamma}{2\varepsilon\left(1+\frac{|\boldsymbol{h}_{i',i}|^{2}\gamma}{2\varepsilon}\right)^{\varepsilon\overline{\varepsilon}-\varepsilon^{2}}\left(1+\frac{|\boldsymbol{h}_{i',i}|^{2}\gamma}{\varepsilon}\right)^{\varepsilon^{2}}}.
\end{equation}
\textit{2-} If $\vec{s}\in\{(1,1)^{T},(-1,-1)^{T}\}$, then
\begin{equation}
\label{ }
\boldsymbol{\varrho}_{i}(\gamma;\vec{s})=\frac{|\boldsymbol{h}_{i,i}|^{2}\gamma\left(1+\frac{|\boldsymbol{h}_{i',i}|^{2}\gamma}{4\varepsilon}\right)^{2\varepsilon\overline{\varepsilon}}}{\varepsilon\left(1+\frac{|\boldsymbol{h}_{i',i}|^{2}\gamma}{2\varepsilon}\right)^{2\varepsilon\overline{\varepsilon}}\left(1+\frac{|\boldsymbol{h}_{i',i}|^{2}\gamma}{\varepsilon}\right)^{\varepsilon^{2}(\nu^{2}+\overline{\nu}^{2})}}.
\end{equation}
\textit{3-} If $\vec{s}\in\{(1,-1)^{T},(-1,1)^{T}\}$, then
\begin{equation}
\label{ }
\boldsymbol{\varrho}_{i}(\gamma;\vec{s})=\frac{|\boldsymbol{h}_{i,i}|^{2}\gamma\left(1+\frac{|\boldsymbol{h}_{i',i}|^{2}\gamma}{4\varepsilon}\right)^{2\varepsilon\overline{\varepsilon}}}{\varepsilon\left(1+\frac{|\boldsymbol{h}_{i',i}|^{2}\gamma}{2\varepsilon}\right)^{2\varepsilon\overline{\varepsilon}}\left(1+\frac{|\boldsymbol{h}_{i',i}|^{2}\gamma}{\varepsilon}\right)^{2\varepsilon^{2}\nu\overline{\nu}}}.
\end{equation}
Finally, it is shown in appendix A that 
\begin{equation}
\label{ }
\mathrm{H}\left(\vec{\boldsymbol{s}}_{1}\vec{\boldsymbol{s}}_{1}^{\dagger}\right)=2\mathscr{H}(\varepsilon)+\varepsilon^{2}\mathscr{H}(\nu^{2}+\overline{\nu}^{2})\end{equation}

Simulation results indicate that $\mathsf{C}_{i}^{(\mathrm{lb})}(\vec{\boldsymbol{h}}_{i})$ is maximized at $\nu=\frac{1}{2}$. Setting $\nu=\frac{1}{2}$, 
\begin{eqnarray}
\label{lol34}
\mathsf{C}_{i}^{(\mathrm{lb})}(\vec{\boldsymbol{h}}_{i})&=&\varepsilon\overline{\varepsilon}\log\left(1+\frac{2^{-2\mathscr{H}(\varepsilon)-\varepsilon^{2}}|\boldsymbol{h}_{i,i}|^{2}\gamma}{2\varepsilon\left(1+\frac{|\boldsymbol{h}_{i',i}|^{2}\gamma}{2\varepsilon}\right)^{\varepsilon\overline{\varepsilon}-\varepsilon^{2}}\left(1+\frac{|\boldsymbol{h}_{i',i}|^{2}\gamma}{\varepsilon}\right)^{\varepsilon^{2}}}\right)\notag\\
&&+\frac{\varepsilon^{2}}{2}\log\left(1+\frac{2^{-2\mathscr{H}(\varepsilon)-\varepsilon^{2}}|\boldsymbol{h}_{i,i}|^{2}\gamma\left(1+\frac{|\boldsymbol{h}_{i',i}|^{2}\gamma}{4\varepsilon}\right)^{2\varepsilon\overline{\varepsilon}}}{\varepsilon\left(1+\frac{|\boldsymbol{h}_{i',i}|^{2}\gamma}{2\varepsilon}\right)^{2\varepsilon\overline{\varepsilon}}\left(1+\frac{|\boldsymbol{h}_{i',i}|^{2}\gamma}{\varepsilon}\right)^{\frac{\varepsilon^{2}}{2}}}\right). \notag\\
\end{eqnarray}
It is also evident that
\begin{equation}
\label{ }
\mathsf{MG}_{\textrm{scheme A}}=\frac{1}{2}-\frac{\overline{\varepsilon}^{2}}{2}-(\varepsilon\overline{\varepsilon})^{2}-\frac{\varepsilon^{4}}{4}
\end{equation}
and 
\begin{equation}
\label{ }
\mathsf{IEF}_{\textrm{scheme A}}=\varepsilon\left(\mathscr{H}(\varepsilon)+\frac{\varepsilon^{2}}{2}\right).
\end{equation}
 
 \textbf{Scheme B-} Assuming no spreading is performed, let $K=1$. Noting the fact that $\mathrm{supp}(\vec{\boldsymbol{s}}_{i})=\mathrm{supp}(\vec{\boldsymbol{s}}_{i'})=\{0,1\}$ and $G_{i}(1)=0$, 
 \begin{equation}
\label{ }
\boldsymbol{\varrho}_{i}(\gamma;1)=\frac{|\boldsymbol{h}_{i,i}|\gamma}{\varepsilon\left(1+\frac{|\boldsymbol{h}_{i',i}|^{2}\gamma}{\varepsilon}\right)^{\varepsilon}}.
\end{equation}
It is easily seen that  $\mathrm{H}(\vec{\boldsymbol{s}}_{1}\vec{\boldsymbol{s}}_{1}^{\dagger})=\mathscr{H}(\varepsilon)$. Therefore, 
\begin{equation}
\label{ }
\mathsf{C}_{i}^{(\mathrm{lb})}(\vec{\boldsymbol{h}}_{i})=\varepsilon\log\left(1+\frac{2^{-\mathscr{H}(\varepsilon)}|\boldsymbol{h}_{i,i}|^{2}\gamma}{\varepsilon\left(1+\frac{|\boldsymbol{h}_{i',i}|^{2}\gamma}{\varepsilon}\right)^{\varepsilon}}\right),
\end{equation}
\begin{equation}
\label{ }
\mathsf{MG}_{\textrm{scheme B}}=\varepsilon\overline{\varepsilon}
\end{equation}
and 
\begin{equation}
\label{ }
\mathsf{IEF}_{\textrm{scheme B}}=\varepsilon\mathscr{H}(\varepsilon).
\end{equation}
Fig. \ref{f5} sketches $\sup_{\varepsilon}\mathrm{E}\left\{\mathsf{C}_{i}^{(\mathrm{lb})}(\vec{\boldsymbol{h}}_{i})\right\}$ for the schemes A and B. It is seen that there is a tradeoff between the rates at medium  and high SNR values. Fig. \ref{f8} demonstrates the best $\varepsilon$ chosen by the users. It is seen that any user in both schemes starts with $\varepsilon=1$ at $\gamma=5\mathrm{dB}$. Selecting $\varepsilon=1$ in scheme B leads to $\mathsf{IEF}_{\textrm{scheme B}}=0$. However, $\mathsf{MG}_{\textrm{scheme B}}$ is kept at zero as well. Therefore, by increasing SNR, the average achievable rate starts to saturate, and hence, users switch to $\varepsilon=0.45$ for $\gamma>20\mathrm{dB}$ to avoid saturation. In scheme A, $\varepsilon$ is set at $1$ for SNR values up to $35\mathrm{dB}$. The yields $\mathsf{IEF}_{\textrm{scheme A}}=\frac{1}{2}$ which can be considered as a reason for poor performance of scheme A in the range $\gamma<15\mathrm{dB}$ compared to scheme B. Since $\mathsf{MG}_{\textrm{scheme A}}\Big|_{\varepsilon=1}=0.25$ is larger than $\mathsf{MG}_{\textrm{scheme A}}\Big|_{\varepsilon=0.45}$, the average achievable rate per user becomes eventually larger in scheme A compared to scheme B as SNR increases.  
\begin{figure}[h!b!t]
  \centering
  \includegraphics[scale=.7] {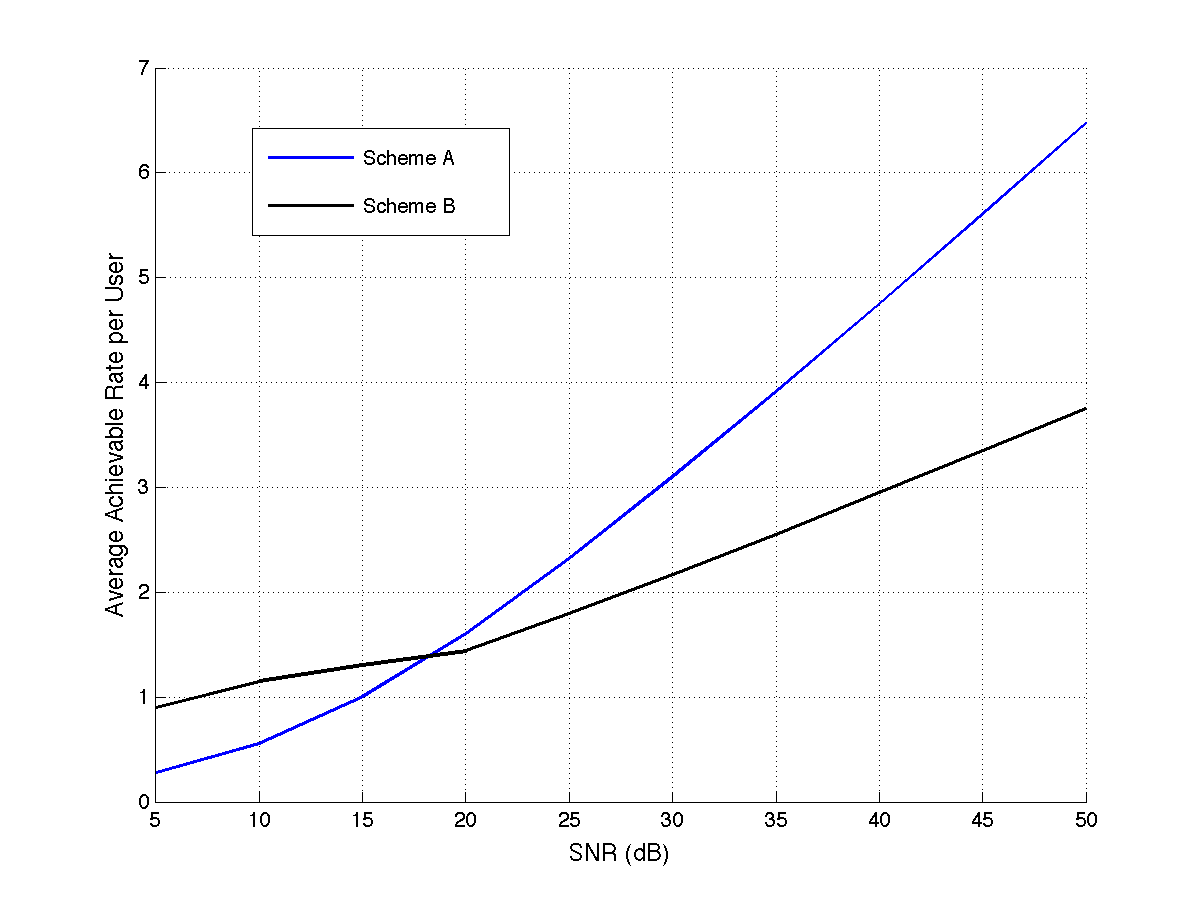}
 \caption{Comparison between $\sup_{\varepsilon}\mathrm{E}\left\{\mathsf{C}_{i}^{(\mathrm{lb})}(\vec{\boldsymbol{h}}_{i})\right\}$ in schemes A and B for different SNR values.}
  \label{f5}
 \end{figure} 
 \begin{figure}[h!b!t]
  \centering
  \includegraphics[scale=.7] {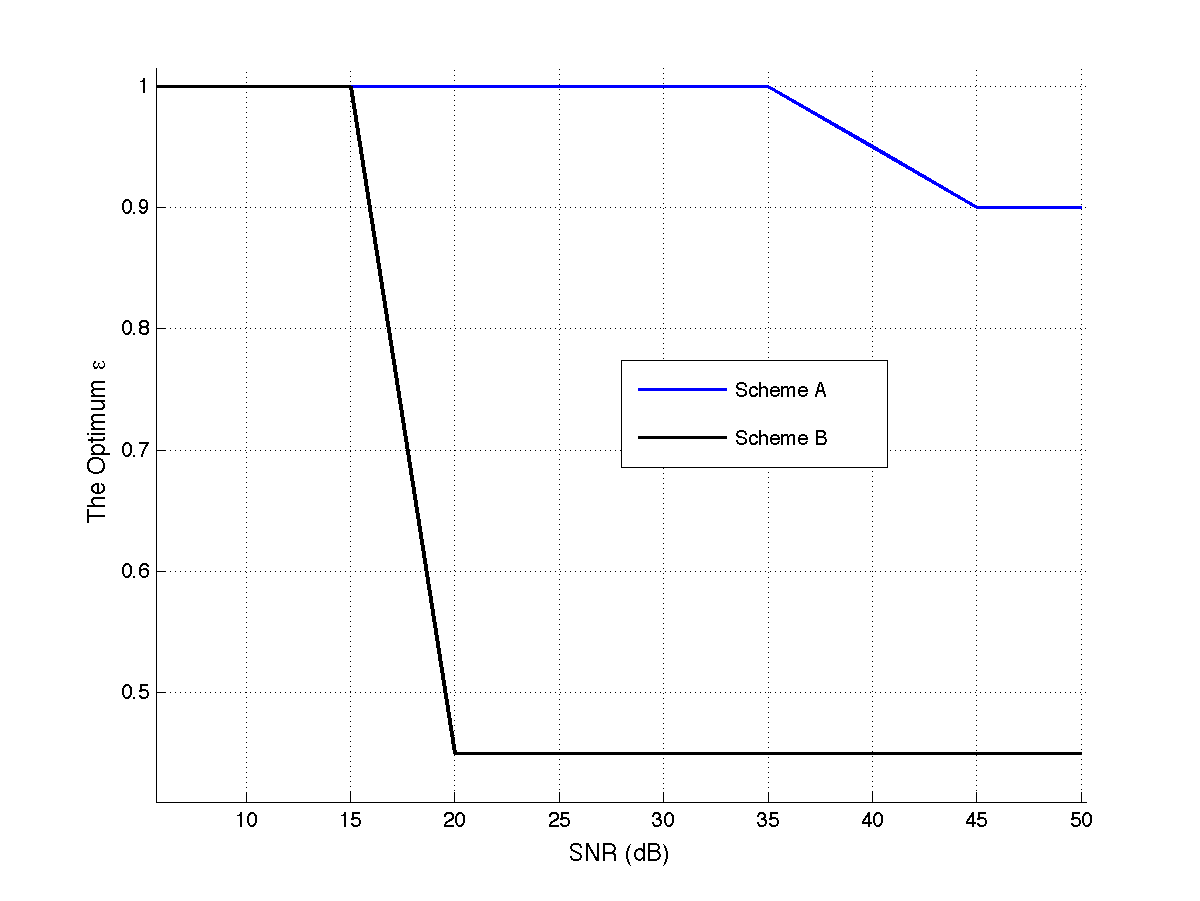}
 \caption{Comparison between $\sup_{\varepsilon}\mathrm{E}\left\{\mathsf{C}_{i}^{(\mathrm{lb})}(\vec{\boldsymbol{h}}_{i})\right\}$ in schemes A and B for different SNR values.}
  \label{f8}
 \end{figure} 
 \begin{figure}[h!b!t]
  \centering
  \includegraphics[scale=.7] {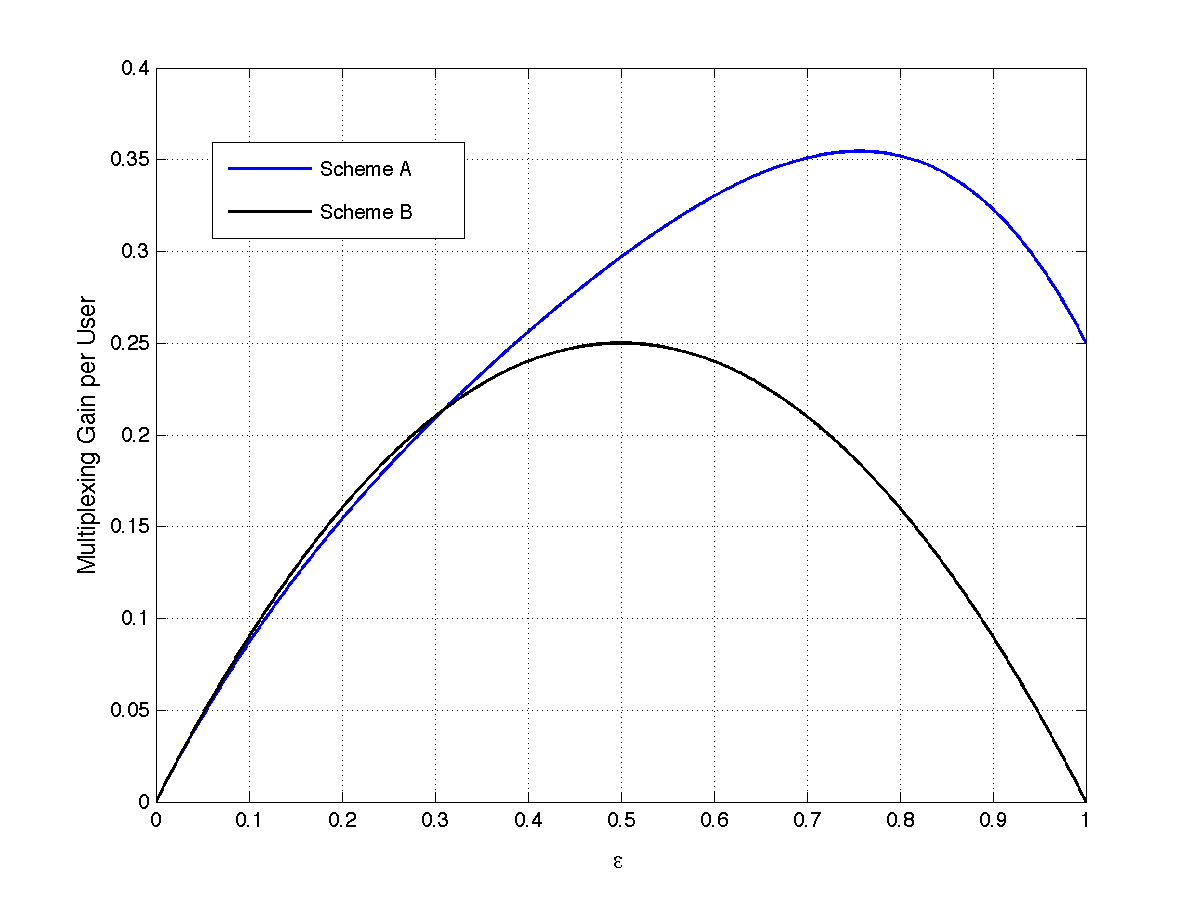}
 \caption{Comparison between $\mathsf{MG}$ in schemes A and B in terms of $\varepsilon$.}
  \label{f6}
 \end{figure}
  \begin{figure}[h!b!t]
  \centering
  \includegraphics[scale=.7] {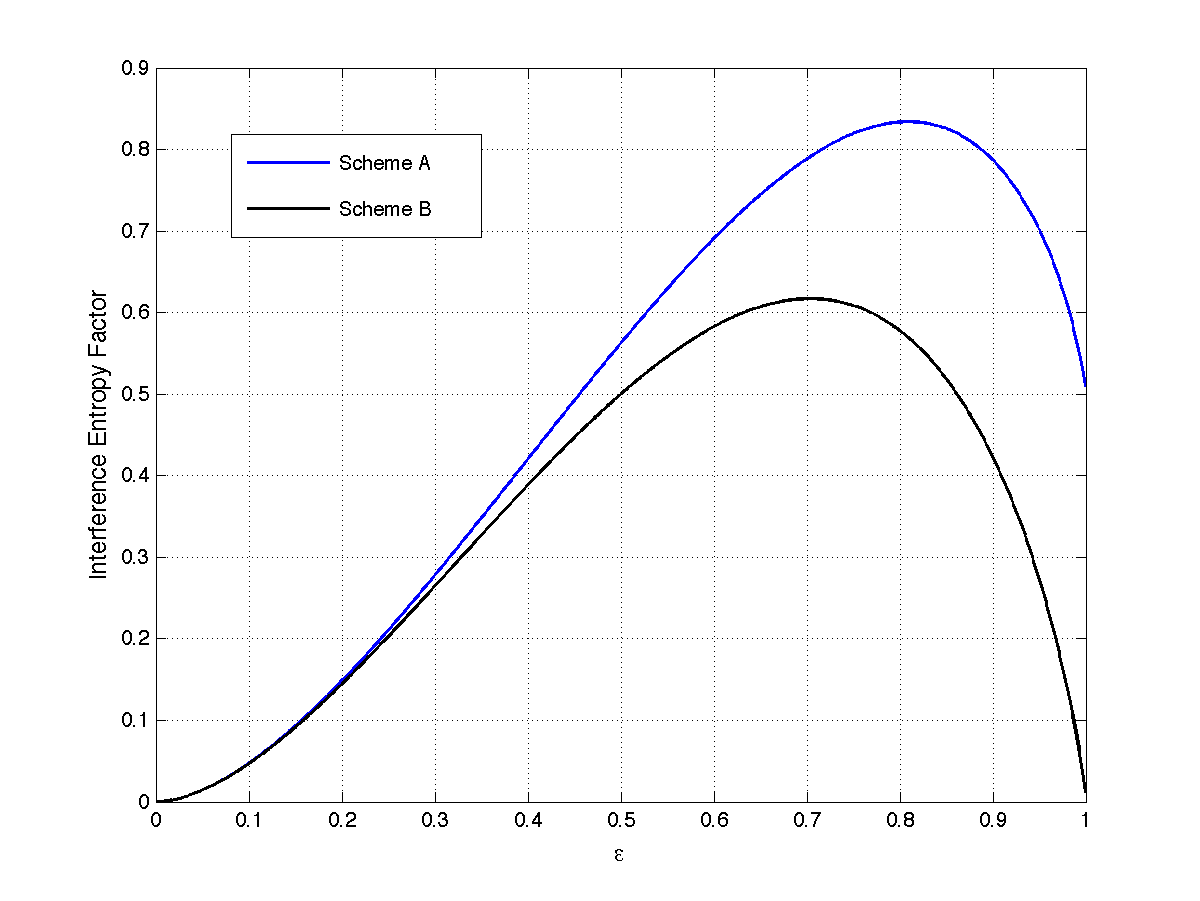}
 \caption{Comparison between $\mathsf{IEF}$ in schemes A and B in terms of $\varepsilon$.}
  \label{f7}
 \end{figure} 
 
 \textit{Example 4-} We consider a decentralized network of  $n>2$ users. We present the following scenarios:


 The signature sequence of any user consists of an spreading code over the alphabet $\{-1,1\}$ where $\mathsf{p}_{1}=\nu$ and $\mathsf{p}_{-1}=\overline{\nu}$, i.e., masking is not applied. The purpose of this example is to show that in contrast to example 3, the optimum value of $\nu$ is not necessarily $\frac{1}{2}$.

 Before proceeding, let us explain why the common intuition is to set $\nu=\frac{1}{2}$.  It is well-known that in an additive noise channel with a stationary noise process, as far as the correlation function\footnote{The correlation function of a zero-mean process $\boldsymbol{\mathsf{x}}[t]$ is the function $\mathrm{E}\{\boldsymbol{\mathsf{x}}[t]\boldsymbol{\mathsf{x}}^{\dagger}[t-\Delta t]\}$ for $\Delta t\in\mathbb{R}$. } is fixed, a stationary Gaussian noise process yields the least mutual information between the input and output. WE call this the \emph{Gaussian} bounding technique. Using this fact, one can obtain a lower bound on $\mathrm{I}(\boldsymbol{x}_{i};\vec{\boldsymbol{y}}_{i}|\vec{\boldsymbol{s}}_{i})$ as 
 \begin{equation}
\label{ }
\mathrm{I}(\boldsymbol{x}_{i};\vec{\boldsymbol{y}}_{i}|\vec{\boldsymbol{s}}_{i})\geq \log\frac{\det\mathrm{Cov}(\vec{\boldsymbol{y}}_{i})}{\det\mathrm{Cov}(\vec{\boldsymbol{w}}_{i}+\vec{\boldsymbol{z}}_{i})}.
\end{equation}
It is easy to see that 
\begin{equation}
\label{ }
\mathrm{Cov}(\vec{\boldsymbol{y}}_{i})=I_{K}+\frac{\gamma}{K}\sum_{j=1}^{n}|\boldsymbol{h}_{j,i}|^{2}\left((1-(2\nu-1)^{2})I_{K}+(2\nu-1)^{2}1_{K\times K}\right)
\end{equation} 
and 
\begin{equation}
\label{ }
\mathrm{Cov}(\vec{\boldsymbol{w}}_{i}+\vec{\boldsymbol{z}}_{i})=I_{K}+\frac{\gamma}{K}\sum_{j\neq i}|\boldsymbol{h}_{j,i}|^{2}\left((1-(2\nu-1)^{2})I_{K}+(2\nu-1)^{2}1_{K\times K}\right).
\end{equation}  
Therefore\footnote{Note that $1_{K\times K}=1_{K\times 1}1^{\mathrm{T}}_{K\times 1}$. Then, one can use the identity $\det(I_{m_{1}}+AB)=\det(I_{m_{2}}+BA)$ for any $m_{1}\times m_{2}$ and $m_{2}\times m_{1}$ matrices $A$ and $B$.}, 
\begin{equation}
\label{ }
\det\mathrm{Cov}(\vec{\boldsymbol{y}}_{i})=\left(1+\frac{(1-(2\nu-1)^{2})\gamma\sum_{j=1}^{n}|\boldsymbol{h}_{j,i}|^{2}}{K}\right)^{K}\left(1+\frac{(2\nu-1)^{2}\gamma\sum_{j=1}^{n}|\boldsymbol{h}_{j,i}|^{2}}{1+\frac{(1-(2\nu-1)^{2})\gamma\sum_{j=1}^{n}|\boldsymbol{h}_{j,i}|^{2}}{K}}\right)
\end{equation}
and 
\begin{equation}
\label{ }
\det\mathrm{Cov}(\vec{\boldsymbol{w}}_{i}+\vec{\boldsymbol{z}}_{i})=\left(1+\frac{(1-(2\nu-1)^{2})\gamma\sum_{j\neq i}|\boldsymbol{h}_{j,i}|^{2}}{K}\right)^{K}\left(1+\frac{(2\nu-1)^{2}\gamma\sum_{j\neq i}|\boldsymbol{h}_{j,i}|^{2}}{1+\frac{(1-(2\nu-1)^{2})\gamma\sum_{j\neq i}|\boldsymbol{h}_{j,i}|^{2}}{K}}\right).
\end{equation}
Finally, we come up with the following lower bound on $\frac{\mathrm{I}(\boldsymbol{x}_{i};\vec{\boldsymbol{y}}_{i}|\vec{\boldsymbol{s}}_{i})}{K}$,
\begin{eqnarray}
\frac{\mathrm{I}(\boldsymbol{x}_{i};\vec{\boldsymbol{y}}_{i}|\vec{\boldsymbol{s}}_{i})}{K}\geq\log\left(1+\frac{\frac{(1-(2\nu-1)^{2})\gamma|\boldsymbol{h}_{i,i}|^{2}}{K}}{1+\frac{(1-(2\nu-1)^{2})\gamma\sum_{j\neq i}|\boldsymbol{h}_{j,i}|^{2}}{K}}\right)+\frac{1}{K}\log\frac{1+\frac{(2\nu-1)^{2}\gamma\sum_{j=1}^{n}|\boldsymbol{h}_{j,i}|^{2}}{1+\frac{(1-(2\nu-1)^{2})\gamma\sum_{j=1}^{n}|\boldsymbol{h}_{j,i}|^{2}}{K}}}{1+\frac{(2\nu-1)^{2}\gamma\sum_{j\neq i}|\boldsymbol{h}_{j,i}|^{2}}{1+\frac{(1-(2\nu-1)^{2})\gamma\sum_{j\neq i}|\boldsymbol{h}_{j,i}|^{2}}{K}}}.
\end{eqnarray}
It is straightforward to see that this lower bound is maximized at $K=1$ and $\nu=\frac{1}{2}$ for any realization of the channel gains.  Hence, 
\begin{equation}
\label{ }
\sup_{\nu, K}\frac{\mathrm{I}(\boldsymbol{x}_{i};\vec{\boldsymbol{y}}_{i}|\vec{\boldsymbol{s}}_{i})}{K}\geq\log\left(1+\frac{\gamma |\boldsymbol{h}_{i,i}|^{2}}{1+\gamma\sum_{j\neq i}|\boldsymbol{h}_{j,i}|^{2}}\right).\end{equation}

Although, this lower bound suggests to set $K=1$ and in case $K>1$, it requires $\nu=\frac{1}{2}$, we demonstrate that  taking a $K>1$ and regulating at some $\nu\neq \frac{1}{2}$ yield achievable rates larger than the threshold 
 \begin{eqnarray}
 \label{wind}
 \tau_{n}&\triangleq& \sup_{\gamma}\mathrm{E}\left\{\log\left(1+\frac{\gamma |\boldsymbol{h}_{i,i}|^{2}}{1+\gamma\sum_{j\neq i}|\boldsymbol{h}_{j,i}|^{2}}\right)\right\}\notag\\
 &=&\frac{1}{(n-2)!}\int_{\zeta\in\mathbb{R}}\int_{\eta\in\mathbb{R}}\eta^{n-2}\log\left(1+\frac{\zeta}{\eta}\right)e^{-\zeta-\eta}d\zeta d\eta.\end{eqnarray}
 In fact, $\tau_{n}$ is the maximum average achievable rate by regulating the transmission rate of the $i^{th}$ user at $\log\left(1+\frac{\gamma |\boldsymbol{h}_{i,i}|^{2}}{1+\gamma\sum_{j\neq i}|\boldsymbol{h}_{j,i}|^{2}}\right)$. In (\ref{wind}), we have used the fact that  $|\boldsymbol{h}_{i,i}|^{2}$ is an exponential random variable with parameter $1$ and $2\sum_{j\neq i}|\boldsymbol{h}_{j,i}|^{2}$  is a $\chi^{2}_{2(n-1)}$ random variable.  
 
 Let $n=4$. In this case, $\tau_{4}=0.4809$. To compute $\mathsf{C}_{i}^{(\mathrm{lb})}(\vec{\boldsymbol{h}}_{i})$, we notice that 
 
 \textit{1-} For any $\vec{s}\in\mathrm{supp}(\vec{\boldsymbol{s}}_{i})$, $\|\vec{s}_{i}\|_{2}^{2}=K$.
 
 \textit{2-} In appendix B, it is shown that 
 \begin{equation}
\label{ }
\mathrm{H}\left(\vec{\boldsymbol{s}}_{1}\vec{\boldsymbol{s}}_{1}^{\dagger}\right)=-\sum_{k=0}^{K}{K\choose k}\left(\nu^{k+1}\overline{\nu}^{K-k}+\nu^{K-k}\overline{\nu}^{k+1}\right)\log\left(\nu^{k+1}\overline{\nu}^{K-k}+\nu^{K-k}\overline{\nu}^{k+1}\right).\end{equation}

In contrast to example 3, computing $\mathsf{C}_{i}^{(\mathrm{lb})}(\vec{\boldsymbol{h}}_{i})$ in closed form is a tedious task. As such, we calculate $\mathrm{E}\left\{\mathsf{C}_{i}^{(\mathrm{lb})}(\vec{\boldsymbol{h}}_{i})\right\}$ through simulations. Setting the SNR at $\gamma=60\mathrm{dB}$, fig. \ref{f77} sketches $\mathrm{E}\left\{\mathsf{C}_{i}^{(\mathrm{lb})}(\vec{\boldsymbol{h}}_{i})\right\}$ in terms of $\mathsf{p}_{1}=\nu$ for different values of $K$. In spite of one's intuition, the average achievable rate per user has a double-hump shape and is not maximized at $\nu=\frac{1}{2}$. It is seen that the best performance is obtained at $K=3$ and $\nu\in\{0.09,0.91\}$. $\square$
  \begin{figure}[h!b!t]
  \centering
  \includegraphics[scale=.7] {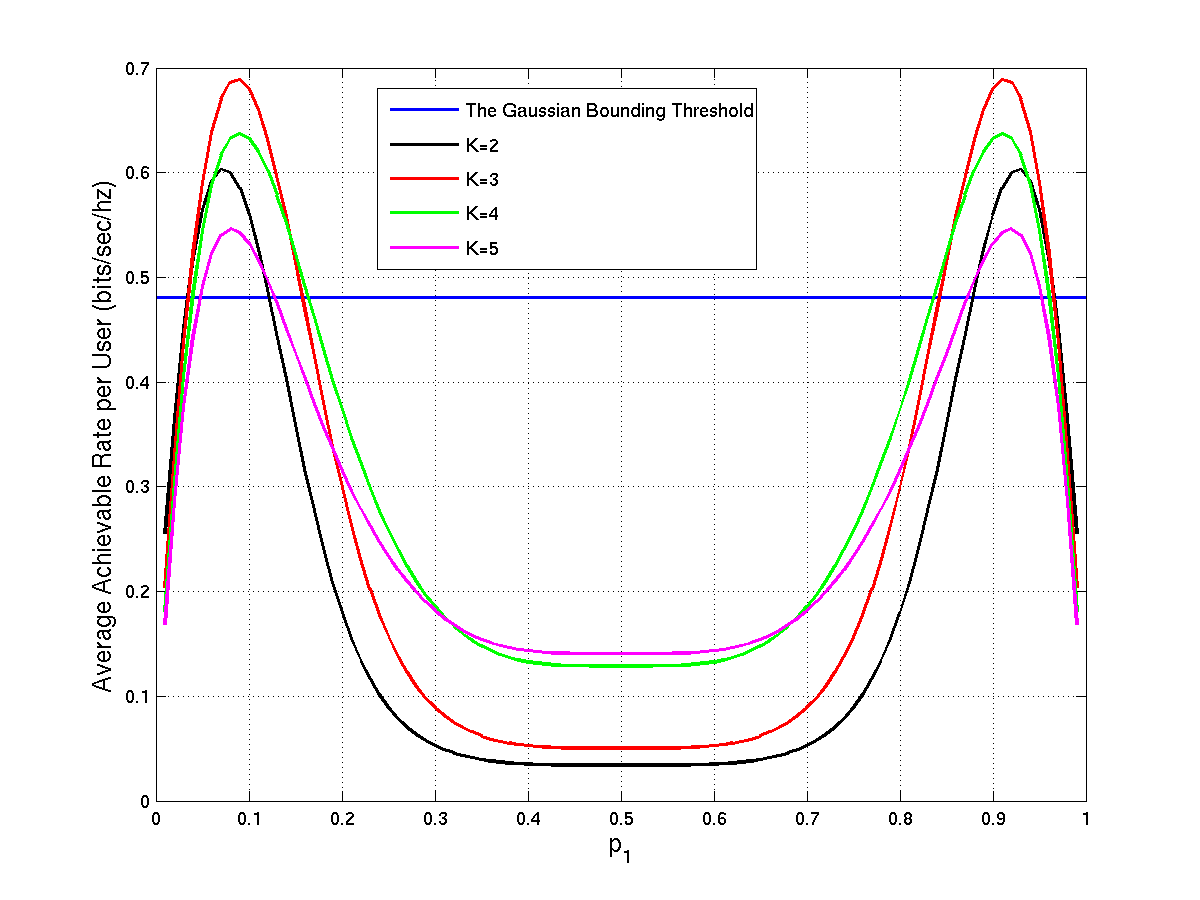}
 \caption{Comparison between $\mathsf{IEF}$ in schemes A and B in terms of $\varepsilon$.}
  \label{f77}
 \end{figure} 
 
 \textit{Remark 2-} To gain some insight on why $\mathrm{E}\left\{\mathsf{C}_{i}^{(\mathrm{lb})}(\vec{\boldsymbol{h}}_{i})\right\}$ is double-hump in example 4, one can study the multiplexing gain per user given in (\ref{MG}). Let us consider a network with $n\geq 4$ users\footnote{It can be shown that this phenomenon does not hold for $n=2,3$.} where the signature codes only consist of spreading over the alphabet $\mathscr{A}=\{-1,1\}$. In general, one can write
 \begin{equation}
\label{ }
\Big\{\vec{\boldsymbol{s}}_{i}\in\mathrm{csp}(\boldsymbol{S}_{i})\Big\}=\Big\{\vec{\boldsymbol{s}}_{i}\in\mathrm{col}(\boldsymbol{S}_{i})\cup\mathrm{col}(-\boldsymbol{S}_{i})\Big\}\bigcup\Big\{\vec{\boldsymbol{s}}_{i}\in\mathrm{csp}(\boldsymbol{S}_{i})\backslash(\mathrm{col}(\boldsymbol{S}_{i})\cup\mathrm{col}(-\boldsymbol{S}_{i}))\Big\}.
\end{equation}
Therefore,
\begin{eqnarray}
\Pr\left\{\vec{\boldsymbol{s}}_{i}\notin\mathrm{csp}(\boldsymbol{S}_{i})\right\}&=&1-\Pr\left\{\vec{\boldsymbol{s}}_{i}\in\mathrm{col}(\boldsymbol{S}_{i})\cup\mathrm{col}(-\boldsymbol{S}_{i})\right\}\notag\\&&-\Pr\left\{\vec{\boldsymbol{s}}_{i}\in\mathrm{csp}(\boldsymbol{S}_{i})\backslash(\mathrm{col}(\boldsymbol{S}_{i})\cup\mathrm{col}(-\boldsymbol{S}_{i}))\right\}\notag\\
&=&\Pr\left\{\vec{\boldsymbol{s}}_{i}\notin\mathrm{col}(\boldsymbol{S}_{i})\cup\mathrm{col}(-\boldsymbol{S}_{i})\right\}-\Pr\left\{\vec{\boldsymbol{s}}_{i}\in\mathrm{csp}(\boldsymbol{S}_{i})\backslash(\mathrm{col}(\boldsymbol{S}_{i})\cup\mathrm{col}(-\boldsymbol{S}_{i}))\right\}.\notag\\\end{eqnarray}
The term $\Pr\left\{\vec{\boldsymbol{s}}_{i}\notin\mathrm{col}(\boldsymbol{S}_{i})\cup\mathrm{col}(-\boldsymbol{S}_{i})\right\}$ can be easily calculated as 
\begin{equation}
\label{ }
\Pr\left\{\vec{\boldsymbol{s}}_{i}\notin\mathrm{col}(\boldsymbol{S}_{i})\cup\mathrm{col}(-\boldsymbol{S}_{i})\right\}=\sum_{k=0}^{K}{K\choose k}\nu^{k}\bar{\nu}^{K-k}\left(1-\nu^{k}\bar{\nu}^{K-k}-\bar{\nu}^{k}\nu^{K-k}\right)^{n-1}.
\end{equation}
On the other hand, computation of the term $\Pr\left\{\vec{\boldsymbol{s}}_{i}\in\mathrm{csp}(\boldsymbol{S}_{i})\backslash(\mathrm{col}(\boldsymbol{S}_{i})\cup\mathrm{col}(-\boldsymbol{S}_{i}))\right\}$ is not an easy task. However, the point is that both $\Pr\left\{\vec{\boldsymbol{s}}_{i}\notin\mathrm{col}(\boldsymbol{S}_{i})\cup\mathrm{col}(-\boldsymbol{S}_{i})\right\}$ and $\Pr\left\{\vec{\boldsymbol{s}}_{i}\in\mathrm{csp}(\boldsymbol{S}_{i})\backslash(\mathrm{col}(\boldsymbol{S}_{i})\cup\mathrm{col}(-\boldsymbol{S}_{i}))\right\}$ have a global maximum at $\nu=\frac{1}{2}$. Hence, there is a chance that their difference is maximized at some $\nu\neq \frac{1}{2}$. This is exactly what happens here.  As an example, fig. \ref{f77c} sketches multiplexing gain per user in terms of $\mathsf{p}_{1}$ in a network with $n=10$ users. It is assumed that the spreading code length is $K=6$. 
 \begin{figure}[h!b!t]
  \centering
  \includegraphics[scale=.7] {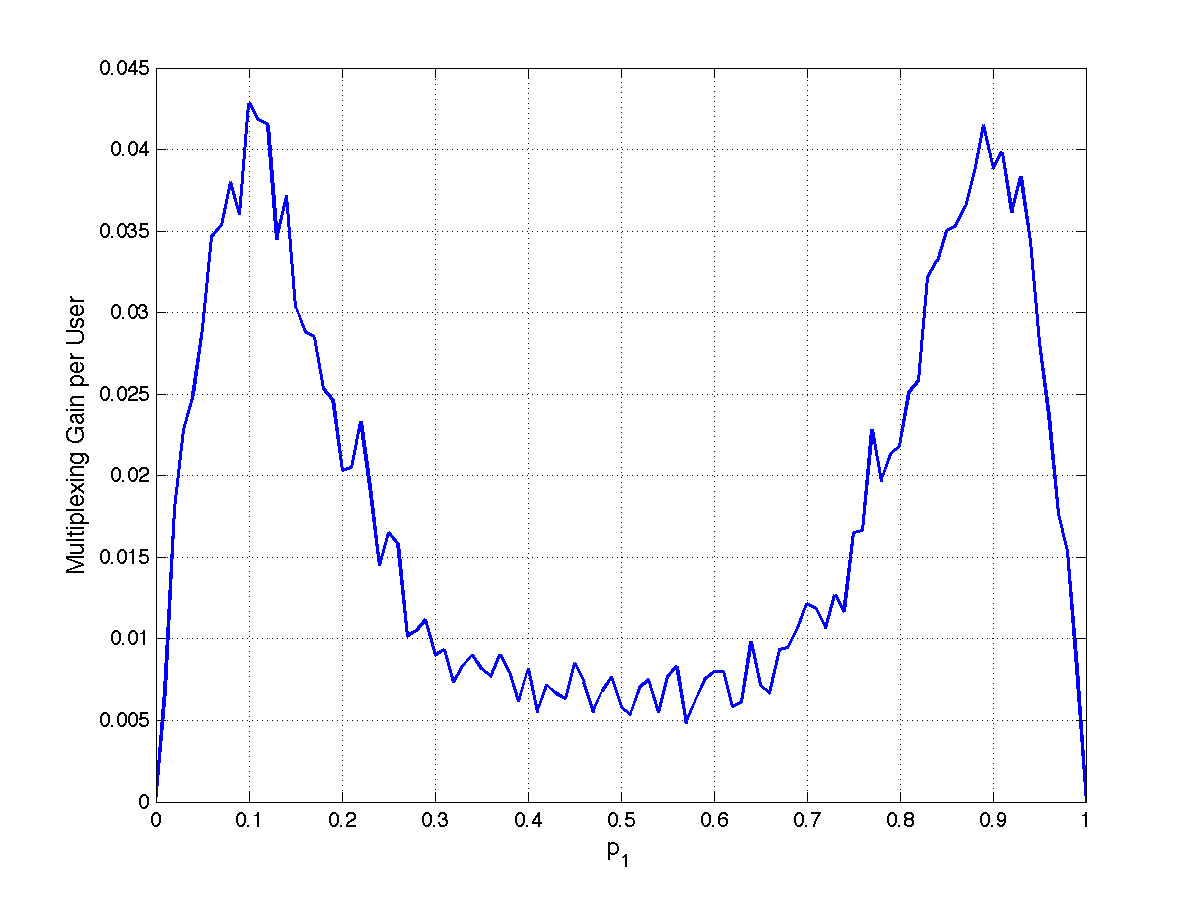}
 \caption{Sketch of $\frac{\Pr\{\vec{\boldsymbol{s}}_{i}\notin\mathrm{csp}(\boldsymbol{S}_{i})\}}{K}$ in terms of $\mathsf{p}_{1}=\nu$ in a network of $n=10$ users with $K=6$.}
  \label{f77c}
 \end{figure} 
 
 \textit{Remark 3-} The expression for the SMG given in (\ref{canada}) does not depend on the spreading/masking strategy. In fact, one can consider a more general scheme where the $i^{th}$ user randomly selects its code $\vec{\boldsymbol{s}}_{i}$ out of a globally known set of codes $\mathfrak{C}\subset \mathbb{R}^{K}\backslash\{0_{K\times 1}\}$ based on a globally known PMF. In case $n=2$, 
 \begin{eqnarray}
\label{ }
\mathsf{SMG}(2)&=&\frac{2\Pr\{\vec{\boldsymbol{s}}_{1}\notin\mathrm{span}(\vec{\boldsymbol{s}}_{2})\}}{K}\notag\\
&=&\frac{2\Pr\left\{\textrm{$\vec{\boldsymbol{s}}_{1}$ and $\vec{\boldsymbol{s}}_{2}$ are not parallel in $\mathbb{R}^{K}$}\right\}}{K}.\end{eqnarray}
Taking $K=2$, let us assume that $\mathfrak{C}$ consists of $L$ vectors in $\mathbb{R}^{2}$ no two of which are parallel with each other. Therefore, 
\begin{equation}
\label{ }
\mathsf{SMG}(2)=1-\frac{1}{L}.
\end{equation}
Since $L$ can be arbitrarily large, the SMG of a network of two users is equal to $1$. In this case, it is easy to see that 
\begin{equation}
\label{ }
\mathsf{IEF}=\log L.
\end{equation}

If $n>2$, taking a set of arbitrarily large non-parallel vectors in some space $\mathbb{R}^{K}$ is by no means a necessarily appropriate collection. Let $\mathfrak{C}=\{\vec{\mathfrak{c}}_{1},\cdots,\vec{\mathfrak{c}}_{L}\}$ consist of $L$ vectors in $\mathbb{R}^{K}$.  For each $1\leq l\leq L$ and $1\leq r\leq L-1$, we denote by $\omega_{l,r}$ the number of distinct subsets $\mathfrak{B}$ of $\mathfrak{C}$ of size $r$ such that $\vec{\mathfrak{c}}_{l}\notin\mathrm{span}(\mathfrak{B})$. We denote these subsets explicitly by $\mathfrak{B}_{l,r}(1),\cdots,\mathfrak{B}_{l,r}(\omega_{l,r})$.  Assuming all users select their codes equally likely over $\mathfrak{C}$, 
\begin{eqnarray}
\label{ }
\Pr\left\{\vec{\boldsymbol{s}}_{1}\notin\mathrm{csp}\left([\vec{\boldsymbol{s}}_{2}|\vec{\boldsymbol{s}}_{3}|\cdots|\vec{\boldsymbol{s}}_{n-1}|\vec{\boldsymbol{s}}_{n}]\right)\right\}=\frac{1}{L}\sum_{l=1}^{L}\Pr\left\{\vec{\mathfrak{c}}_{l}\notin\mathrm{csp}\left([\vec{\boldsymbol{s}}_{2}|\vec{\boldsymbol{s}}_{3}|\cdots|\vec{\boldsymbol{s}}_{n-1}|\vec{\boldsymbol{s}}_{n}]\right)\right\}.\end{eqnarray}
For each $1\leq l\leq L$, 
\begin{eqnarray}
\Pr\left\{\vec{\mathfrak{c}}_{l}\notin\mathrm{csp}\left([\vec{\boldsymbol{s}}_{2}|\vec{\boldsymbol{s}}_{3}|\cdots|\vec{\boldsymbol{s}}_{n-1}|\vec{\boldsymbol{s}}_{n}]\right)\right\}=\sum_{r=1}^{L-1}\sum_{m=1}^{\omega_{l,r}}\Pr\left\{\forall\vec{\mathfrak{b}}\in\mathfrak{B}_{l,r}(m),\exists j\geq 2: \vec{\boldsymbol{s}}_{j}=\vec{\mathfrak{b}}\right\}.
\end{eqnarray}
It is easy to see that\footnote{Assuming the $1^{st},\cdots,(r-1)^{th}$ and $r^{th}$ elements of $\mathfrak{B}_{l,r}(m)$ are chosen by $t_{1},\cdots,t_{r-1}$ and $t_{r}$ users respectively, this can happen in $\sum_{\substack{t_{1}+\cdots+t_{r}=n-1\\t_{1},\cdots,t_{r}\geq 1}}\frac{(n-1)!}{t_{1}!\cdots t_{r}!}$ different ways.} \begin{equation}\Pr\left\{\forall\vec{\mathfrak{b}}\in\mathfrak{B}_{l,r}(m),\exists j\geq 2: \vec{\boldsymbol{s}}_{j}=\vec{\mathfrak{b}}\right\}=\frac{\sum_{\substack{t_{1}+\cdots+t_{r}=n-1\\t_{1},\cdots,t_{r}\geq 1}}\frac{(n-1)!}{t_{1}!\cdots t_{r}!}}{L^{n-1}}\end{equation} for any $1\leq l\leq L$ and $1\leq m\leq \omega_{l,r}$. Therefore, 
\begin{eqnarray}
\Pr\left\{\vec{\boldsymbol{s}}_{1}\notin\mathrm{csp}\left([\vec{\boldsymbol{s}}_{2}|\vec{\boldsymbol{s}}_{3}|\cdots|\vec{\boldsymbol{s}}_{n-1}|\vec{\boldsymbol{s}}_{n}]\right)\right\}=\frac{1}{L^{n}}\sum_{l=1}^{L}\sum_{r=1}^{L-1}\omega_{l,r}\rho_{r,n}
\end{eqnarray}
where
\begin{equation}
\label{ }
\rho_{r,n}\triangleq \sum_{\substack{t_{1}+\cdots+t_{r}=n-1\\t_{1},\cdots,t_{r}\geq 1}}\frac{(n-1)!}{t_{1}!\cdots t_{r}!}.\end{equation}
Finally, the achieved SMG is 
\begin{equation}
\label{bool}
\mathsf{SMG}(n)=\frac{n\sum_{l=1}^{L}\sum_{r=1}^{L-1}\omega_{l,r}\rho_{r,n}}{KL^{n }}.\end{equation}
We remark that there is no closed formula for $\rho_{r,n}$, however, one can use the recursion 
\begin{equation}
\label{ }
r^{n-1}=\rho_{r,n}+\sum_{r'=1}^{r-1}{r\choose r'}\rho_{r-r',n}
\end{equation}
to compute this quantity. By (\ref{bool}), one can easily see that $\mathsf{SMG}(n)$ is maximized if $\omega_{l,r}$ is as large as possible for each $1\leq l\leq L$ and $1\leq r\leq L-1$. We know that $\omega_{l,r}\leq {L-1\choose r}$. This upper bound is achieved if $\mathfrak{C}$ consists of $L\leq K$ independent vectors in $\mathbb{R}^{K}$. In this case,  
\begin{equation}
\label{ }
\mathsf{SMG}(n)=\frac{n\sum_{r=1}^{L-1}{L-1\choose r}\rho_{r,n}}{KL^{n-1}}.
\end{equation} 
It is not hard to see that $\sum_{r=1}^{L-1}{L-1\choose r}\rho_{r,n}=(L-1)^{n-1}$. Hence, 
\begin{equation}
\label{ }
\mathsf{SMG}(n)=\frac{n}{K}\left(1-\frac{1}{L}\right)^{n-1}.
\end{equation}
To get the largest SMG, one may let $L=K$ yielding
\begin{eqnarray}
\label{ }
\sup_{K\geq 1}\mathsf{SMG}(n)=\left(1-\frac{1}{n}\right)^{n-1}
\end{eqnarray}
which is the result obtained in example 1 via masking without spreading. 

 \section{Optimality Results}
We have already seen that applying masking on top of spreading can result in larger achievable rates due to increasing the attained multiplexing gain. However, our results so far are based on the achievable rate $\mathsf{C}_{i}^{(\mathrm{lb})}(\vec{h}_{i})$ which is only a lower bound on the \emph{capacity} of the $i^{th}$ user. In deriving $\mathsf{C}_{i}^{(\mathrm{lb})}(\vec{h}_{i})$, the PDF of the transmitted signals is taken to be complex Gaussian which is not necessarily optimal. As such, we have no optimality arguments so far. 

In this section, we question the optimality of masking without spreading. In fact, we are interested to see if at any SNR level, there is an optimal PDF such that generating the transmitted signals based on this PDF makes spreading unnecessary. For this purpose, we define the \emph{masking capacity} of a user as the largest achievable rate by this user assuming all users follow the masking strategy with no spreading applied. We also require \emph{fairness} conditions by which we imply that users generate their signals using the same PDF.  Fixing $\varepsilon\in(0,1]$, the masking capacity of the $i^{th}$ user is defined by 
\begin{equation}
\label{ }
\mathscr{MC}_{i}(\varepsilon;\gamma,(h_{j,i})_{j=1}^{n})\triangleq\sup_{\substack{\boldsymbol{x}_{1},\cdots,\boldsymbol{x}_{n}\sim\mathrm{i.i.d}\\\mathrm{E}\{|\boldsymbol{x}_{1}|^{2}\}\leq \gamma}}\mathrm{I}(\boldsymbol{x}_{i},\boldsymbol{\mathfrak{m}}_{i};\boldsymbol{y}_{i})\end{equation}
where
\begin{equation}
\label{ }
\boldsymbol{y}_{i}=\varepsilon^{-\frac{1}{2}} h_{i,i}\boldsymbol{\mathfrak{m}}_{i}\boldsymbol{x}_{i}+\varepsilon^{-\frac{1}{2}}\sum_{j\neq i}\beta h_{j,i}\boldsymbol{\mathfrak{m}}_{j}\boldsymbol{x}_{j}+\boldsymbol{z}_{i}
\end{equation}
in which $\boldsymbol{\mathfrak{m}}_{i}$ is the masking coefficient of the $i^{th}$ user which is a $\mathrm{Ber}(\varepsilon)$ random variable and $\boldsymbol{z}_{i}$ is the $\mathcal{CN}(0,1)$ ambient noise random variable.  The parameter $\varepsilon$ is designed based on maximizing a globally available utility function such as $\mathrm{E}\left\{\mathsf{C}_{i}^{(\mathrm{lb})}(\vec{\boldsymbol{h}}_{i})\right\}$ assuming $(h_{i,j})_{i,j=1}^{n}$ are realizations of $\mathrm{i.i.d.}$ random variables with a continuous PDF.

We focus on a decentralized network of $n=2$ users. We call the users as user \#1 and user \#2. According to the results in example 3 (scheme B), the decision rule to regulate $\varepsilon$ is 
\begin{eqnarray}
\label{pel}
\hat{\varepsilon}=\arg\max_{\varepsilon\in(0,1]}\mathrm{E}\left\{\varepsilon\log\left(1+\frac{2^{-\mathscr{H}(\varepsilon)}|\boldsymbol{h}_{1,1}|^{2}\gamma}{\varepsilon\left(1+\frac{|\boldsymbol{h}_{2,1}|^{2}\gamma}{\varepsilon}\right)^{\varepsilon}}\right)\right\}.
\end{eqnarray}  
The main result of the paper is the following. 
\begin{thm}
There exist $\alpha_{1}\in [0,\frac{1}{2})$ and $\alpha_{2}\in(\frac{1}{2},1]$ such that for any $h_{1,1},h_{2,1}\in\mathbb{C}$, it is possible to achieve rates larger than $\mathscr{MC}_{1}(\hat{\varepsilon};\gamma,h_{1,1},h_{2,1})$ for sufficiently large values of $\gamma$ where $\hat{\varepsilon}\in(\alpha_{1},\alpha_{2})$ is given in (\ref{pel}).
\end{thm}
To prove Theorem 1, we need the following Lemma. 
   \begin{lem}
   Let $\mathbf{Z}_{1}$ and $\mathbf{Z}_{2}$ be circularly symmetric complex Gaussian random variables with variances $\sigma_{1}^{2}$ and $\sigma_{2}^{2}$ respectively and $\mathbf{X}$ be independent of $(\mathbf{Z}_{1},\mathbf{Z}_{2})$. Then,  the answer to the optimization problem
   \begin{equation}
\label{ }
\sup_{\mathbf{X}:\mathrm{E}\{|\mathbf{X}|^{2}\}\leq P}\mathrm{h}(\mathbf{X}+\mathbf{Z}_{1})-\xi\mathrm{h}(\mathbf{X}+\mathbf{Z}_{2})
\end{equation}
is a circularly symmetric complex Gaussian $\mathbf{X}$ for any $P>0$ and any $\xi\geq 1$. Also, if $\sigma_{1}^{2}\leq \sigma_{2}^{2}$, the same conclusion holds for any $\xi\in\mathbb{R}$.   \end{lem}
   \begin{proof}
   This is a direct consequence of Theorem 1 in \cite{LV}.   \end{proof}
Our strategy is to find an upper bound on $\mathscr{MC}_{1}(\varepsilon;\gamma,h_{1,1},h_{2,1})$ for arbitrary $\varepsilon\in(0,1]$ and proposing an achievable rate which surpasses this upper bound.
\subsection{Upper Bound on $\mathscr{MC}_{1}(\varepsilon;\gamma,h_{1,1},h_{2,1})$}
We proceed as follows. We have
\begin{eqnarray}
\label{ee1}
\mathrm{I}(\boldsymbol{x}_{1},\boldsymbol{\mathfrak{m}}_{1};\boldsymbol{y}_{1})&=&\mathrm{I}(\boldsymbol{x}_{1};\boldsymbol{y}_{1}|\boldsymbol{\mathfrak{m}}_{1})+\mathrm{I}(\boldsymbol{\mathfrak{m}}_{1};\boldsymbol{y}_{1})\notag\\
&\leq &\mathrm{I}(\boldsymbol{x}_{1};\boldsymbol{y}_{1}|\boldsymbol{\mathfrak{m}}_{1})+\mathrm{H}(\boldsymbol{\mathfrak{m}}_{1})\notag\\
&=&\varepsilon\,\mathrm{I}(\boldsymbol{x}_{1};\varepsilon^{-\frac{1}{2}}h_{1,1}\boldsymbol{\mathfrak{m}}_{1}\boldsymbol{x}_{1}+h_{2,1}\boldsymbol{\mathfrak{m}}_{2}\boldsymbol{x}_{2}+\boldsymbol{z}_{1}|\boldsymbol{\mathfrak{m}}_{1}=1)\notag\\
&&+\bar{\varepsilon}\,\,\mathrm{I}(\boldsymbol{x}_{1};\varepsilon^{-\frac{1}{2}}h_{1,1}\boldsymbol{\mathfrak{m}}_{1}\boldsymbol{x}_{1}+\varepsilon^{-\frac{1}{2}}h_{2,1}\boldsymbol{\mathfrak{m}}_{2}\boldsymbol{x}_{2}+\boldsymbol{z}_{1}|\boldsymbol{\mathfrak{m}}_{1}=0)+\mathscr{H}(\varepsilon)\notag\\
&\stackrel{(a)}{=}&\varepsilon\mathrm{I}(\boldsymbol{x}_{1};\varepsilon^{-\frac{1}{2}}h_{1,1}\boldsymbol{x}_{1}+\varepsilon^{-\frac{1}{2}}h_{2,1}\boldsymbol{\mathfrak{m}}_{2}\boldsymbol{x}_{2}+\boldsymbol{z}_{1})+\mathscr{H}(\varepsilon)\notag\\
&\stackrel{(b)}{\leq}&\varepsilon\mathrm{I}(\boldsymbol{x}_{1};\varepsilon^{-\frac{1}{2}}h_{1,1}\boldsymbol{x}_{1}+\varepsilon^{-\frac{1}{2}}h_{2,1}\boldsymbol{\mathfrak{m}}_{2}\boldsymbol{x}_{2}+\boldsymbol{z}_{1}|\boldsymbol{\mathfrak{m}}_{2})+\mathscr{H}(\varepsilon)\notag\\
&=&\varepsilon\bar{\varepsilon}\,\mathrm{I}(\boldsymbol{x}_{1};\varepsilon^{-\frac{1}{2}}h_{1,1}\boldsymbol{x}_{2}+\boldsymbol{z}_{1})+\varepsilon^{2} \mathrm{I}(\boldsymbol{x}_{1};\varepsilon^{-\frac{1}{2}}h_{1,1}\boldsymbol{x}_{1}+\varepsilon^{-\frac{1}{2}}h_{2,1}\boldsymbol{x}_{2}+\boldsymbol{z}_{1})+\mathscr{H}(\varepsilon)\notag\\
&&=\varepsilon\bar{\varepsilon}\Big(\mathrm{h}(\varepsilon^{-\frac{1}{2}}h_{1,1}\boldsymbol{x}_{1}+\boldsymbol{z}_{1})-\mathrm{h}(\boldsymbol{z}_{1})\Big)\notag\\
&&+\varepsilon^{2}\Big(\mathrm{h}(\varepsilon^{-\frac{1}{2}}h_{1,1}\boldsymbol{x}_{1}+\varepsilon^{-\frac{1}{2}}h_{2,1}\boldsymbol{x}_{2}+\boldsymbol{z}_{1})-\mathrm{h}(\varepsilon^{-\frac{1}{2}}h_{2,1}\boldsymbol{x}_{2}+\boldsymbol{z}_{1})\Big)+\mathscr{H}(\varepsilon)\notag\\
&\stackrel{(c)}{=}&\varepsilon\bar{\varepsilon}\,\left(\mathrm{h}(\varepsilon^{-\frac{1}{2}}h_{1,1}\boldsymbol{x}_{1}+\boldsymbol{z}_{1})-\frac{\varepsilon}{\bar{\varepsilon}}\,\,\,\mathrm{h}(\varepsilon^{-\frac{1}{2}}h_{2,1}\boldsymbol{x}_{1}+\boldsymbol{z}_{1})\right)\notag\\
&&+\varepsilon^{2} \mathrm{h}(\varepsilon^{-\frac{1}{2}}h_{1,1}\boldsymbol{x}_{1}+\varepsilon^{-\frac{1}{2}}h_{2,1}\boldsymbol{x}_{2}+\boldsymbol{z}_{1})-\varepsilon\bar{\varepsilon}\log(\pi e)+\mathscr{H}(\varepsilon)\notag\\
&\stackrel{(d)}{=}&\varepsilon\bar{\varepsilon}\,\left(\mathrm{h}(\boldsymbol{x}_{1}+\boldsymbol{z}'_{1})-\frac{\varepsilon}{\bar{\varepsilon}}\,\,\,\mathrm{h}(\boldsymbol{x}_{1}+\boldsymbol{z}''_{1})\right)
+\varepsilon^{2} \mathrm{h}(\varepsilon^{-\frac{1}{2}}h_{1,1}\boldsymbol{x}_{1}+\varepsilon^{-\frac{1}{2}}h_{2,1}\boldsymbol{x}_{2}+\boldsymbol{z}_{1})\notag\\
&&+\varepsilon\bar{\varepsilon}\log (\varepsilon^{-1}|h_{1,1}|^{2})-\varepsilon^{2}\log (\varepsilon^{-1}|h_{2,1}|^{2})-\varepsilon\bar{\varepsilon}\log(\pi e)+\mathscr{H}(\varepsilon)\end{eqnarray}
where $(a)$ follows by the fact that $\mathrm{I}(\boldsymbol{x}_{1};h_{1,1}\boldsymbol{\mathfrak{m}}_{1}\boldsymbol{x}_{1}+h_{2,1}\boldsymbol{\mathfrak{m}}_{2}\boldsymbol{x}_{2}+\boldsymbol{z}_{1}|\boldsymbol{\mathfrak{m}}_{1}=0)=0$, $(b)$ is by the fact that the mutual information between the input and output of the channel increases if a ``genie'' provides the receiver side of user \#1  with $\boldsymbol{\mathfrak{m}}_{2}$, $(c)$ follows by the fact that $\boldsymbol{x}_{1}$ and $\boldsymbol{x}_{2}$ are identically distributed and the fact that $\mathrm{h}(\boldsymbol{z}_{1})=\log(\pi e)$ and finally $(d)$ follows by the fact that for any complex random variable $\mathbf{X}$ and $a\in\mathbb{C}$, we have $\mathrm{h}(a\mathbf{X})=\mathrm{h}(\mathbf{X})+\log |a|^{2}$. Also, we have $\boldsymbol{z}'_{1}\sim\mathcal{CN}\left(0,\frac{\varepsilon}{|h_{1,1}|^{2}}\right)$ and $\boldsymbol{z}''_{1}\sim\mathcal{CN}\left(0,\frac{\varepsilon}{|h_{2,1}|^{2}}\right)$ in the last equality in (\ref{ee1}). Denoting the upper bound in (\ref{ee1}) by $\mathsf{UB}$,
\begin{eqnarray}
\mathscr{MC}_{1}(\varepsilon;\gamma,h_{1,1},h_{2,1})&\leq&\sup_{\substack{\boldsymbol{x}_{1},\boldsymbol{x}_{2}\sim\mathrm{i.i.d}\\\mathrm{E}\{|\boldsymbol{x}_{1}|^{2}\}\leq \gamma}}\mathsf{UB}\notag\\
&\stackrel{}{\leq}&\varepsilon\bar{\varepsilon}\sup_{\substack{\boldsymbol{x}_{1},\boldsymbol{x}_{2}\sim\mathrm{i.i.d}\\\mathrm{E}\{|\boldsymbol{x}_{1}|^{2}\}\leq \gamma}}\left(\mathrm{h}(\boldsymbol{x}_{1}+\boldsymbol{z}'_{1})-\frac{\varepsilon}{\bar{\varepsilon}}\,\,\mathrm{h}(\boldsymbol{x}_{1}+\boldsymbol{z}''_{1})\right)\notag\\&&+\varepsilon^{2}\sup_{\substack{\boldsymbol{x}_{1},\boldsymbol{x}_{2}\sim\mathrm{i.i.d}\\\mathrm{E}\{|\boldsymbol{x}_{1}|^{2}\}\leq \gamma}}\mathrm{h}(\varepsilon^{-\frac{1}{2}}h_{1,1}\boldsymbol{x}_{1}+\varepsilon^{-\frac{1}{2}}h_{2,1}\boldsymbol{x}_{2}+\boldsymbol{z}_{1})\notag\\
\label{koo}
&&+\varepsilon\bar{\varepsilon}\log (\varepsilon^{-1}|h_{1,1}|^{2})-\varepsilon^{2}\log (\varepsilon^{-1}|h_{2,1}|^{2})-\varepsilon\bar{\varepsilon}\log(\pi e).\end{eqnarray}

It is trivial that
\begin{eqnarray}
\label{th2}
\sup_{\substack{\boldsymbol{x}_{1},\boldsymbol{x}_{2}\sim\mathrm{i.i.d}\\\mathrm{E}\{|\boldsymbol{x}_{1}|^{2}\}\leq \gamma}}\mathrm{h}(\varepsilon^{-\frac{1}{2}}h_{1,1}\boldsymbol{x}_{1}+\varepsilon^{-\frac{1}{2}}h_{2,1}\boldsymbol{x}_{2}+\boldsymbol{z}_{1})=\log\left(\pi e\varepsilon^{-1}\left(|h_{1,1}|^{2}+|h_{2,1}|^{2}\right)\gamma+1\right)
\end{eqnarray}
which follows by the maximum entropy Lemma\cite{53}. 

Applying Lemma 1, if $\frac{\varepsilon}{\bar{\varepsilon}}\geq 1$ or $|h_{1,1}|>|h_{2,1}|$, or equivalently, $\varepsilon\geq \frac{1}{2}$ or $|h_{1,1}|>|h_{2,1}|$, the answer to the optimization $\max_{\substack{\boldsymbol{x}_{1},\boldsymbol{x}_{2}\sim\mathrm{i.i.d}\\\mathrm{E}\{|\boldsymbol{x}_{1}|^{2}\}\leq \gamma}}\left(\mathrm{h}(\boldsymbol{x}_{1}+\boldsymbol{z}'_{1})-\frac{\varepsilon}{\bar{\varepsilon}}\,\,\mathrm{h}(\boldsymbol{x}_{1}+\boldsymbol{z}''_{1})\right)$ is a complex Gaussian $\boldsymbol{x}_{1}$. We note that the power of the optimum Gaussian signal $\boldsymbol{x}_{1}$ is not necessarily $\gamma$. Let the optimum $\boldsymbol{x}_{1}$ be a $\mathcal{N}(0,v)$ random variable. We distinguish the following cases. 

\textit{Case 1-} If $\varepsilon\geq\frac{1}{2}$ and $\frac{h_{1,1}}{h_{2,1}}<\left(\frac{\varepsilon}{\bar{\varepsilon}}\right)^{\frac{1}{2}}$, then $v=0$. 

\textit{Case 2-} If $\varepsilon>\frac{1}{2}$, $\frac{h_{1,1}}{h_{2,1}}>\left(\frac{\varepsilon}{\bar{\varepsilon}}\right)^{\frac{1}{2}}$ and $\gamma>\frac{\varepsilon^{2}\bar{\varepsilon}}{2\varepsilon-1}\left(\frac{1}{h_{2,1}^{2}}-\frac{\varepsilon}{\bar{\varepsilon}h_{1,1}^{2}}\right)$, then $v=\frac{\varepsilon\bar{\varepsilon}}{2\varepsilon-1}\left(\frac{1}{h_{2,1}^{2}}-\frac{\varepsilon}{\bar{\varepsilon}h_{1,1}^{2}}\right)$. 

\textit{Case 3-} If $\varepsilon\leq\frac{1}{2}$ and $\frac{h_{1,1}}{h_{2,1}}>1$, then $v=\frac{\gamma}{\varepsilon}$. 

Verification of these cases is a straightforward task which is omitted here for the sake of brevity.  Therefore, as far as $\varepsilon\geq \frac{1}{2}$, the term $\sup_{\substack{\boldsymbol{x}_{1},\boldsymbol{x}_{2}\sim\mathrm{i.i.d}\\\mathrm{E}\{|\boldsymbol{x}_{1}|^{2}\}\leq \gamma}}\left(\mathrm{h}(\boldsymbol{x}_{1}+\boldsymbol{z}'_{1})-\frac{\varepsilon}{\bar{\varepsilon}}\,\,\mathrm{h}(\boldsymbol{x}_{1}+\boldsymbol{z}''_{1})\right)$ saturates by increasing $\gamma$. Using this fact together with (\ref{koo}) and (\ref{th2}), 
\begin{equation}
\label{s222}
\mathscr{MC}_{1}(\varepsilon; h_{1,1},h_{2,1},\gamma)\lesssim \varepsilon^{2}\log\gamma.
\end{equation}
 as far as $\varepsilon\geq \frac{1}{2}$. On the other hand, if $\varepsilon<\frac{1}{2}$ and $\frac{h_{1,1}}{h_{2,1}}>1$, 
 \begin{equation}
 \sup_{\substack{\boldsymbol{x}_{1},\boldsymbol{x}_{2}\sim\mathrm{i.i.d}\\\mathrm{E}\{|\boldsymbol{x}_{1}|^{2}\}\leq \gamma}}\left(\mathrm{h}(\boldsymbol{x}_{1}+\boldsymbol{z}'_{1})-\frac{\varepsilon}{\bar{\varepsilon}}\,\,\mathrm{h}(\boldsymbol{x}_{1}+\boldsymbol{z}''_{1})\right)\sim \frac{\bar{\varepsilon}-\varepsilon}{\bar{\varepsilon}}\log\gamma.
 \end{equation}
Using this together with  (\ref{koo}) and (\ref{th2}),
\begin{equation}
\label{s444}
\mathscr{MC}_{1}(\varepsilon; h_{1,1},h_{2,1},\gamma)\lesssim \varepsilon\bar{\varepsilon}\log\gamma
\end{equation}
as far as $\varepsilon<\frac{1}{2}$ and $\frac{h_{1,1}}{h_{2,1}}<1$. However, we can remove the condition $\frac{h_{1,1}}{h_{2,1}}>1$ by a simple arguement. Let us fix $h_{2,1}$. It is clear that $\mathscr{MC}_{1}(\varepsilon; h,h_{2,1},\gamma)<\mathscr{MC}_{1}(\varepsilon; h',h_{2,1},\gamma)$ for $h<h_{2,1}<h'$. Since $\mathscr{MC}_{1}(\varepsilon; h',h_{2,1},\gamma)\lesssim \varepsilon\bar{\varepsilon}\log\gamma$, we get $\mathscr{MC}_{1}(\varepsilon;h,h_{2,1},\gamma)\lesssim \varepsilon\bar{\varepsilon}\log\gamma$. Hence, (\ref{s444}) holds for all $\varepsilon<\frac{1}{2}$ regardless of the values of $h_{1,1}$ and $h_{2,1}$.

To recap, we have shown that
\begin{equation}
\label{s555}
\mathscr{MC}_{1}(\varepsilon; h_{1,1},h_{2,1},\gamma)\lesssim\left\{\begin{array}{cc}
    \varepsilon^{2}\log\gamma  &     \varepsilon\geq \frac{1}{2}\\
    \varepsilon\bar{\varepsilon}\log\gamma  &   \varepsilon<\frac{1}{2}
\end{array}\right.\end{equation}
We end this subsection with the following Corollary.
\begin{coro}
If $\varepsilon\leq \frac{1}{2}$, 
\begin{equation}
\mathscr{MC}_{1}(\varepsilon; h_{1,1},h_{2,1},\gamma)\sim\varepsilon\bar{\varepsilon}\log\gamma.
\end{equation}
\end{coro}
\begin{proof}
By the results in example 3, $\mathscr{MC}_{1}(\varepsilon; h_{1,1},h_{2,1},\gamma)\gtrsim\varepsilon\bar{\varepsilon}\log\gamma$ for every $\varepsilon\in(0,1)$. However, by (\ref{s555}), $\mathscr{MC}_{1}(\varepsilon; h_{1,1},h_{2,1},\gamma)\lesssim\varepsilon\bar{\varepsilon}\log\gamma$ for all $\varepsilon\leq \frac{1}{2}$. This concludes the proof. \end{proof}

\subsection{Achieving Rates Larger Than $\mathscr{MC}_{1}(\varepsilon; h_{1,1},h_{2,1},\gamma)$}
 Applying spreading on top of masking, we show that there is a range of $\varepsilon$ such that it is possible to achieve rates larger than $\mathscr{MC}_{1}(\varepsilon; h_{1,1},h_{2,1},\gamma)$ as far as $\gamma$ is sufficiently large. To transmit its Gaussian signal $\boldsymbol{x}_{i}\sim\mathcal{CN}\left(0,\gamma\right)$, user \#$i$ spreads $\boldsymbol{x}_{i}$ along a $2\times 1$ random vector $\vec{\boldsymbol{\mathfrak{s}}}_{i}$ consisting of $\mathrm{i.i.d.}$ random numbers taking values in a finite alphabet $\mathscr{A}$ with equal probability. Thereafter, this user applies the masking process by constructing  the $2\times 1$ masking vector $\vec{\boldsymbol{\mathfrak{m}}}_{i}$ consisting of $\mathrm{i.i.d.}$ Bernoulli random variables taking the values $0$ and $1$ with probabilities $\bar{\varepsilon}$ and $\varepsilon$ respectively.  We assume that $\vec{\boldsymbol{\mathfrak{s}}}_{i}$ and $\vec{\boldsymbol{\mathfrak{m}}}_{i}$ are known to both ends of user \#$i$. Finally, this user transmits $\beta\boldsymbol{x}_{i}\vec{\boldsymbol{\mathfrak{m}}}_{i}\odot\vec{\boldsymbol{\mathfrak{s}}}_{i}$ in two consecutive transmission slots where  $\beta$ is to ensure the total transmission power per symbol $\boldsymbol{x}_{i}$ is $\gamma$. Assuming both users are synchronous, the following vector is received at the receiver side of user \#1
\begin{eqnarray}
\label{nb}
\vec{\boldsymbol{y}}_{1}=\beta h_{1,1}\boldsymbol{x}_{1}\vec{\boldsymbol{\mathfrak{m}}}_{1}\odot\vec{\boldsymbol{\mathfrak{s}}}_{1}+\beta h_{2,1}\boldsymbol{x}_{2}\vec{\boldsymbol{\mathfrak{m}}}_{2}\odot\vec{\boldsymbol{\mathfrak{s}}}_{2}+\vec{\boldsymbol{z}}_{1}
\end{eqnarray}
where $\vec{\boldsymbol{z}}_{1}$ is a vector of independent $\mathcal{CN}(0,1)$ random variables representing the ambient noise samples at the receiver side of user \#1. The achievable rate for this user is
\begin{equation}
\label{ }
R_{1}\triangleq\frac{\mathrm{I}\left(\boldsymbol{x}_{1};\vec{\boldsymbol{y}}_{1}|\vec{\boldsymbol{\mathfrak{s}}}_{1},\vec{\boldsymbol{\mathfrak{m}}}_{1}\right)}{2}.
\end{equation}
By our results in section III, 
\begin{equation}
\label{ }
\mathrm{I}\left(\boldsymbol{x}_{1};\vec{\boldsymbol{y}}_{1}|\vec{\boldsymbol{\mathfrak{s}}}_{1},\vec{\boldsymbol{\mathfrak{m}}}_{1}\right)\sim\Pr\Big\{\vec{\boldsymbol{\mathfrak{m}}}_{1}\odot\vec{\boldsymbol{\mathfrak{s}}}_{1}\notin\mathrm{span}(\vec{\boldsymbol{\mathfrak{m}}}_{2}\odot\vec{\boldsymbol{\mathfrak{s}}}_{2})\Big\}\log\gamma.
\end{equation}
Hence,
\begin{equation}
\label{e33}
R_{1}\sim\frac{1}{2}\Pr\Big\{\vec{\boldsymbol{\mathfrak{m}}}_{1}\odot\vec{\boldsymbol{\mathfrak{s}}}_{1}\notin\mathrm{span}(\vec{\boldsymbol{\mathfrak{m}}}_{2}\odot\vec{\boldsymbol{\mathfrak{s}}}_{2})\Big\}\log\gamma.
\end{equation}
We are interested in values of $\varepsilon$ so that $\mathscr{MC}_{1}(\varepsilon; h_{1,1},h_{2,1},\gamma)\lesssim R_{1}$ is strictly satisfied. By (\ref{s555}) and (\ref{e33}), it is sufficient to show that there is a range for $\varepsilon$ such that
\begin{equation}
\label{x2}
\frac{1}{2}\Pr\Big\{\vec{\boldsymbol{\mathfrak{m}}}_{1}\odot\vec{\boldsymbol{\mathfrak{s}}}_{1}\notin\mathrm{span}(\vec{\boldsymbol{\mathfrak{m}}}_{2}\odot\vec{\boldsymbol{\mathfrak{s}}}_{2})\Big\}>\max\{\varepsilon^{2},\varepsilon\bar{\varepsilon}\}.
\end{equation}

Let $\mathscr{A}=\{-1,1\}$. In this case the elements of $\vec{\boldsymbol{\mathfrak{m}}}_{1}\odot\vec{\boldsymbol{\mathfrak{s}}}_{1}$ and $\vec{\boldsymbol{\mathfrak{m}}}_{2}\odot\vec{\boldsymbol{\mathfrak{s}}}_{2}$ are $\mathrm{i.i.d.}$ random variables taking the values $0$, $1$ and $-1$ with probabilities $\bar{\varepsilon}$, $\frac{\varepsilon}{2}$ and $\frac{\varepsilon }{2}$ respectively. The event $\vec{\boldsymbol{\mathfrak{m}}}_{1}\odot\vec{\boldsymbol{\mathfrak{s}}}_{1}\in\mathrm{span}(\vec{\boldsymbol{\mathfrak{m}}}_{2}\odot\vec{\boldsymbol{\mathfrak{s}}}_{2})$ occurs if and only if $\vec{\boldsymbol{\mathfrak{m}}}_{1}\odot\vec{\boldsymbol{\mathfrak{s}}}_{1}=0_{2\times 1}$ or $\vec{\boldsymbol{\mathfrak{m}}}_{1}\odot\vec{\boldsymbol{\mathfrak{s}}}_{1}\neq0_{2\times 1}$ while $\vec{\boldsymbol{\mathfrak{m}}}_{1}\odot\vec{\boldsymbol{\mathfrak{s}}}_{1}=\pm \vec{\boldsymbol{\mathfrak{m}}}_{2}\odot\vec{\boldsymbol{\mathfrak{s}}}_{2}$. Then, one can easily see that $\Pr\Big\{\vec{\boldsymbol{\mathfrak{m}}}_{1}\odot\vec{\boldsymbol{\mathfrak{s}}}_{1}\notin\mathrm{span}(\vec{\boldsymbol{\mathfrak{m}}}_{2}\odot\vec{\boldsymbol{\mathfrak{s}}}_{2})\Big\}=1-\bar{\varepsilon}^{2}-2(\varepsilon\bar{\varepsilon})^{2}-\frac{\varepsilon^{4}}{2}$. This can also be deduced from  (\ref{boro}). Substituting this in (\ref{x2}) requires
\begin{equation}
1-\bar{\varepsilon}^{2}-2(\varepsilon\bar{\varepsilon})^{2}-\frac{\varepsilon^{4}}{2}>2\max\{\varepsilon^{2},\varepsilon\bar{\varepsilon}\}.
\end{equation}
 This simplifies to $5\varepsilon^{2}-8\varepsilon+2<0$ for $\varepsilon<\frac{1}{2}$ and $5\varepsilon^{3}-8\varepsilon^{2}+10\varepsilon-4<0$ for $\varepsilon\geq \frac{1}{2}$. Solving these inequalities, we get $\varepsilon\in(0.3101,0.5653)$. 
 
 It is not hard to see that $\hat{\varepsilon}$ given in (\ref{pel}) is in the interval $(0.4,0.5)$ for all $\gamma>30\mathrm{dB}$. Setting  $\alpha_{1}=0.3101$ and $\alpha_{2}=0.5653$, we see that $\hat{\varepsilon}\in(\alpha_{1},\alpha_{2})$ and $R_{1}$ is larger than $\mathscr{MC}_{1}(\hat{\varepsilon}; h_{1,1},h_{2,1},\gamma)$ for large values of $\gamma$. This completes the proof of Theorem 1.
 
Next, we demonstrate that increasing the size of the underlying alphabet can expand the range of $\varepsilon$ for which achieving a rate larger than $\mathscr{MC}_{1}(\varepsilon; h_{1,1},h_{2,1},\gamma)$ is possible.

\textit{Remark 3-} If $\mathscr{A}=\{-2,-1,1,2\}$, the elements of $\vec{\boldsymbol{\mathfrak{m}}}_{1}\odot\vec{\boldsymbol{\mathfrak{s}}}_{1}$ and $\vec{\boldsymbol{\mathfrak{m}}}_{2}\odot\vec{\boldsymbol{\mathfrak{s}}}_{2}$ are $\mathrm{i.i.d.}$ random variables taking the values $0$, $-2$, $-1$, $1$ and $2$ with probabilities $\bar{\varepsilon}$, $\frac{\varepsilon}{4}$, $\frac{\varepsilon }{4}$, $\frac{\varepsilon}{4}$ and $\frac{\varepsilon}{4}$ respectively. The event $\vec{\boldsymbol{\mathfrak{m}}}_{1}\odot\vec{\boldsymbol{\mathfrak{s}}}_{1}\in\mathrm{span}(\vec{\boldsymbol{\mathfrak{m}}}_{2}\odot\vec{\boldsymbol{\mathfrak{s}}}_{2})$ occurs if and only if $\vec{\boldsymbol{\mathfrak{m}}}_{1}\odot\vec{\boldsymbol{\mathfrak{s}}}_{1}=0_{2\times 1}$ or $\vec{\boldsymbol{\mathfrak{m}}}_{1}\odot\vec{\boldsymbol{\mathfrak{s}}}_{1}\neq0_{2\times 1}$ while $\vec{\boldsymbol{\mathfrak{m}}}_{1}\odot\vec{\boldsymbol{\mathfrak{s}}}_{1}=\pm \vec{\boldsymbol{\mathfrak{m}}}_{2}\odot\vec{\boldsymbol{\mathfrak{s}}}_{2}$ or $\vec{\boldsymbol{\mathfrak{m}}}_{1}\odot\vec{\boldsymbol{\mathfrak{s}}}_{1}=\pm2 \vec{\boldsymbol{\mathfrak{m}}}_{2}\odot\vec{\boldsymbol{\mathfrak{s}}}_{2}$ or $\vec{\boldsymbol{\mathfrak{m}}}_{1}\odot\vec{\boldsymbol{\mathfrak{s}}}_{1}=\pm\frac{1}{2} \vec{\boldsymbol{\mathfrak{m}}}_{2}\odot\vec{\boldsymbol{\mathfrak{s}}}_{2}$. We get $\Pr\Big\{\vec{\boldsymbol{\mathfrak{m}}}_{1}\odot\vec{\boldsymbol{\mathfrak{s}}}_{1}\notin\mathrm{span}(\vec{\boldsymbol{\mathfrak{m}}}_{2}\odot\vec{\boldsymbol{\mathfrak{s}}}_{2})\Big\}=1-\bar{\varepsilon}^{2}-2(\varepsilon\bar{\varepsilon})^{2}-\frac{3\varepsilon^{4}}{16}$.
Substituting this in (\ref{x2}) requires
\begin{equation}
1-\bar{\varepsilon}^{2}-2(\varepsilon\bar{\varepsilon})^{2}-\frac{3\varepsilon^{4}}{16}>2\max\{\varepsilon^{2},\varepsilon\bar{\varepsilon}\}.
\end{equation}
Hence, $35\varepsilon^{2}-64\varepsilon+16<0$ for $\varepsilon<\frac{1}{2}$ and $35\varepsilon^{3}-64\varepsilon^{2}+80\varepsilon-32<0$ for $\varepsilon\geq \frac{1}{2}$. Solving these inequalities, $\varepsilon\in(0.2988,0.5873)$. $\square$

\section{Conclusion}
We proposed an approach towards communication in decentralized wireless networks of separate transmitter-receiver pairs. A randomized signaling scheme was introduced in which each user locally spreads its Gaussian signal along a  randomly generated spreading code comprised of a sequence of nonzero elements over a certain alphabet. Along with spreading, each transmitter also masks its output independently from transmission to transmission. Using a conditional version of entropy power inequality and a key lemma on the differential entropy of mixed Gaussian random vectors, achievable rates were developed for the users. Assuming the channel gains are realization of independent continuous random variables, each user finds the optimum parameters in constructing the randomized spreading and masking sequences by maximizing the average achievable rate per user. It was seen that  as the number of users increases, the achievable Sum Multiplexing Gain of the network approaches that of a centralized orthogonal scheme where multiuser interference is completely avoided.  It was observed that in general the elements of a spreading code are not equiprobable over the underlying alphabet. This particularly happens if the number of active users is greater than three. Finally, using the recently developed extremal inequality of Liu-Viswanath, we presented an optimality result showing that transmission of Gaussian signals via spreading and masking yields higher achievable rates than the maximum achievable rate attained by applying masking only.  

\section*{Appendix A}
By Proposition 1,
\begin{equation}
\label{ }
\lim_{\gamma\to\infty}\frac{\mathsf{C}_{i}}{\log\gamma}\geq \frac{\Pr\left\{\vec{\boldsymbol{s}}_{i}\notin\mathrm{csp}(\boldsymbol{S}_{i})\right\}}{K}.\end{equation} 
In this appendix, we  prove that 
\begin{equation}
\label{ }
\lim_{\gamma\to\infty}\frac{\mathsf{C}_{i}}{\log\gamma}\leq \frac{\Pr\left\{\vec{\boldsymbol{s}}_{i}\notin\mathrm{csp}(\boldsymbol{S}_{i})\right\}}{K}.
\end{equation} 
By (\ref{e14}), it suffices to show that $\lim_{\gamma\to\infty}\frac{\mathrm{I}(\vec{\boldsymbol{x}}_{i};\vec{\boldsymbol{y}}_{i}|\vec{\boldsymbol{s}}_{i})}{\log\gamma}\leq \Pr\left\{\vec{\boldsymbol{s}}_{i}\notin\mathrm{csp}(\boldsymbol{S}_{i})\right\}$. Let us consider the \emph{informed} $i^{th}$ user where the receiver is aware of $\vec{\boldsymbol{s}}_{i}$ and $\boldsymbol{S}_{i}$. The achievable rate of this virtual user is $\frac{\mathrm{I}(\vec{\boldsymbol{x}}_{i};\vec{\boldsymbol{y}}_{i}|\boldsymbol{s}_{i},\boldsymbol{S}_{i})}{K}$. It is clear that $\mathrm{I}(\vec{\boldsymbol{x}}_{i};\vec{\boldsymbol{y}}_{i}|\vec{\boldsymbol{s}}_{i})\leq \mathrm{I}(\vec{\boldsymbol{x}}_{i};\vec{\boldsymbol{y}}_{i}|\vec{\boldsymbol{s}}_{i},\boldsymbol{S}_{i})$. However, 
\begin{eqnarray}
\label{bol12}
\mathrm{I}(\vec{\boldsymbol{x}}_{i};\vec{\boldsymbol{y}}_{i}|\vec{\boldsymbol{s}}_{i},\boldsymbol{S}_{i})&=&\sum_{\substack{\vec{s}\in\mathrm{supp}(\vec{\boldsymbol{s}}_{i})\\S\in\mathrm{range}(\boldsymbol{S}_{i})}}\Pr\{\vec{\boldsymbol{s}}_{i}=\vec{s}\}\Pr\{\boldsymbol{S}_{i}=S\}\mathrm{I}(\vec{\boldsymbol{x}}_{i};\vec{\boldsymbol{y}}_{i}|\vec{\boldsymbol{s}}_{i}=\vec{s},\boldsymbol{S}_{i}=S)\notag\\
&\stackrel{(a)}{=}&\sum_{\substack{\vec{s}\in\mathrm{supp}(\vec{\boldsymbol{s}}_{i})\\S\in\mathrm{range}(\boldsymbol{S}_{i})}}\Pr\{\vec{\boldsymbol{s}}_{i}=\vec{s}\}\Pr\{\boldsymbol{S}_{i}=S\}\log\frac{\det\left(\mathrm{cov}\left(\vec{\boldsymbol{y}}_{i}|\vec{\boldsymbol{s}}_{i}=\vec{s},\boldsymbol{S}_{i}=S\right)\right)}{\det\left(\mathrm{cov}\left(\vec{\boldsymbol{w}}_{i}+\vec{\boldsymbol{z}}_{i}|\boldsymbol{S}_{i}=S\right)\right)}\notag\\
&=&\sum_{\substack{\vec{s}\in\mathrm{supp}(\vec{\boldsymbol{s}}_{i})\\S\in\mathrm{range}(\boldsymbol{S}_{i})}}\Pr\{\vec{\boldsymbol{s}}_{i}=\vec{s}\}\Pr\{\boldsymbol{S}_{i}=S\}\log\det\left(\mathrm{cov}\left(\vec{\boldsymbol{y}}_{i}|\vec{\boldsymbol{s}}_{i}=\vec{s},\boldsymbol{S}_{i}=S \right)\right)\notag\\
&&-\sum_{S\in\mathrm{range}(\boldsymbol{S}_{i})}\Pr\{\boldsymbol{S}_{i}=S\}\log\det\left(\mathrm{cov}\left(\vec{\boldsymbol{w}}_{i}+\vec{\boldsymbol{z}}_{i}|\boldsymbol{S}_{i}=S\right)\right)\end{eqnarray}
where $(a)$ follows by the fact that fixing $\boldsymbol{S}_{i}=S$ converts the channel of the $i^{th}$ informed user to an additive Gaussian channel. On the other hand, 
\begin{eqnarray}
\label{bol11}
&&\sum_{\substack{\vec{s}\in\mathrm{supp}(\vec{\boldsymbol{s}}_{i})\\S\in\mathrm{range}(\boldsymbol{S}_{i})}}\Pr\{\vec{\boldsymbol{s}}_{i}=\vec{s}\}\Pr\{\boldsymbol{S}_{i}=S\}\log\det\left(\mathrm{cov}\left(\vec{\boldsymbol{y}}_{i}|\vec{\boldsymbol{s}}_{i}=\vec{s},\boldsymbol{S}_{i}=S \right)\right)\notag\\
&=&\sum_{\substack{\vec{s}\in\mathrm{supp}(\vec{\boldsymbol{s}}_{i})\\S\in\mathrm{range}(\boldsymbol{S}_{i})}}\Pr\{\vec{\boldsymbol{s}}_{i}=\vec{s}\}\Pr\{\boldsymbol{S}_{i}=S\}\log\det\left(I_{K}+\beta^{2}\gamma|h_{i,i}|^{2}\vec{s}\vec{s}^{T}+\beta^{2}\gamma S\Xi_{i}\Xi_{i}^{T}S^{T}\right).\notag\\
\end{eqnarray}
Noting that $\log\det\left(I_{K}+\beta^{2}\gamma|h_{i,i}|^{2}\vec{s}\vec{s}^{T}+\beta^{2}\gamma S\Xi_{i}\Xi_{i}^{T}S^{T}\right)$ scales like $\mathrm{rank}\left(\left[\vec{s}\,\,\,S\right]\right)\log\gamma$, we conclude that the first term on the right hand side of (\ref{bol12}) scales like $\mathrm{E}\left\{\mathrm{rank}\left(\left[\vec{\boldsymbol{s}}_{i}\,\,\,\boldsymbol{S}_{i}\right]\right)\right\}\log\gamma$. By the same token, the second term on the right hand side of (\ref{bol12}) scales like $\mathrm{E}\left\{\mathrm{rank}(\boldsymbol{S}_{i})\right\}\log\gamma$. Therefore, $\mathrm{I}(\vec{\boldsymbol{x}}_{i};\vec{\boldsymbol{y}}_{i}|\vec{\boldsymbol{s}}_{i})$ is upper bounded by a quantity which scales like  $\Big(\mathrm{E}\{\mathrm{rank}\left([\vec{\boldsymbol{s}}_{i}\,\,\,\boldsymbol{S}_{i}\right])\}-\mathrm{E}\{\mathrm{rank}(\boldsymbol{S}_{i})\}\Big)\log\gamma$. The result of the Proposition is immediate.
  \section*{Appendix B}
Let $\vec{\boldsymbol{s}}_{1}=\begin{pmatrix}
    \boldsymbol{s}_{1,1}  & \boldsymbol{s}_{1,2}
\end{pmatrix}^{T}$. Therefore, 
\begin{eqnarray}
\label{lol1}
\mathrm{H}\left(\vec{\boldsymbol{s}}_{1}\vec{\boldsymbol{s}}_{1}^{\dagger}\right)&=&\mathrm{H}\left(|\boldsymbol{s}_{1,1}|^{2},|\boldsymbol{s}_{1,2}|^{2},\boldsymbol{s}_{1,1}\boldsymbol{s}_{1,2}\right)\notag\\
&=&\mathrm{H}\left(|\boldsymbol{s}_{1,1}|^{2},|\boldsymbol{s}_{1,2}|^{2}\right)+\mathrm{H}\left(\boldsymbol{s}_{1,1}\boldsymbol{s}_{1,2}\big||\boldsymbol{s}_{1,1}|^{2},|\boldsymbol{s}_{1,2}|^{2}\right)\notag\\
&=&\mathrm{H}\left(|\boldsymbol{s}_{1,1}|^{2}\right)+\mathrm{H}\left(|\boldsymbol{s}_{1,2}|^{2}\right)+\mathrm{H}\left(\boldsymbol{s}_{1,1}\boldsymbol{s}_{1,2}\big||\boldsymbol{s}_{1,1}|^{2},|\boldsymbol{s}_{1,2}|^{2}\right)\notag\\
&=&\mathrm{H}\left(|\boldsymbol{s}_{1,1}|\right)+\mathrm{H}\left(|\boldsymbol{s}_{1,2}|\right)+\mathrm{H}\left(\boldsymbol{s}_{1,1}\boldsymbol{s}_{1,2}\big||\boldsymbol{s}_{1,1}|,|\boldsymbol{s}_{1,2}|\right)\notag\\
&=&2\mathscr{H}(\varepsilon)+\mathrm{H}\left(\boldsymbol{s}_{1,1}\boldsymbol{s}_{1,2}\big||\boldsymbol{s}_{1,1}|,|\boldsymbol{s}_{1,2}|\right).\end{eqnarray}
To compute $\mathrm{H}\left(\boldsymbol{s}_{1,1}\boldsymbol{s}_{1,2}\big||\boldsymbol{s}_{1,1}|,|\boldsymbol{s}_{1,2}|\right)$, we have
\begin{eqnarray}
\mathrm{H}\left(\boldsymbol{s}_{1,1}\boldsymbol{s}_{1,2}\big||\boldsymbol{s}_{1,1}|,|\boldsymbol{s}_{1,2}|\right)&=&\mathrm{H}\left(\boldsymbol{s}_{1,1}\boldsymbol{s}_{1,2}\big||\boldsymbol{s}_{1,1}|=1,|\boldsymbol{s}_{1,2}|=1\right)\Pr\{|\boldsymbol{s}_{1,1}|=1,|\boldsymbol{s}_{1,2}|=1\}\notag\\
&&+\mathrm{H}\left(\boldsymbol{s}_{1,1}\boldsymbol{s}_{1,2}\big||\boldsymbol{s}_{1,1}|=0,|\boldsymbol{s}_{1,2}|=1\right)\Pr\{|\boldsymbol{s}_{1,1}|=0,|\boldsymbol{s}_{1,2}|=1\}\notag\\
&&+\mathrm{H}\left(\boldsymbol{s}_{1,1}\boldsymbol{s}_{1,2}\big||\boldsymbol{s}_{1,1}|=1,|\boldsymbol{s}_{1,2}|=0\right)\Pr\{|\boldsymbol{s}_{1,1}|=1,|\boldsymbol{s}_{1,2}|=0\}\notag\\
&&+\mathrm{H}\left(\boldsymbol{s}_{1,1}\boldsymbol{s}_{1,2}\big||\boldsymbol{s}_{1,1}|=0,|\boldsymbol{s}_{1,2}|=0\right)\Pr\{|\boldsymbol{s}_{1,1}|=0,|\boldsymbol{s}_{1,2}|=0\}\notag\\
&\stackrel{(a)}{=}&\mathrm{H}\left(\boldsymbol{s}_{1,1}\boldsymbol{s}_{1,2}\big||\boldsymbol{s}_{1,1}|=1,|\boldsymbol{s}_{1,2}|=1\right)\Pr\{|\boldsymbol{s}_{1,1}|=1,|\boldsymbol{s}_{1,2}|=1\}\end{eqnarray}
where $(a)$ is by the fact that the terms $\mathrm{H}\left(\boldsymbol{s}_{1,1}\boldsymbol{s}_{1,2}\big||\boldsymbol{s}_{1,1}|=0,|\boldsymbol{s}_{1,2}|=1\right)$, $\mathrm{H}\left(\boldsymbol{s}_{1,1}\boldsymbol{s}_{1,2}\big||\boldsymbol{s}_{1,1}|=1,|\boldsymbol{s}_{1,2}|=0\right)$ and $\mathrm{H}\left(\boldsymbol{s}_{1,1}\boldsymbol{s}_{1,2}\big||\boldsymbol{s}_{1,1}|=0,|\boldsymbol{s}_{1,2}|=0\right)$ are zero. On the other hand, it is easy to see that $\Pr\{\boldsymbol{s}_{1,1}\boldsymbol{s}_{1,2}=1\big||\boldsymbol{s}_{1,1}|=1,|\boldsymbol{s}_{1,2}|=1\}=\nu^{2}+\overline{\nu}^{2}$. This implies $\mathrm{H}\left(\boldsymbol{s}_{1,1}\boldsymbol{s}_{1,2}\big||\boldsymbol{s}_{1,1}|=1,|\boldsymbol{s}_{1,2}|=1\right)=\mathscr{H}(\nu^{2}+\overline{\nu}^{2})$. Therefore, 
\begin{eqnarray}
\label{lol2}
\mathrm{H}\left(\boldsymbol{s}_{1,1}\boldsymbol{s}_{1,2}\big||\boldsymbol{s}_{1,1}|,|\boldsymbol{s}_{1,2}|\right)=\mathscr{H}(\nu^{2}+\overline{\nu}^{2})\Pr\{|\boldsymbol{s}_{1,1}|=1,|\boldsymbol{s}_{1,2}|=1\}
=\varepsilon^{2}\mathscr{H}(\nu^{2}+\overline{\nu}^{2}).\end{eqnarray}
Using (\ref{lol1}) and (\ref{lol2}), 
\begin{equation}
\label{ }
\mathrm{H}\left(\vec{\boldsymbol{s}}_{1}\vec{\boldsymbol{s}}_{1}^{\dagger}\right)=2\mathscr{H}(\varepsilon)+\varepsilon^{2}\mathscr{H}(\nu^{2}+\overline{\nu}^{2}).
\end{equation}
\section*{Appendix B}
Let $\vec{\boldsymbol{s}}_{1}=\begin{pmatrix}
    \boldsymbol{s}_{1,1}&\cdots&\boldsymbol{s}_{1,K}
\end{pmatrix}$. We have
\begin{eqnarray}
\label{nol}
\mathrm{H}(\vec{\boldsymbol{s}}_{1}\vec{\boldsymbol{s}}_{1}^{\dagger})&=&\mathrm{H}((\boldsymbol{s}_{1,k}\boldsymbol{s}_{1,l})_{k,l=1}^{K})\notag\\
&\stackrel{(a)}{=}&\mathrm{H}\left(\big(\boldsymbol{s}_{1,k}\boldsymbol{s}_{1,l}\big)_{\substack{k,l=1\\k\neq l}}^{K}\right)\notag\\
&\stackrel{(b)}{=}&\mathrm{H}\left(\boldsymbol{s}_{1,1}\boldsymbol{s}_{1,2},\boldsymbol{s}_{1,1}\boldsymbol{s}_{1,3},\cdots,\boldsymbol{s}_{1,1}\boldsymbol{s}_{1,K}\right)\notag\\\end{eqnarray}
where $(a)$ is by the fact that $\boldsymbol{s}_{1,k}^{2}=1$ for any $1\leq k\leq K$ and $(b)$ is by the fact that for any two distinct numbers $k,l\in\{2,\cdots,K\}$, the knowledge about $\boldsymbol{s}_{1,k}\boldsymbol{s}_{1,l}$ can be obtained by knowing $\boldsymbol{s}_{1,1}\boldsymbol{s}_{1,k}$ and $\boldsymbol{s}_{1,1}\boldsymbol{s}_{1,l}$ by the fact that $\boldsymbol{s}_{1,1}\boldsymbol{s}_{1,k}\boldsymbol{s}_{1,1}\boldsymbol{s}_{1,l}=\boldsymbol{s}_{1,k}\boldsymbol{s}_{1,l}\boldsymbol{s}_{1,1}^{2}=\boldsymbol{s}_{1,k}\boldsymbol{s}_{1,l}$. Let us define 
\begin{equation}
\label{ }
\widetilde{\vec{\boldsymbol{s}}}_{1}=\begin{pmatrix}
     \boldsymbol{s}_{1,2} &\cdots&\boldsymbol{s}_{1,K}
\end{pmatrix}^{T}.
\end{equation} 
By (\ref{nol}),  $\mathrm{H}(\vec{\boldsymbol{s}}_{1}\vec{\boldsymbol{s}}_{1}^{\dagger})=\mathrm{H}(\boldsymbol{s}_{1,1}\widetilde{\vec{\boldsymbol{s}}}_{1})$. Let $\mathcal{E}$ be the event where $\boldsymbol{s}_{1,1}=1$, while $k$ of the elements of $\widetilde{\vec{\boldsymbol{s}}}_{1}$, namely, $\boldsymbol{s}_{1,l_{1}},\cdots,\boldsymbol{s}_{1,l_{k-1}}$ and $\boldsymbol{s}_{1,l_{k}}$ are $1$ and the rest are $-1$ for some $0\leq k\leq K$ and $2\leq l_{1}<l_{2}<\cdots<l_{k}\leq K$. Also, let $\mathcal{F}$ be the event where $\boldsymbol{s}_{1,1}=-1$, $\boldsymbol{s}_{1,l}=-1$ for $l\in\{l_{1},l_{2},\cdots,l_{k}\}$ and $\boldsymbol{s}_{1,l}=1$ for $l\notin\{l_{1},l_{2},\cdots,l_{k}\}$. It is clear that 
\begin{equation}
\label{mol}
\boldsymbol{s}_{1,1}\widetilde{\vec{\boldsymbol{s}}}_{1}\mathbb{1}_{\mathcal{E}}=\boldsymbol{s}_{1,1}\widetilde{\vec{\boldsymbol{s}}}_{1}\mathbb{1}_{\mathcal{F}}.
\end{equation} 
We know that $\Pr\{\mathcal{E}\}=\nu^{k+1}\overline{\nu}^{K-k}$ and $\Pr\{\mathcal{F}\}=\nu^{K-k}\overline{\nu}^{k+1}$. Hence, using (\ref{mol}),
\begin{equation}
\label{ }
\mathrm{H}(\vec{\boldsymbol{s}}_{1}\vec{\boldsymbol{s}}_{1}^{\dagger})=-\sum_{k=0}^{K}{K\choose k}\left(\nu^{k+1}\overline{\nu}^{K-k}+\nu^{K-k}\overline{\nu}^{k+1}\right)\log\left(\nu^{k+1}\overline{\nu}^{K-k}+\nu^{K-k}\overline{\nu}^{k+1}\right).
\end{equation}  
   \bibliographystyle{IEEEbib}

 \end{document}